%% file: secure_2way_relay_IT_R1final.tex
\documentclass[journal,draftcls,onecolumn]{IEEEtran}

\usepackage{basicreq}
\def\markchange{}

\title{Secure Compute-and-Forward \\ in a Bidirectional Relay}
\author{Shashank Vatedka,~\IEEEmembership{Student~Member,~IEEE,}
Navin~Kashyap,~\IEEEmembership{Senior~Member,~IEEE,}
~Andrew~Thangaraj,~\IEEEmembership{Senior~Member,~IEEE}

\thanks{S.~Vatedka and N.~Kashyap (\{shashank,nkashyap\}@ece.iisc.ernet.in) are with the Department of Electrical Communication Engineering, Indian Institute of Science, Bangalore, India. }%
\thanks{A.~Thangaraj (andrew@ee.iitm.ac.in) is with the Department of Electrical Engineering, Indian Institute of Technology, Madras, India.}
\thanks{This work was presented in part at ISIT 2012, Cambridge, Mass., USA, and at ISIT 2013, Istanbul, Turkey.}}

\begin{document}
\maketitle

\begin{abstract}
\markchange{
  We consider the basic bidirectional relaying problem, in which two users in a wireless network wish to exchange messages through an intermediate relay node. In the compute-and-forward strategy, the relay computes a function of the two messages using the naturally-occurring sum of symbols simultaneously transmitted by user nodes in a Gaussian
multiple access (MAC) channel, and the computed function value is forwarded to the user nodes in an ensuing broadcast phase. In this paper, we study the problem under an additional security constraint, which requires that each user's message be kept secure from the relay. We consider two types of security constraints: perfect secrecy, in which the MAC channel output seen by the relay is independent of each user's message; and strong secrecy, which is a form of asymptotic independence. We propose a coding scheme based on nested lattices, the main feature of which is that given a pair of nested lattices that satisfy certain ``goodness'' properties, we can explicitly specify probability distributions for randomization at the encoders to achieve the desired security criteria. In particular, our coding scheme guarantees perfect or strong secrecy even in the absence of channel noise. The noise in the channel only affects reliability of computation at the relay, and for Gaussian noise, we derive achievable rates for reliable and secure computation. We also present an application of our methods to the multi-hop line network in which a source needs to transmit messages to a destination through a series of intermediate relays.
}
\end{abstract}
\section{Introduction}\label{sec:intro}

Consider a network having three nodes, denoted by $\tA$, $\tB$ and $\tR$\markchange{, as shown in Fig.~\ref{fig:bidirection}}. The nodes $\tA$ and $\tB$, henceforth called the user nodes, wish to exchange information with each other. However, they are connected only to $\tR$, and not to each other directly. The node $\mathtt{R}$ acts as a bidirectional relay between $\mathtt{A}$ and $\mathtt{B}$, and facilitates communication \markchange{between them}.  All nodes are assumed to operate in half-duplex mode (they cannot transmit and receive simultaneously), and all links between nodes are wireless (unit channel  gain) additive white Gaussian noise (AWGN) channels. Bidirectional relaying in such settings has been studied extensively in the recent literature~\cite{Baik,Nazer11,Popovski,Wilson,Zhang}.

We use the compute-and-forward framework proposed in~\cite{Nazer11,Wilson} for bidirectional relaying, and we briefly describe a binary version for completeness and clarity. Suppose that $\mathtt{A}$ and $\mathtt{B}$ possess bits $X$ and $Y$, respectively. We will assume that $X$ and $Y$ are generated independently and uniformly at random. The goal in bidirectional relaying is to transmit $X$ to $\mathtt{B}$ and $Y$ to $\mathtt{A}$ through $\mathtt{R}$. To achieve this goal, a compute-and-forward protocol takes place in two phases as shown in Fig.~\ref{fig:bidirec}: 
(1) the (Gaussian) multiple access phase or the MAC phase, 
where the user nodes simultaneously transmit to the relay, and (2) the broadcast phase, where the relay transmits to the user nodes. In the MAC phase, the user nodes $\mathtt{A}$ and $\mathtt{B}$ independently modulate their bits $X$ and $Y$ into \markchange{real-valued} symbols $U$ and $V$, respectively. The relay receives an instance of a random variable $W$, \markchange{that} can be modeled as
\begin{equation}
      W=U+V+Z,
      \label{eq:0}
\end{equation}
where it is assumed that the links $\mathtt{A}\to\mathtt{R}$ and $\mathtt{B}\to\mathtt{R}$ have unit gain, $Z$ denotes additive white Gaussian noise independent of $U$ and $V$, and communication is assumed to be synchronized. Using $W$, the relay computes the XOR of the two message bits, i.e., $X\oplus Y$, and in the broadcast phase, encodes it into a real symbol which is transmitted to the two users over a broadcast channel. Note that $\mathtt{A}$ and $\mathtt{B}$ can recover $Y$ and $X$, respectively, from $X\oplus Y$. 
\begin{figure}[tbh]
     \begin{center}
          \resizebox{4cm}{!}{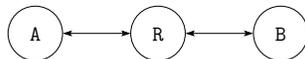}
    \caption{Bidirectional relay.} 
    \end{center}
\label{fig:bidirection}
\end{figure}
\begin{figure}[tbh]
      \begin{center}
	    \resizebox{6cm}{!}{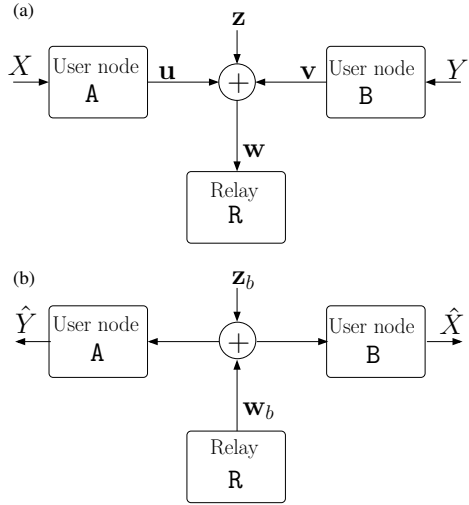} 
      \end{center}
      \caption{Bidirectional relaying: (a) MAC phase, (b) Broadcast phase.}
      \label{fig:bidirec}
\end{figure}

In \markchange{the} compute-and-forward bidirectional relaying problem \markchange{described} above, we study the scenario where an additional secrecy constraint is imposed on the relay $\mathtt{R}$. Specifically, we require that, in the MAC phase, the relay remain ignorant of the individual bits $X$ and $Y$, while still being able to compute the XOR $X\oplus Y$ reliably. 
\markchange{The relay is assumed to be ``honest-but-curious'': it behaves like a passive eavesdropper, but otherwise helps in the exchange of messages. 
}
We study the problem under two secrecy constraints: \markchange{perfect secrecy, which we describe next, and strong secrecy, which we describe further below}. \markchange{\emph{Perfect secrecy} refers to the requirement} that the relay be fully ignorant of the individual bits, i.e., that the random variables $U+V$, $X$, and $Y$  be pairwise independent. More generally, the user nodes encode the messages $X$ and $Y$ into $d$-dimensional real vectors $\U$ and $\V$ respectively, and we require $\U+\V$ to be statistically independent of each individual message.   The problem of secure bidirectional relaying in \markchange{the} presence of an untrusted relay under a perfect secrecy constraint has not been studied prior to this work, and this is a major contribution of this paper. 

We \markchange{propose} a coding scheme \markchange{for secure bidirectional relaying} that uses a pair of nested lattices $(\Lfd,\Lcd)$, with $\Lcd\subset \Lfd$. 
In our scheme, the messages are mapped to the cosets of the \emph{coarse lattice} $\Lcd$ in the \emph{fine lattice} $\Lfd$. Given a message (say, the $j$th coset, $\L_j$) at the user node, the output of the encoder is a random point chosen from that coset according to a distribution $p_j$. This distribution is obtained by sampling and normalizing over $\L_j$, a well-chosen density function $f$ on $\R^d$. We will show that if the characteristic function of $f$ is  supported within the fundamental Voronoi region of the Fourier dual of $\Lcd$, then it is possible to achieve perfect secrecy. We then study the average transmit power and achievable rates for reliable and secure communication. We will show that a transmission rate of \markchange{$\left[\frac{1}{2}\log_2\frac{\cP}{\nsvar}-\log_2 2e\right]^{+}$} is achievable with perfect secrecy, \markchange{where $[x]^+$ denotes $\max\{ x,0 \}$}. Our coding scheme  for security is explicit, in that given \emph{any} pair of nested lattices, we precisely specify the distributions $p_j$ that must be used to obtain independence between $\U+\V$ and the individual messages.

We later relax the secrecy constraint, and only demand that the mutual information between $\U+\V$ and the individual messages be arbitrarily small for large block lengths, \markchange{a requirement that is} referred to as \emph{strong secrecy}~\cite{Maurer00}. We again use a nested-lattice coding scheme, but now the distributions $p_j$ are obtained by sampling and normalizing a Gaussian function, instead of a density having a compactly supported characteristic function. The idea of using probability mass functions (pmfs) obtained by sampling Gaussians was used~\cite{Ling13} in the context of the Gaussian wiretap channel, and we will make use of the techniques developed there. Using this scheme, we show that a rate of \markchange{$\left[\frac{1}{2}\log_2\left(\frac{1}{2}+\frac{\cP}{\nsvar}\right)-\frac{1}{2}\log_2 2e\right]^+$} is achievable. 

\markchange{We show that our schemes can achieve secrecy even in the absence of noise, and that the addition of noise cannot leak any extra information to the relay. This allows us to develop the solution in two parts: first, we give coding schemes based on nested lattices that achieve secrecy over a noiseless channel. Then, we require the lattices to satisfy certain additional ``goodness'' properties in order to have reliable decoding in the presence of noise. The signal (codeword) transmitted by each user acts as a jamming signal for the other user's message, and this helps achieve secrecy. In our scheme, the channel noise is not used to increase confidentiality, unlike the Gaussian wiretap channel~\cite{Ling13} where an increase in the noise variance on the eavesdropper's link can be used to achieve higher transmission rates. It may be possible to harness the additive noise in the MAC phase to obtain higher achievable rates, but we do not pursue this in the present work. However, our approach does offer an advantage: since our scheme guarantees secrecy in the absence of noise, the security properties continue to hold even when channel noise is present, and this is true \emph{irrespective} of the noise distribution. Indeed, our scheme provides secrecy even if the channel noise follows an unknown probability distribution, a property that is in general not satisfied by coding schemes for wiretap channels. We only require the noise to be additive and independent of the transmitted codewords.}

It is worth emphasizing the basic idea behind the construction of  encoders in our coding schemes. Given a pair of nested lattices, the user nodes send points from the fine lattice in the nested lattice pair according to a pmf obtained by sampling a well-chosen density function at the fine lattice points. The choice of the density function determines the level of security that is achievable.

In prior work, the problem of secure bidirectional relaying in the presence of an untrusted relay was studied by He and Yener in~\cite{HeYener}, who showed that the mutual information rate, defined to be $\frac{1}{d}\mathcal{I}(X;\U+\V)=\frac{1}{d}\mathcal{I}(Y;\U+\V)$ goes to zero for large blocklengths $d$. They later studied the problem under a strong secrecy constraint  in~\cite{HeYenerstrong}, and gave a scheme based on nested lattice codes and universal hash functions. \markchange{Using probabilistic arguments, they showed the existence of linear hash functions for randomization at the encoders that achieve strong secrecy. In both scenarios, they showed that a rate of $\left[ \frac{1}{2}\log_2\left(\frac{1}{2}+\frac{\cP}{\nsvar}\right)-1 \right]^{+}$ is achievable. The achievable rates guaranteed by our strongly secure scheme is slightly lower than that obtained in~\cite{HeYenerstrong}. However,  our scheme avoids the use of hash functions, and given a pair of nested lattices that satisfy certain ``goodness'' properties\footnote{Unfortunately, there are no \markchange{known} explicit constructions of lattices that satisfy these properties, but only existence results based on probabilistic arguments.}, we give an explicit probability distribution for randomization at the encoders that can be used to obtain strong secrecy. }

\markchange{The idea of using nested lattice codes for secure communication is not new. They have been proposed for secure communication in other scenarios, particularly the Gaussian wiretap channel (see e.g.,~\cite{Belfiore10,Ling13,Belfiore}). They have also been used in interference networks~\cite{Agrawal09}, and for secret key generation using correlated Gaussian sources~\cite{Nitinawarat12}.
}

\markchange{Recall that the compute-and-forward protocol has two phases: a MAC phase and a broadcast phase.
We will restrict our study exclusively to the MAC phase, since there is no security requirement in the broadcast phase and the relay can use a capacity-approaching code to broadcast $X\oplus Y$ to the users.}

\subsection*{Organization of the paper}

We establish some basic notation and recall some definitions related to lattices in Section~\ref{sec:defns}.
We describe the secure bidirectional relaying problem in Section~\ref{sec:probdescp}, and then proceed to design coding schemes under the perfect secrecy constraint in Section~\ref{sec:perfectsecr}. The main result under the perfect secrecy constraint is given in Theorem~\ref{thm:perfect_main}. We give a randomized encoding scheme for any arbitrary nested lattice code that achieves perfect secrecy in the absence of noise in Section~\ref{sec:noiseless}, then study the effect of additive noise and find achievable transmission rates in Section~\ref{sec:gnoise}.
Thereafter, we study the same problem under a strong secrecy constraint, design coding schemes, and evaluate the performance in Section~\ref{sec:strongsecr}, with the main result summarized in Theorem~\ref{thm:strong_main}. 
\markchange{In Section~\ref{sec:multihop}, we show that these schemes can be extended to  the multi-hop line network~\cite{HeYener}  and find achievable transmission rates under the two secrecy constraints. We make some concluding remarks in Section~\ref{sec:conc}.} Most of the technical proofs are given in appendices.

\section{Definitions and Notation}\label{sec:defns}
We first describe the notation we will use throughout the paper. We denote the set of real numbers by $\R$, and integers by $\Z$. We use the notation $\R^+$ for the set of nonnegative real numbers. The number of elements in a finite set $S$ is denoted by $|S|$. 
\markchange{If $x$ is a real number, then $[x]^{+}$ is defined as $\max\{x,0\}$.} 
Random vectors are denoted in boldface upper case, e.g., $\U$, and their instances in boldface lower case, as in $ \u $. The components of the vectors are denoted in normal font, e.g., $ \x = [x_{1}\:x_{2}]^{T} $. Matrices are represented in sans-serif, as in $ \mathsf{H} $. The Euclidean ($ \ell^{2} $) norm of a column vector $ \h $ is denoted by $ \Vert \h \Vert $. The identity matrix of size $ M\times M $ is denoted by $ \mathsf{I}_{M} $. 

The probability of an event $A$ is denoted by $ \Pr[A] $. 
If $ X $ is a random variable, then $ \mathcal{H}(X) $ denotes the entropy of $ X $. \markchange{The symbol $\mathbb{E}[\cdot]$ denotes expectation.} 
\markchange{The characteristic function of a random variable $X$ is the function $\psi(t)=\mathbb{E}[e^{iXt}]$, for $t\in\R$. }
For random variables $ X,Y $, the notation $ X\independent Y $ means that  $ X $ and $ Y $ are 
independent. The mutual information between $X$ and $Y$ is denoted by $\mathcal{I}(X;Y)$. 

Let $f(n)$ and $g(n)$ be \markchange{ sequences} of positive real numbers. We say that $g(n)=o(f(n))$ if $g(n)/f(n)\to 0$ as $n\to\infty$. Also, $g(n)=o_n(1)$ if $g(n)\to 0$ as $n\to\infty$. 
Furthermore, $g(n)=\Omega(f(n))$ if there exists a constant $K>0$ such that $g(n)>Kf(n)$ for all sufficiently large $n$, and $g(n)=\mathcal{O}(f(n))$ if there exists a constant $K>0$ such that $g(n)<Kf(n)$ for all sufficiently large $n$. 
\subsection{Lattices in $\R^d$}\label{sec:latticedefns}

 We briefly recall some definitions of lattices and their properties. For a more detailed treatment, see e.g.,~\cite{Barvinok,Conway}.

Let $k,d$ be positive integers with $k\leq d$. Suppose $\u_1,\u_2,\ldots,\u_k$ are linearly independent column vectors in $\R^d$. Then the set of all integer-linear combinations of the $\u_i$'s, $\L=\{ \sum_{i=1}^ka_i\u_i : a_i\in \Z, 1\leq i\leq k\}$, is called a $k$-dimensional \emph{lattice} in $\R^d$. It is easy to verify that $\L$ forms an Abelian group under componentwise addition. The collection of vectors $\{\u_1,\u_2,\ldots,\u_k\}$ is called a \emph{basis} for the lattice $\L$; \markchange{clearly, the basis of a lattice is not unique, e.g., $\{-\u_1,-\u_2,\ldots,-\u_k  \}$ is also a basis. }

The $k\times d$ matrix \markchange{$\mathsf{A}:= [\u_1\;\u_2\;\cdots\;\u_k]^T$} is called a \emph{generator matrix} of $\L$, and we say that the vectors $\u_1,\u_2,\ldots,\u_k$ generate $\L$. We write \markchange{$\L=\mathsf{A}^T\Z^k:= \{ \mathsf{A}^T\x : \x\in \Z^k \}$}.  If $\L$ is full-rank (i.e., $\L$ is a $d$-dimensional lattice in $\R^d$), then the \emph{determinant} of $\L$, denoted by $\text{det}\L$ is defined to be $|\text{det}(\mathsf{A})|$. It is a standard fact that $\text{det}\L$ does not depend on the generator matrix. 
Unless mentioned otherwise, we will henceforth consider full-rank lattices in $\R^d$. 

If $\L$ and $\Lc$ are two lattices in $\R^d$ such that $\Lc\subset\L$, then we say that $\Lc$ is a \emph{sublattice} of $\L$, or $\Lc$ is \emph{nested} within $\L$.  We call $\L_0$ the \emph{coarse lattice}, and $\L$, the \emph{fine lattice}. The number of cosets of $\L_0$ in $\L$ is called the \emph{index} of $\L_0$ in $\L$, denoted by $|\L/\L_0|$. It is a standard fact that $|\L/\L_0|=\text{det}\L_0/\text{det}\L$~\cite[Theorem 5.2]{Barvinok}.

If $\mathsf{A}$ is a generator matrix of a lattice $\L$, then $\L^{*}:=\{ (\mathsf{A}^{-1})^T\z:\z\in \Z^d\}$ is called the \emph{dual lattice} of $\L$. The dual lattice $\L^{*}$ is also equal to $\{ \x: \sum_{i=1}^{d}x_iy_i\in\Z \text{ for every } \y\in\L\}$\cite{Barvinok}. The \emph{Fourier dual} of $\L$, denoted $\hat{\L}$, is defined as $2\pi \L^{*}$.

For any $\x\in\R^d$, we define the nearest neighbour quantizer $Q_{\L}(\x):= \text{arg min}_{\mathbf{\lambda}\in\L}\Vert \x-\mathbf{\lambda} \Vert$ to be the function which maps $\x$ to the closest point in $\Lf$. The \emph{fundamental Voronoi region} of $\L$ is defined as $\cV(\L):=\{ \y:Q_{\L}(\y)=\mathbf{0} \}$. The volume of the fundamental Voronoi region, $\text{vol}(\cV(\L))$ is equal to $\text{det}\L$~\cite{Barvinok, Conway}.

For any $\x\in \R^d$, we define the modulo-$\L$ operation as $[\x]\bmod\L:= \x-Q_{\L}(\x)$. In other words, $[\x]\bmod\L$ gives the quantization error of the nearest neighbour quantizer $Q_{\L}(\cdot)$. Figure~\ref{fig:modL0} illustrates the $Q_{\L}(\cdot)$ and the modulo-$\L$ operations.

\begin{figure}[t]
 \centerline{\resizebox{8cm}{!}{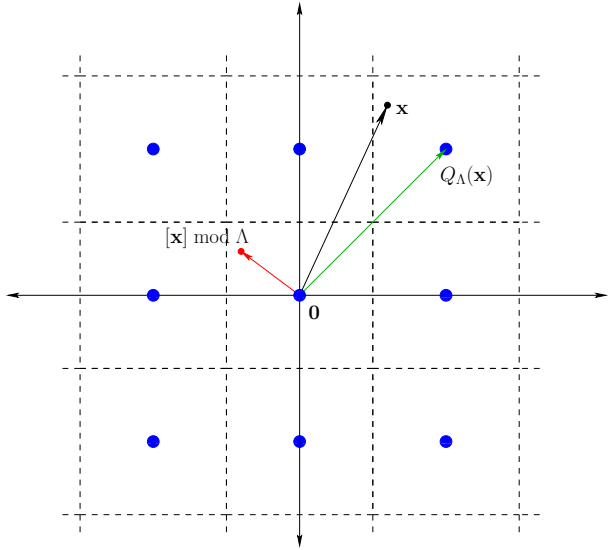}}
 \caption{Illustrating the $Q_{\L}(.)$ and the $[.]\bmod \L$ operation for the $\Z^2$ lattice.}
 \label{fig:modL0}
\end{figure}

The \emph{covering radius} of $\L$, denoted by $\rcov(\L)$, is defined as the radius of the smallest closed ball in $\R^d$ centered at $\mathbf{0}$ which contains $\cV(\L)$. 
The \emph{effective radius}, $\reff(\L)$, is defined as the radius of a ball in $\R^d$ having the same volume as that of $\cV(\L)$. The \emph{packing radius}, $\rpack(\L)$, is the radius of the largest open ball centered at $\mathbf{0}$ which is contained in $\cV(\L)$. Clearly, $\rcov(\L)\geq\reff(\L)\geq\rpack(\L)$. These parameters are illustrated for the hexagonal lattice in Fig.~\ref{fig:rcov_rpack}.

\begin{figure}
      \begin{center}
	    \resizebox{5cm}{!}{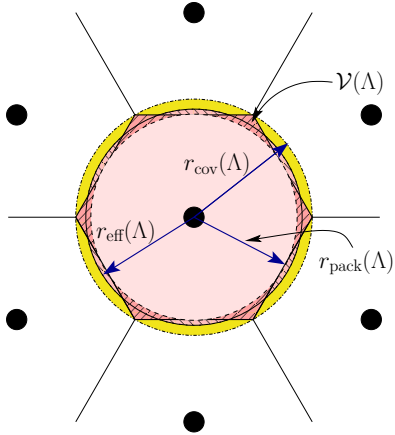}
      \end{center}
\caption{Illustrating the covering, packing and effective radii of the hexagonal lattice.}
\label{fig:rcov_rpack}
\end{figure}
The \emph{normalized second moment per dimension} of $\L$ is defined as
\begin{equation}
 \mathcal{G}_{\L}=\frac{1}{d\left(\text{det}\L \right)^{1+2/d}}\int_{\cV(\L)}\Vert \y \Vert^2\: d\y.
\label{eq:Glambda}
\end{equation}


\section{Description of the Problem} \label{sec:probdescp}

The general set-up is as follows: two user nodes, denoted by $\mathtt{A}$ and $\mathtt{B}$, possess messages taking values independently and uniformly in a finite set. For the purposes of computation at the relay, the messages are mapped into random variables $X$ and $Y$ taking values in a finite Abelian group $\Gpd$, where the choice of $\Gpd$ is left to the system designer. The mapping is such that the random variables $X$ and $Y$ remain uniformly distributed over $\Gpd$, and we will see later that this distribution helps in achieving secrecy. The addition operation in the group $\Gpd$ is denoted $\oplus$.
\markchange{The encoder at node $\tA$ maps the given message $X$ into a random $d$-dimensional real vector $\U$. In a similar fashion, the encoder at $\tB$ maps the message $Y$ to a random vector $\V$. }
The user nodes transmit their respective vectors to the relay simultaneously, and at the end of the MAC phase, the relay obtains 
\begin{equation}
 \W=\U+\V+\bZ,
\label{eq:mac_ch}
\end{equation}
where $\bZ$ is a Gaussian random vector with zero mean and covariance matrix $\sigma^{2}\mathsf{I}_{d}$, 
where $+$ denotes componentwise real addition. 
The coding scheme at each user node must ensure that the relay can recover $X\oplus Y$ reliably from $\W$, and one of the following:
\markchange{
\begin{itemize}
     \item \emph{Perfect secrecy}: The mutual information between $\W$ and each individual message is exactly zero\footnote{Equivalently, we want $\W\independent X$ and $\W\independent Y$.}, i.e., $\mathcal{I}(\W;X)=\mathcal{I}(\W;Y)=0$. 
     \item \emph{Strong secrecy}: $\mathcal{I}(\W;X)$ and $\mathcal{I}(\W;Y)$ can be made arbitrarily small for all sufficiently large $d$.
\end{itemize}}
We in fact impose a slightly stronger security criterion than the one mentioned above. Even in the absence of noise, the mutual information between $\W=\U+\V$ and each individual message must be either zero (perfect secrecy) or can be made arbitrarily small for all sufficiently large $d$ (strong secrecy). 
Since the additive noise is independent of everything else, $X\to \U+\V\to \U+\V+\bZ$ forms a Markov chain, and using the data processing inequality, $\mathcal{I}(X;\U+\V+\bZ)\leq \mathcal{I}(X;\U+\V)$. Likewise, $\mathcal{I}(Y;\U+\V+\bZ)\leq \mathcal{I}(Y;\U+\V)$. 
Therefore, any scheme that achieves perfect (strong) secrecy in the absence of noise will also achieve perfect (strong) secrecy in a noisy channel.
 
The messages must also be protected from corruption by the additive noise in the multiple access phase. 
 Since the messages are uniformly distributed over $\Gpd$, $\frac{1}{d}\log_2|\Gpd|$ gives the average number of bits of information sent to the relay by each user node in one channel use in the MAC phase. Our aim will be to ensure secure computation of $X\oplus Y$ at the highest possible rate (which we define to be $\frac{1}{d}\log_2|\Gpd|$) for a given power constraint at the user nodes. To formalize these notions, we have the following definition: 

\begin{definition}
For a positive integer $d$, a $(d,M^{(d)})$ code for the MAC phase of the bidirectional relay channel with user nodes $\mathtt{A}$, $\mathtt{B}$ and relay $\mathtt{R}$ consists of the following:
\begin{enumerate}
\item \textbf{Messages:} Nodes $\mathtt{A}$ and $\mathtt{B}$ possess messages $X$ and $Y$, respectively, drawn independently and uniformly from a finite Abelian group $\Gp^{(d)}$ with $M^{(d)}=|\Gp^{(d)}|$ elements. 
\item \textbf{Codebook:} The codebook, denoted by $\cC$, is a discrete subset of $\R^d$, not necessarily finite. The elements of $\cC$ are called codewords. The codebook consists of all those vectors that are allowed to be transmitted by the user nodes to the relay.
\markchange{\item \textbf{Encoders:} The encoder at each node is a randomized mapping from $\Gp^{(d)}$ to $\R^{d}$, specified by the distributions $ p_{\U|X}(\u|x)=\Pr[\U=\mathbf{u}|X=x]$ and $ p_{\V|Y}(\v|y)=\Pr[\V=\mathbf{v}|Y=y]$ for all $\mathbf{u},\v\in \cC$ and $x,y\in \Gp^{(d)}$.  At node $\mathtt{A}$, given a message $x\in \Gp^{(d)}$ as input, the encoder outputs a codeword $\u\in \cC$ at random, according to $p_{\U|X}(\u|x)$. Similarly, at node $\mathtt{B}$, with $y$ as input, the encoder outputs $\v\in\cC$ according to $p_{\V|Y}(\v|y)$. 
The messages $x$ and $y$ are encoded independently.
}
 The \emph{rate} of the code is defined to be
  \begin{equation}
    \label{eq:02}
    \rte^{(d)}=\dfrac{\log_2M^{(d)}}{d}.
  \end{equation}
The code has an \emph{average transmit power per dimension} defined as
\begin{equation}
  \label{eq:3}
  \pow^{(d)} = \frac{1}{d} \mathbb{E}\Vert \U \Vert^{2}=\frac{1}{d} \mathbb{E}\Vert \V \Vert^{2}.
\end{equation}
\item \textbf{Decoder:} The relay $\mathtt{R}$ receives a vector $\W\in\R^{(d)}$ as given in (\ref{eq:mac_ch}). The decoder, $\cD^{(d)}: \R^{d} \to \Gp^{(d)}$  maps the received vector to an element of the set of messages. The average probability of error of the decoder is defined as
\[
	\Pe := \mathbb{E} \big[ \Pr [ \cD^{(d)}(\W) \neq X \oplus Y ] \big] ,
	\]
where $\mathbb{E}$  denotes expectation over the messages, $X,Y$, and over the encoders ($\U,\V$ given $X,Y$).
\end{enumerate}
\label{def:code}
\end{definition}

\section{Perfect Secrecy}\label{sec:perfectsecr}

We first study the case where perfect statistical independence between $\U+\V$ and the individual messages is required,
and the relay must be able to reliably compute $X\oplus Y$ (where $\oplus$ denotes addition within $\Gpd$) from the received vector. To summarize, we have the following requirements for secure compute-and-forward: 
\begin{enumerate}
 \item[(S1)] $(\U,X)\independent (\V,Y)$.
 \item[(S2)] $(\U+\V)\independent X$ and $(\U+\V)\independent Y$.
 \item[(S3)] $\U+\V$ almost surely determines $X\oplus Y$.
\end{enumerate}
If conditions (S1)--(S3) are satisfied, the relay has no means of finding the individual messages. Property (S3) ensures that the relay can decode $X\oplus Y$, which can then be encoded/modulated for further transmission over the broadcast channel. On reception of the broadcast message, since user $\mathtt{A}$ (resp. $\mathtt{B}$) knows $X$ (resp. $Y$), it can recover $Y$ (resp. $X$).

If the relay only had access to $X\oplus Y$ instead of $\U+\V$, the problem of secure communication would have been trivial due to the uniformity and independence of $X$ and $Y$. However, the relay receives the real sum of $\U$ and $\V$, which makes the problem harder. For example, suppose that $d=1$, and $\Gp^{(1)}=\Z_2$, the group of integers modulo $2$. Consider the coding scheme $\U=X$, and $\V=Y$. Then, in the absence of noise, whenever $\U+\V=0$ or $\U+\V=2$, the relay can determine both $X$ and $Y$. 

The performance of a coding scheme is generally evaluated in terms of the average transmit power,
and the transmission rate.
To make these notions formal, we define achievable power-rate pairs as follows.
\begin{definition} 
A power-rate pair $(\cP,\cR)$ is \emph{achievable with perfect secrecy} if, for every $\delta>0$, there exists a sequence of $(d,M^{(d)})$ codes such that
\begin{itemize}
 \item conditions (S1)--(S3) are satisfied for all $d$,
\end{itemize}
and for all sufficiently large $d$, 
\begin{itemize}
      \item the transmission rate, $\rte^{(d)}$, is greater than $\cR-\delta$;
      \item the average transmit power per dimension $\pow^{(d)}$, is less than $\cP+\delta$; and
      \item the average probability of decoding error, $\Pe$, is less than $\delta$.
\end{itemize}
\label{def:pr_pairsperfect}
\end{definition}

The objective of the next couple of sections will be to prove the following result.
\begin{theorem}
 A power-rate pair of
\[
      \left( \cP, \left[\frac{1}{2}\log_2\frac{\cP}{\nsvar}-\log_2(2e)\right]^{+} \right)
\]
is achievable with perfect secrecy in the MAC phase of the bidirectional relay.
\label{thm:perfect_main}
\end{theorem}

\section{Perfect secrecy: The Noiseless Setting}\label{sec:noiseless}

To get a clear picture as to how secure communication can be achieved, we first describe the binary case. The messages $X$ and $Y$ are chosen independently and uniformly at random from $\{0,1\}$, or equivalently, the set of integers  modulo-2 ($\Gp=\Z_{2}$). They are modulated to $U$ and $V$ respectively, which take values in $\R$. 
Studying the one-dimensional case will give us the intuition needed to tackle the general case, and we will see that the techniques developed here extend quite naturally to the $d$-dimensional setting.

We  will show that there exist distributions on $U$ and $V$ that permit secure computation defined by properties (S1)--(S3). This is somewhat surprising since we cannot have non-degenerate real-valued random variables $U,V$ that satisfy $(U+V)\independent U$ and $(U+V)\independent V$, as shown in the following proposition:

\begin{proposition}
Let $U$ and $V$ be independent real-valued random variables, 
and let $+$ denote addition over $\R$. Then, we have 
$(U + V) \independent U$ and $(U + V) \independent V$ iff
$U$ and $V$ are constant a.s. (i.e., there exist $a,b \in \R$ such
that $\Pr[U=a] = \Pr[V=b] = 1$). 
\label{prop:int_sum}
\end{proposition}
\begin{proof}
The ``if'' part is trivial, so let us prove the ``only if''
part. Let $W = U+V$, so that by assumption, $U$, $V$ and $W$ are
pairwise independent. Let $\varphi_U$, $\varphi_V$ and
$\varphi_W$ denote  the characteristic functions of $U$, $V$ and $W$,
respectively. In particular, $\varphi_W = \varphi_U \varphi_V$.
From $U = W-V$, we also have that $\varphi_U = \varphi_W
\overline{\varphi_V}$, where $\overline{\varphi_V}$ denotes the complex
conjugate of $\varphi_V$. Putting the two equalities together, we
obtain $\varphi_U = \varphi_U |\varphi_V|^2$. To be precise,
$\varphi_U(t) = \varphi_U(t) |\varphi_V(t)|^2$ for all $t \in \R$.

Now, characteristic functions are continuous and take the value 1 
at $t=0$. Hence, $\varphi_U$ is non-zero within the interval
$[-\delta,\delta]$ for some $\delta > 0$.
Thus, $|\varphi_V(t)| = 1$ for all $t \in [-\delta,\delta]$. By 
a basic property of characteristic functions (see Lemma~4 of
Section~XV.1 in \cite{Feller}), this implies that there exists $b \in \R$
such that $\varphi_V(t) = e^{ibt}$ for all $t \in \R$, thus proving
that $V = b$ with probability 1.

A similar argument using $V = W-U$ shows that $U$ is also
constant with probability 1.
\end{proof}

\subsection{Secure Computation of XOR at the Relay}
In this section, $X$ and $Y$ are independent and identically distributed (iid) uniform binary random variables (rvs), and $X \oplus Y$ denotes their modulo-2 sum (XOR). We describe a construction of integer-valued rvs $U$ and $V$ satisfying the properties (S1)--(S3). 

\subsubsection{Conditions on PMFs and Characteristic Functions}
We first derive conditions under which integer-valued rvs $U$ and $V$ can satisfy the desired properties.
We introduce some notation: for $k \in \Z$, let 
$p_U(k) = \Pr[U = k]$, $p_V(k) = \Pr[V = k]$, and for $a \in \{0,1\}$, 
let $p_{U|a}(k) = \Pr[U = k \mid X = a]$, 
$p_{V|a}(k) = \Pr[V = k \mid Y = a]$. Thus, 
$p_U = (1/2)(p_{U|0} + p_{U|1})$ and $p_V = (1/2)(p_{V|0} + p_{V|1})$.

Property~(S1) is equivalent to requiring that the joint probability mass function (pmf) of $(U,V,X,Y)$ be expressible as
\begin{equation}
p_{UVXY}(k,l,a,b) = (1/2)(1/2)  p_{U|a}(k) p_{V|b}(l) 
\label{eq:Z1}
\end{equation}
for $k,l \in \Z$ and $a,b \in \{0,1\}$.
  Without the requirement that $U+V \independent X$ and $U+V \independent Y$,
 it is trivial to define $U$ and $V$ such that (S3) is satisfied:
 for example, take $U = X$ and $V = Y$. 
Property~(S3) is satisfied by any $U,V$ such that
\begin{equation}
\begin{array}{c}
p_{U|0}(k) = p_{V|0}(k) = 0 \ \text{ for all odd } k \in \Z, \\
p_{U|1}(k) = p_{V|1}(k) = 0 \ \text{ for all even } k \in \Z. 
\end{array}
\label{eq:Z3}
\end{equation}

Finally, we turn our attention to (S2). We want $(U+V)\independent X$ and $(U+V)\independent Y$. Let us define, for $k \in \Z$, $p_{U+V}(k) = \Pr[U+V=k]$, and
for $a \in \{0,1\}$, $p_{U+V \mid X=a}(k) = \Pr[U+V=k \mid X=a]$ and 
$p_{U+V \mid Y = a}(k) = \Pr[U+V=k \mid Y=a]$.  Assuming $(U,X)
\independent (V,Y)$, we have
$p_{U+V} = p_U \ast p_V$,  $p_{U+V \mid X = a} = p_{U|a} \ast p_V$,
and $p_{U+V \mid Y = a} = p_U \ast p_{V|a}$,
where $\ast$ denotes the convolution operation. Thus, when $(U,X)
\independent (V,Y)$,  (S2) holds iff 
\begin{equation}
p_U \ast p_V = p_{U|a} \ast p_V = p_U \ast p_{V|a} \ 
\text{ for } a \in \{0,1\}.
\label{eq:Z2a}
\end{equation}

It helps to view this in the Fourier domain. Let $\varphi_U$, $\varphi_V$, $\varphi_{U|a}$ etc.\ denote 
the respective characteristic functions of the pmfs $p_U$, $p_V$, $p_{U|a}$ etc.\ --- for example, 
$\varphi_{U|a}(t) = \sum_{k\in\Z} p_{U|a}(k) e^{i k t}$. 
Then, (\ref{eq:Z2a}) is equivalent to 
\begin{equation}
\varphi_U \varphi_V = \varphi_{U|a}\varphi_V 
= \varphi_U \varphi_{V|a}\ \text{ for } a \in \{0,1\}.
\label{eq:Z2}
\end{equation}
Note that $\varphi_U = (1/2)(\varphi_{U|0} + \varphi_{U|1})$
and $\varphi_V = (1/2)(\varphi_{V|0} + \varphi_{V|1})$. Hence,
(\ref{eq:Z2}) should be viewed as a requirement on the conditional
pmfs $p_{U|a}$ and $p_{V|a}$, $a \in \{0,1\}$.

In summary, we have the following lemma.

\begin{lemma}
Suppose that the conditional pmfs $p_{U|a}$ and $p_{V|a}$, $a \in \{0,1\}$, satisfy (\ref{eq:Z3}) and (\ref{eq:Z2}). Then, the rvs $U,V,X,Y$ with joint pmf given by (\ref{eq:Z1}) have properties (S1)--(S3).
\label{basic_lemma}
\end{lemma}

The observations made up to this point also allow us to prove the following negative result.\footnote{In fact, a stronger negative result can be shown --- see Proposition~\ref{prop:exp_decay}.}

\begin{proposition}
 Properties (S1)--(S3) cannot be satisfied by integer-valued rvs $U,V$
 that are finitely supported.
 \label{prop:finite_supp}
 \end{proposition}
 \begin{proof}
 Suppose that $U$ and $V$ are finitely supported $\Z$-valued rvs.
 Then, $\varphi_U(t)$ and $\varphi_V(t)$ are finite linear combinations of some
 exponentials $e^{ik_1t},\ldots,e^{ik_nt}$. 
 Equivalently, the real and imaginary parts of $\varphi_U$ and $\varphi_V$ 
 are trigonometric polynomials. Thus, either $\varphi_U$ (resp.\ $\varphi_V$) 
 is identically zero, or it has a discrete set of zeros. The former
 is impossible as $\varphi_U(0) = \varphi_V(0) = 1$.
 Now, suppose that (S1) and (S2) are satisfied, which means that (\ref{eq:Z2})
 must hold. The equality $\varphi_U \varphi_V = \varphi_U \varphi_{V|a}$
 in (\ref{eq:Z2}) implies that $\varphi_{V|a}(t) = \varphi_V(t)$ for all $t$ 
 such that $\varphi_U(t) \neq 0$. But since $\varphi_U(t)$ has a discrete set 
 of zeros, continuity of characteristic functions in fact implies
 that $\varphi_{V|a}(t) = \varphi_V(t)$ for all $t$. An analogous argument
 shows that $\varphi_{U|a}(t) = \varphi_U(t)$ for all $t$. Hence, 
 $U \independent X$ and $V \independent Y$. From this, and (S1), 
 we obtain that $U + V \independent X \oplus Y$, thus precluding (S3).
 \end{proof}

Practical communication systems generally have a maximum power constraint, which means that we would like to have $U,V$ \markchange{be} finitely supported. But from Proposition~\ref{prop:finite_supp}, we see that it is not possible to have finitely supported $U,V$ that permit secure computation of the XOR at the relay. Therefore, in order to ensure secure computation, we will have to relax the power constraint to an \emph{average power constraint} on the user nodes. This means that we require finite-variance, integer-valued random variables $U,V$, with infinite support, that satisfy properties (S1)--(S3), or equivalently, the hypotheses of Lemma~\ref{basic_lemma}.

We now give a construction of $U,V$ that satisfy the hypotheses of Lemma~\ref{basic_lemma}.
We will choose a density function whose characteristic function is \markchange{compactly} supported. The random variables $U$ and $V$ are chosen according to a distribution obtained by sampling and appropriately normalizing this density function.
To study this in more detail, we rely upon methods and results from Fourier analysis.
The key tool we need is the Poisson summation formula, which we
briefly recall here. Our description is based largely on Section~XIX.5 in \cite{Feller}.

\subsection{The Poisson Summation Formula} \label{sec:Poi_summ}

\markchange{ 
Fix a positive integer $d$, and let $\L$ be a full-rank lattice in $\R^d$. Recall from Section~\ref{sec:latticedefns} that $\hat{\L}$ denotes the Fourier dual of $\L$. 
}

\markchange{ 
Let $\psi: \R^d \to \C$ be the characteristic function of a $\R^d$-valued random variable, such that $\int_{\R^d} |\psi(\t)| \, d\t < \infty$. In particular, $\psi$ is continuous and $\psi(\0) = 1$. 
Since $\psi$ is absolutely integrable, the random variable has a continuous density $f: \R^d \to \R^+$.
The Poisson summation formula can be expressed as follows: for any $\s \in \R^d$, we have for all $\zetabf \in \R^d$, 
\begin{equation}
\sum_{\n \in \hat{\L}} \psi(\zetabf + \n) \, e^{-i \langle \n, \, \s \rangle}
= (\det \L) \sum_{\k \in \L} f(\k + \s) \, e^{i \langle \k + \s,  \, \zetabf \rangle},
\label{eq:psf_L}
\end{equation}
provided that the series on the left converges to a continuous function $\Psi(\zetabf)$. It should be pointed out that texts in Fourier analysis typically state the Poisson summation formula for an arbitrary $L^1$ function $f$, and would then require that $f$ and $\psi$ decay sufficiently quickly --- see e.g., \cite[Chapter~VII, Corollary~2.6]{SW71} or \cite[Eq.~(17.1.2)]{Barvinok} ---  for (\ref{eq:psf_L}) to hold. However, as argued by Feller in proving the formula in the one-dimensional setting \cite[Chapter~XIX, equation~(5.9)]{Feller}, in the special case of a non-negative $L^1$ function $f$, it is sufficient to assume that the left-hand side (LHS) of (\ref{eq:psf_L}) converges to a continuous function $\Psi(\zetabf)$.
}

\markchange{ 
Note that $\Psi(\0) = (\det \L) \sum_{\k\in\L} f(\k+\s)$, which is a
non-negative quantity. If $\Psi(\0) \ne 0$, then dividing both
sides of (\ref{eq:psf_L}) by $\Psi(\0)$ yields the important fact that
$\Psi(\zetabf)/\Psi(\0)$ is the characteristic function of a discrete random
variable supported within the set $\L + \s := \{\k + \s: \k \in \L\}$, the probability
mass at the point $\k+\s$ being equal to $f(\k+\s)/\sum_{\boldsymbol{\ell} \in \L} f(\boldsymbol{\ell} + \s)$.
}

\markchange{
 A special case of interest is when $\psi$ is compactly supported; specifically, it is supported within the fundamental Voronoi region of $\hat{\L}$: $\psi(\t) = 0$ for all $\t \notin \cV(\hat{\L})$. In this case, we can readily show that the series on the LHS of (\ref{eq:psf_L}) converges to a continuous function $\Psi$. Indeed, if we define $\widetilde{\psi}(\t) := \psi(\t)e^{-i\langle \t,\s \rangle}$, then the series on the LHS of \eqref{eq:psf_L} may be written as $\Psi(\zetabf) := e^{i \langle\zetabf, \s \rangle} \widetilde{\Psi}(\zetabf)$, where 
 $$
 \widetilde{\Psi}(\zetabf) := \sum_{\n \in \hat{\L}} \widetilde{\psi}(\zetabf + \n).
$$
Now, recall that $\psi$, being a characteristic function, is continuous on $\R^d$; hence, so is $\widetilde{\psi}$. Also, by assumption, $\psi$ is supported within $\cV(\hat{\L})$; hence, so is $\widetilde{\psi}$. In particular, by continuity, $\widetilde{\psi}$ must be $0$ on the boundary of $\cV(\hat{\L})$; therefore, the supports of $\widetilde{\psi}(\cdot)$ and $\widetilde{\psi}(\cdot + \n)$ do not intersect for any non-zero $\n \in \hat{\L}$. From this, we infer that $\widetilde{\Psi}$, which is formed by the superposition of continuous functions with disjoint supports, must be continuous. Hence, we can conclude that $\Psi(\zetabf) = e^{i \langle\zetabf, \s \rangle} \widetilde{\Psi}(\zetabf)$ is a continuous function.
}

\markchange{ 
Moreover, it is clear that $\Psi(\0) = \psi(\0)$, and since $\psi$ is a characteristic function, $\psi(\0) = 1$. As explained above, this shows that $\Psi$ is the characteristic function of a discrete rv supported within $\L + \s$.  In fact, by plugging in $\zetabf = \0$ in (\ref{eq:psf_L}) we obtain that $\Psi(\0) = (\det \L) \sum_{\k \in \L} f(\k+\s)$, which shows that $\sum_{\k\in\L} f(\k+\s) = 1/(\det \L)$. For future reference, we summarize this in the form of a proposition.
}

\markchange{
\begin{proposition}
Let $\L$ be a full-rank lattice in $\R^d$. Let $\psi:\R^d \to \C$ be a characteristic function such that $\psi(\t) = 0$ for all $\t \notin \cV(\hat{\L})$, and let $f:\R^d \to \R^+$ be the corresponding probability density function. 
Then, for any $\s \in \R^d$, the function $\Psi: \R^d \to \C$ defined by 
\begin{equation*}
\Psi(\zetabf) = \sum_{\n \in \hat{\L}} \psi(\zetabf + \n) \, e^{-i \langle \n, \, \s \rangle}
\end{equation*}
is the characteristic function of a random variable supported within the set $\L+\s  := \{\k + \s:  \k \in \L\}$.
The probability mass at the point $\k + \s$ is equal to $(\det \L) \, f(\k+\s)$. 
\label{prop:psf_Rd}
\end{proposition}
}

\markchange{
It should be noted that compactly supported characteristic functions do indeed exist --- see e.g., \cite[Section~XV.2, Table~1]{Feller}, \cite{EGR04}, \cite{RS04}. We also give an explicit construction in Example~\ref{ex:psi} in Section~\ref{sec:constr}.
}

\markchange{ 
Applying Proposition~\ref{prop:psf_Rd} to the one-dimensional lattice $T \Z = \{kT : k \in \Z\}$, with $T > 0$, we obtain the corollary below.
}

\markchange{
\begin{corollary}
Let $\psi$ be a characteristic function of a real-valued random variable 
such that $\psi(t) = 0$ whenever $|t| \ge \pi/T$ for some $T > 0$, and let $f$ be the corresponding
probability density function. Then, for any $s \in \R$, 
the function $\Psi: \R \to \C$ defined by 
\begin{equation*}
\Psi(\zeta) = \sum_{n=-\infty}^\infty \psi(\zeta + 2n\pi/T) \,
e^{-is(2n\pi/T)} 
\end{equation*}
is the characteristic function of a discrete random variable supported
within the set $\{kT+s: k \in \Z\}$. The probability mass at the point
$kT + s$ is equal to $T f(kT+s)$.
\label{cor:psf}
\end{corollary}
}

\markchange{ 
This corollary plays a central role in the construction described next.
}

\subsection{Construction of $\Z$-Valued RVs Satisfying (S1)--(S3)\label{sec:constr}}

\begin{figure}[t]
\centerline{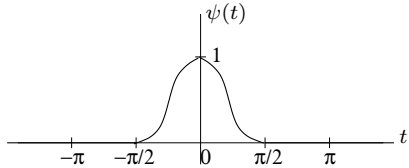}
\caption{A generic characteristic function supported on $[-\pi/2,\pi/2]$.}
\label{fig:generic_psi}
\end{figure}

 We now describe the construction of integer-valued rvs that satisfy (S1)--(S3). Let $\psi$ be a characteristic function (of a continuous rv $X$) with the properties that 
\begin{itemize}
\item[(C1)] $\psi(t) = 0$ for $|t| \ge \pi/2$, and
\item[(C2)] $\psi(t)$ is real and non-negative for all $t \in \R$.\footnote{There
  is no loss of generality in imposing this requirement. Suppose that an rv
$X$ has characteristic function $\psi$, which is complex-valued in general. 
Let $X_1,X_2$ be iid rvs with the same distribution as $X$. Then, $X_1-X_2$ has
characteristic function $\psi \bar{\psi} = |\psi|^2$.}
\end{itemize}
 A generic such $\psi$ is depicted in Figure~\ref{fig:generic_psi}; we give a specific example a little later in this section. Since $\psi$ is real-valued, it must be an even function: 
$\psi(-t) = \psi(t)$ for all $t \in \R$.  Also, $\psi(0) = 1$. 
Moreover, since $\psi$ is integrable over $\R$, by the Fourier inversion formula,
the rv $X$ has a continuous density $f$. Note that Corollary~\ref{cor:psf} holds for $T \le 2$.

Let $\varphi$ be the periodic function with period $2\pi$ that agrees
with $\psi$ on $[-\pi,\pi]$, as depicted in Figure~\ref{fig:phi}. 
Note that $\varphi(\zeta) = \sum_{n=-\infty}^{\infty} \psi(\zeta +
2\pi n)$. Thus, applying Corollary~\ref{cor:psf} 
with $T = 1$ and $s = 0$, we find that $\varphi$
is the characteristic function of an integer-valued rv, 
with pmf given by
\begin{equation}
p(k) = f(k) \text{ for all } k \in \Z.
\label{eq:p}
\end{equation}

 \begin{figure}[t]
 \centerline{\resizebox{8cm}{!}{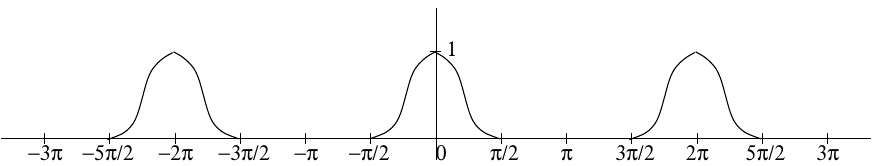}}
 \caption{Period-$2\pi$ extension of generic $\psi$ from Figure~\ref{fig:generic_psi}.}
 \label{fig:phi}
 \vspace{0.4cm}
\centerline{\resizebox{8cm}{!}{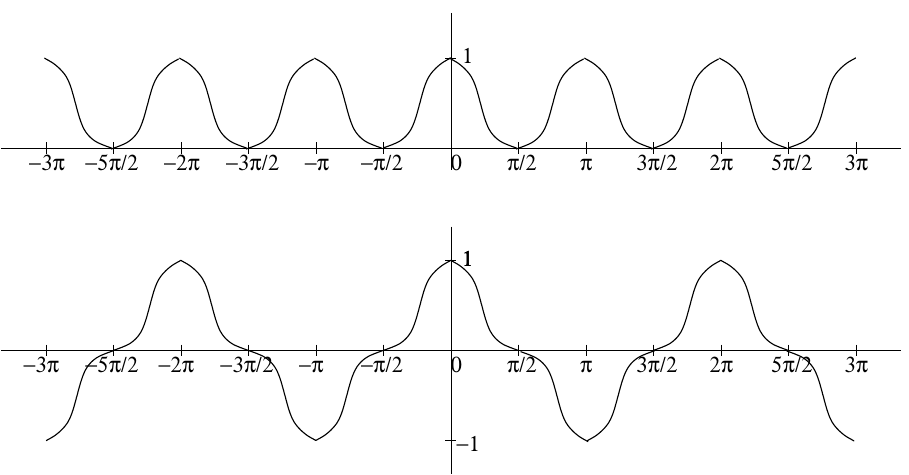}}
\caption{The periodic functions $\varphi_0$ and $\varphi_1$ derived
  from $\psi$.}
\label{fig:phi0_phi1}
 \end{figure}

Next, for $s = 0,1$, define $\varphi_s$ as follows: 
for $\zeta \in \R$,
$$
\varphi_s(\zeta) = \sum_{n=-\infty}^\infty \psi(\zeta+n\pi) e^{-i s n \pi}.
$$
It is easily seen that $\varphi_0$ is the periodic extension of $\psi$
with period $\pi$, i.e., $\varphi_0$ is the periodic function with
period $\pi$ that agrees with $\psi$ on $[-\pi/2,\pi/2]$, as 
depicted at the top of Figure~\ref{fig:phi0_phi1} for a generic $\psi$ shown in Figure \ref{fig:generic_psi}.
On the other hand, $\varphi_1$ is periodic with period $2\pi$: 
its graph is obtained from that of $\varphi_0$ by reflecting about the
$\zeta$-axis every second copy of $\psi$, as depicted at the bottom of
Figure~\ref{fig:phi0_phi1}.

Applying Corollary~\ref{cor:psf} with $T=2$ and $s \in \{0,1\}$, 
we get that $\varphi_0$ and $\varphi_1$ are characteristic functions 
of rvs supported within the even and odd integers, respectively. 
The pmf corresponding to $\varphi_0$ is given by
\begin{equation}
p_0(k ) = \begin{cases}
2f(k) & \text{ if $k$ is an even integer} \\
0 & \text{ otherwise}.
\end{cases}
\label{eq:p0}
\end{equation}
and that corresponding to $\varphi_1$ is
\begin{equation}
p_1(k) = \begin{cases}
2f(k) & \text{ if $k$ is an odd integer} \\
0 & \text{ otherwise}.
\end{cases}
\label{eq:p1}
\end{equation}
From (\ref{eq:p})--(\ref{eq:p1}), we have $p(k) = \frac12 (p_0(k) + p_1(k))$ for all $k \in \Z$.

Finally, note that since $\varphi_0(t)$ and $\varphi_1(t)$ differ from
$\varphi(t)$ only when $\varphi(t) = 0$, we have
\begin{equation}
\varphi^2 = \varphi \varphi_0 = \varphi \varphi_1.
\label{eq:phi}
\end{equation}

\markchange{With these facts in hand, we can describe the construction of $\Z$-valued rvs $U$ and $V$ satisfying properties (S1)--(S3). Set $p_{U|0} = p_{V|0} = p_0$ and $p_{U|1} = p_{V|1} = p_1$. This
implies that $p_U = p_V = p$, where $p$ is as defined in (\ref{eq:p}).
Clearly, (\ref{eq:Z3}) holds. To verify (\ref{eq:Z2}), note that, by virtue of
(\ref{eq:phi}), we have for $a \in \{0,1\}$,
$$
\varphi_U \varphi_V = \varphi^2 = \varphi \varphi_a.
$$
But, by construction,  $\varphi_U \varphi_{V|a} =
\varphi_{V}\varphi_{U|a} = \varphi \varphi_a$.
Therefore, by Lemma~\ref{basic_lemma}, the rvs
$(U,V,X,Y)$ with joint pmf given by (\ref{eq:Z1}) have the
properties (S1)--(S3).}

\markchange{Recall from the discussion following Proposition~\ref{prop:finite_supp} that we need the rvs $U$ and $V$ to have finite variance. To ensure this, we use the fact \cite[pp.~512--513]{Feller} that a probability distribution $F$ with characteristic function $\chi$ has finite variance iff $\chi$ is twice differentiable; in this case, $\chi'(0) = i\mu$ and $\chi''(0) = -\mu_2$, where $\mu$ and $\mu_2$ are the mean and second moment of $F$. Thus, the rvs $U$ and $V$ (with pmf $p$ as above) have finite variance iff the characteristic function $\varphi$ is twice differentiable. In this case, as $\varphi$ is real, so is $\varphi'(0)$, which implies that $U$ and $V$ have zero mean. Hence, their variances are equal to their second moments, and so, $\text{Var}(U) = \text{Var}(V) = - \varphi''(0)$. By construction, $\varphi$ is twice differentiable iff $\psi$ is twice differentiable and $\varphi''(0) = \psi''(0)$. We summarize our construction of the rvs $U$ and $V$ in the following theorem.}

\begin{theorem}
Let $X,Y$ be iid Bernoulli$(1/2)$ rvs.
Suppose that we are given a probability density function $f:\R \to \R^+$ with a non-negative real characteristic function $\psi$ such that $\psi(t) = 0$ for $|t| \ge \pi/2$. Set $p_{U|0} = p_{V|0} = p_0$ and $p_{U|1} = p_{V|1} = p_1$, where $p_0$ and $p_1$ are as in (\ref{eq:p0}) and \ref{eq:p1}). Then, the resulting $\Z$-valued rvs $U$ and $V$ satisfy properties (S1)--(S3). Additionally,  the rvs $U$ and $V$ have finite variance iff $\psi$ is twice differentiable, in which case the variance equals $-\psi''(0)$.
\label{mod2_thm}
\end{theorem}

Based on Theorem~\ref{mod2_thm}, secure computation of XOR at the relay works as follows: the nodes $\mathtt{A}$ and $\mathtt{B}$ modulate their bits independently to an integer $k$, with probability $p_0(k)$ (from (\ref{eq:p0})) if the bit is 0, or with probability $p_1(k)$ (from (\ref{eq:p1})) if the bit is 1. The probability distributions can be chosen such that the modulated symbols have finite average power. The average transmit power is equal to the variance of the modulated random variable, which is $-\psi''(0)$, and a handle on this can be obtained by choosing $\psi$ carefully. The relay receives the sum of the two integers, which is independent of the individual bits $X$ and $Y$ (of $\mathtt{A}$ and $\mathtt{B}$ respectively). However, the XOR of the two bits can be recovered at $\mathtt{R}$ with probability 1. This is done by simply mapping the received integer $W$ to $1$, if $W$ is odd, and $0$ if $W$ is even. To gain a better understanding of the construction of the rvs, let us see an example.

\begin{example}
\markchange{Consider the density (from \cite[Section~XV.2, Table~1]{Feller})
\begin{equation}
f(x) = \begin{cases} \frac1{2\pi} & \text{ if } x = 0 \\ 
\frac{1-\cos x}{\pi x^2} & \text{ if } x \neq 0
\end{cases}
\label{eq:f2}
\end{equation}
which has characteristic function 
\begin{equation}
\hat{f}(t) = \max\{0, 1-|t|\}
\label{eq:fhat}
\end{equation}
The function $\hat{f}$ is plotted in Figure~\ref{fig:triangle}.
In particular, $\hat{f}(t) = 0$ for $|t| \ge 1$.}

 \begin{figure}[t]
 \centerline{\resizebox{4.5cm}{!}{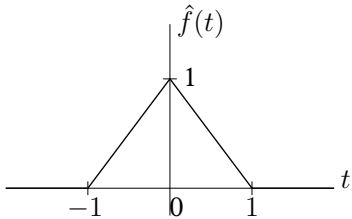}}
 \caption{$\hat{f}(t) = \max\{0, 1-|t|\}$.}
 \label{fig:triangle}
 \end{figure}

\markchange{The function $\hat{f}$ is compactly supported but it is not differentiable at $0$. This can be rectified by considering instead $g = \hat{f} \ast \hat{f}$, where $\ast$ denotes convolution, which can be explicitly computed to be
\begin{equation}
g(t) = (\hat{f}\ast\hat{f})(t) = \begin{cases}
\frac12 |t|^3 - t^2 + \frac23 & \text{ if } |t| \le 1 \\
\frac16 (2-|t|)^3 & \text{ if } 1 \le |t| \le 2 \\
0 & \text{ otherwise}
\end{cases}
\label{eq:g}
\end{equation}
}

\markchange{Now, define $h(x) := (3\pi^2/4) \, [f(\pi x/ 4)]^2$, with $f$ as in (\ref{eq:f2}). We prove in Appendix~A that $h$ is a probability density function whose characteristic function is given by 
$$
\psi(t) = {\textstyle \frac32 \, g(\frac{4t}{\pi})},
$$
where $g$ is as in (\ref{eq:g}).
It can be directly verified that $\psi$ is non-negative with $\psi(t) = 0$ for 
$|t| \ge \pi/2$, and that  $\psi$ is twice differentiable, with $\psi''(0) = -48/\pi^2$.}

\markchange{Thus, rvs $U$ and $V$ can be constructed as in
 Theorem~\ref{mod2_thm} with $\text{var}(U) = \text{var}(V) = 48/\pi^2$.}

\label{ex:psi}
\end{example}

\smallskip

 \begin{remark}
 It is even possible to construct compactly supported $C^\infty$
 characteristic functions. Constructions of such functions are given in \cite{RS04}.
 In fact, \cite{RS04} constructs compactly supported characteristic
 functions $\psi$ such that the corresponding density functions $f$
 are even functions satisfying $\lim_{x \to \infty} x^m f(x) = 0$ for
 all $m > 0$. This implies that all the absolute moments 
 $\int_{-\infty}^{\infty} |x|^m f(x) \, dx$ exist, and hence,
 $\psi$ is a $C^\infty$ function (see \cite[p.~512]{Feller}). If such a characteristic function $\psi$ is used in the
 construction described in Theorem~\ref{mod2_thm}, then
 the resulting $\Z$-valued rvs $U,V$ will have pmfs 
 $p_U(k),p_V(k)$ whose tails decay faster than any polynomial in $k$. 
 To be precise, 
 $\lim_{k\to\infty} k^m p_U(k) = \lim_{k\to\infty} k^m p_V(k) = 0$ for
 any $m > 0$. 
 \label{rem:Cinfty}
 \end{remark}

The above remark shows that we can have $\Z$-valued rvs $U,V$ satisfying properties (S1)--(S3), with pmfs decaying faster than any polynomial. However, the rate of decay cannot be much faster than that. Indeed, it is not possible to construct $\Z$-valued rvs with exponentially decaying pmfs that satisfy properties (S1)--(S3). Define a pmf $p(k)$, $k \in \Z$, to be \emph{light-tailed} if there are positive constants $C$ and $\lambda$ such that $p(k) \le C \lambda^{-|k|}$ for all sufficiently large $|k|$.

\begin{proposition}
Properties (S1)--(S3) cannot be satisfied by integer-valued rvs $U,V$ having light-tailed pmfs.
\label{prop:exp_decay}
\end{proposition}
\emph{Proof.}\footnote{This proof was conveyed to the authors by Manjunath Krishnapur.} 
Suppose that $U,V$ are $\Z$-valued rvs satisfying (S1) and (S2).
Using $\varphi_U = (1/2)(\varphi_{U|0} + \varphi_{U|1})$ and $\varphi_V = (1/2)(\varphi_{V|0} + \varphi_{V|1})$
in (\ref{eq:Z2}), we readily obtain
\begin{equation}
\varphi_{U|0}^2 = \varphi_{U|1}^2 \ \text{ and } \ \varphi_{V|0}^2 = \varphi_{V|1}^2. 
\label{eq:expdecay1}
\end{equation}

If $U,V$ have light-tailed pmfs, then $p_{U|a}$ and $p_{V|a}$, $a \in \{0,1\}$, must also be light-tailed, since $p_{U|a} \le 2p_U$ and $p_{V|a} \le 2p_V$. The key observation is that the characteristic function of a light-tailed pmf is real-analytic, i.e., it has a power series expansion $\sum_{n=0}^\infty c_n t^n$, with $c_n \in \C$, that is valid for all $t \in \R$ \cite[Chapter~7]{Lukacs}. Thus, $\varphi_{U|a}$ and $\varphi_{V|a}$, for $a \in \{0,1\}$, are real-analytic. It \markchange{follows} by comparing power series coefficients, that if functions $g$ and $h$ are real-analytic and $g^2 = h^2$, then either $g = h$ or $g = -h$. Applying this to (\ref{eq:expdecay1}), we find that $\varphi_{U|0} = \pm \varphi_{U|1}$, and similarly for $V$. In fact, since $\varphi_U$ and $\varphi_V$ cannot be identically $0$, we actually have $\varphi_{U|0} = \varphi_{U|1} = \varphi_U$, and similarly for $V$. This implies that $U \independent X$ and $V \independent Y$. From this, and (S1), we obtain that $U + V \independent X \oplus Y$, thus precluding (S3). \endproof






\subsection{Extension to Finite Abelian Groups\label{sec:group}}

A close look at the modulations in the previous section reveals the following structure: we had a fine lattice $\Lf=\Z$ and a coarse lattice $\Lc=2\Z$, with the quotient group $\Lf/\Lc$, consisting of the two cosets $2\Z$ and $1+2\Z$, making up the probabilistically-chosen modulation alphabet. Given a message $X\in\Lf/\Lc$, the encoder outputs a random point from the coset $X$ according to a carefully chosen probability distribution. 
  Note that the quotient group in this case is isomorphic to $\Z_2$, and this enables recovery of the XOR of the bits (addition in $\Z_2$) from integer addition of \markchange{the} transmitted symbols modulo the coarse lattice. Also, the choice of the probability distribution (from Theorem \ref{mod2_thm}) ensures that the choice of coset at each transmitter is independent of the integer sum at the relay. We shall extend the construction described in the previous subsection to $d$ dimensions, thereby obtaining a scheme that satisfies properties~(S1)--(S3).

Now, any finite Abelian group $\Gp$ can be expressed as the quotient group $\Lf/\Lc$ for some pair of nested lattices $\Lc \subseteq \Lf$. Indeed, any such $\Gp$  is isomorphic to a direct sum of cyclic groups: $\Gp \cong \Z_{N_1} \oplus \Z_{N_2} \oplus \cdots \oplus \Z_{N_k}$ for some positive integers $N_1,N_2,\ldots,N_k$ \cite[Theorem~2.14.1]{Herstein}. Here, $\Z_{N_j}$ denotes the group of integers modulo-$N_j$. Taking $\Lf = \Z^d$ and $\Lc = \mathsf{A}^T \, \Z^d$, where $\mathsf{A}$ is the diagonal matrix $\text{diag}(N_1,N_2,\ldots,N_k)$, we have $\Gp \cong \Lf/\Lc$.
So, the finite Abelian group case is equivalent to considering the quotient group, i.e., the group of cosets, of a coarse lattice $\Lc$ within a fine lattice $\Lf$. These lattices may be taken to be full-rank lattices in $\R^d$. 

As an example, let $N\ge 2$ be an integer, and let $\Z_N = \{0,1,\ldots,N-1\}$ denote the set of
 integers modulo $N$. Let $X,Y$ be iid random variables uniformly
 distributed over $\Z_N$, and let $X \oplus Y$ now denote their modulo-$N$
 sum. Similar to the binary case discussed so far, given a non-negative real characteristic function $\psi$ such that
 $\psi(t) = 0$ for $|t| \ge \pi/N$, we can construct $\Z$-valued 
 random variables $U,V$, jointly distributed with $X,Y$,
 for which properties (S1)--(S3) hold. 
In this case, the finite Abelian group can be taken as the group of cosets 
of the coarse lattice $N\Z$ within the fine lattice $\Z$, which is isomorphic to $\Z_N$.
 

Let $\Lc$ be a sublattice of $\Lf$ of index $M$ (i.e., the number of cosets of $\Lc$ in $\Lf$ is $M$). List the cosets of $\Lc$ in $\Lf$ as $\L_0,\L_1,\ldots,\L_{M-1}$, which constitute the quotient group $\Gp = \Lf/\Lc$.
As before, $\oplus$ denotes addition within $\Gp$. 

Consider rvs $X,Y$ uniformly distributed over $\Gp$. We wish to construct rvs $U,V$ taking values in $\Lf$, having the properties (S1)--(S3). The following theorem shows that this is possible. Here, $\R^{+}$ denotes the set of all non-negative real numbers.

\begin{theorem}
Suppose that $\psi:\R^d \to \R^+$ is the characteristic function of a probability density function $f: \R^d \to \R^+$, such that $\psi(\t) = 0$ for $\t \notin \cV(\hLc)$, where $\hLc$ is the Fourier dual of $\Lc$. For $j = 0,1,\ldots,M-1$, define the pmf $p_j$ as follows: 
\begin{equation}
p_j(\k) = \begin{cases} 
|\det \L_0| f(\k) & \text{ if } \k \in \L_j \\
0 & \text{ otherwise}.
\end{cases}
\label{eq:pj_perfect}
\end{equation}
 Finally, define a random variable $U$ (resp.\ $V$) jointly distributed with $X$ (resp.\ $Y$) as follows: if $X = \L_j$ (resp.\ $Y = \L_j$), $U$ (resp.\ $V$) is a random point from $\L_j$ picked according to the distribution $p_j$. 
Then, the resulting $\Lf$-valued rvs $U,V$ satisfy properties (S1)--(S3). \markchange{Additionally, $\mathbb{E}{\|U\|}^2$ and $\mathbb{E}{\|V\|}^2$ are finite iff $\psi$ is twice differentiable at $\0$, in which case $\mathbb{E}{\|U\|}^2 = \mathbb{E}{\|V\|}^2 = -\Delta\psi(\0)$, where $\Delta = \sum_{j=1}^d \partial_j^2$ is the Laplacian operator.}
\label{thm:L_thm}
\end{theorem}
As with Theorem~\ref{mod2_thm} and XOR, the above theorem allows for secure computation at the relay of the group operation $X \oplus Y$. The theorem is proved using Proposition~\ref{prop:psf_Rd}, in a manner completely analogous to Theorem~\ref{mod2_thm}. The interested reader is directed to Appendix~B for the proof.

Constructing compactly supported twice-differentiable (or even
$C^\infty$) characteristic functions $\psi:\R^d \to \R^+$, $d \ge 1$,
is straightforward, given our previous constructions of such functions 
from $\R$ to $\R^+$.  
Suppose that for $i = 1,2,\ldots,d$,
$\psi_i:\R \to \R^+$ is the characteristic function of a random
variable $X_i$, such that $\psi_i(t) = 0$ for $|t| \ge \lambda_i$, 
with $\lambda_i > 0$, and $X_1,X_2,\ldots,X_d$ are mutually independent. Then, $\psi(t_1,\ldots,t_d) = \prod_{i=1}^d
\psi_i(t_i)$ is the characteristic function of the random vector
$\mathbf{X} = (X_1,\ldots,X_d)$. Note that $\psi$ is compactly
supported: $\psi(\t) = 0$ for $\t \notin \prod_{i=1}^d
(-\lambda_i,\lambda_i)$. Moreover, if the $\psi_i$s are
twice-differentiable (or $C^\infty$) for all $i$, then so is $\psi$. \markchange{Constructions other than product constructions are also in abundance; see e.g., \cite{EGR04}, \cite{RS04} and Theorem~\ref{thm:minvar} below. A smooth, compactly supported characteristic function in $\R^2$ is depicted in Figure~\ref{fig:2dcharfunc}.}

\begin{figure}
     \centering{\includegraphics[width=13cm]{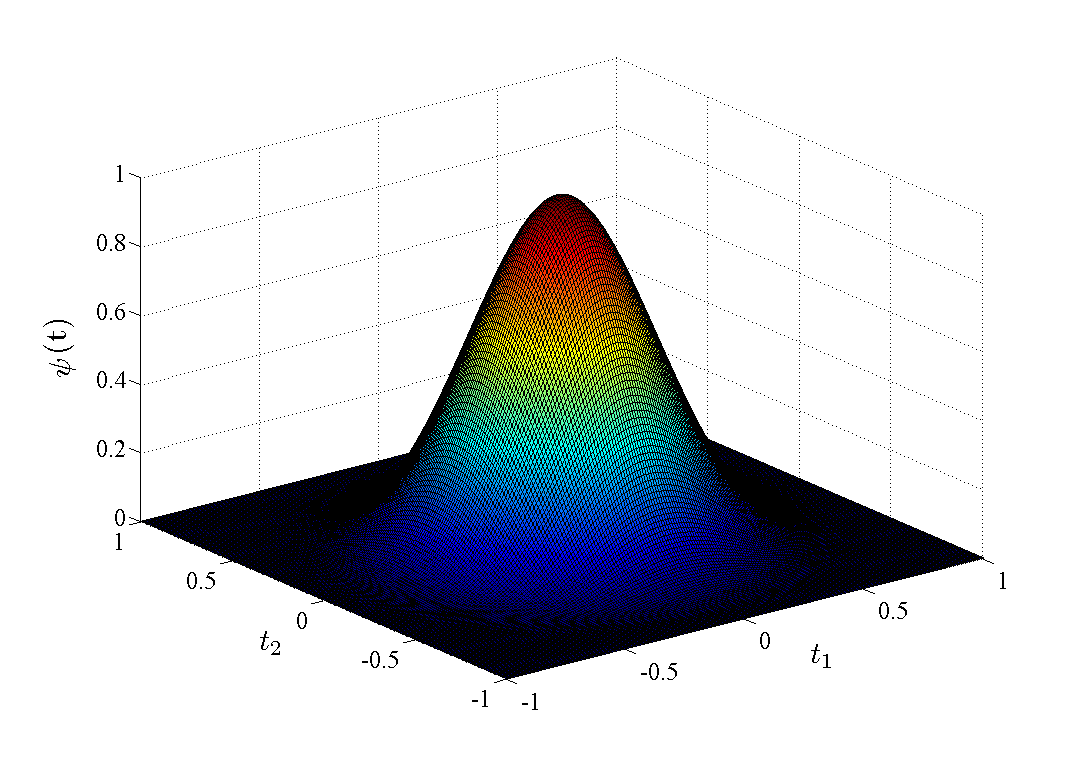}}
     \caption{Example of a characteristic function supported within $\cV(2\Z^2)$.}
     \label{fig:2dcharfunc}
\end{figure}

Our objective is to design codes (as defined in Definition~\ref{def:code}) for secure computation at the relay.
With the construction described above, the rate of the code depends on the number of cosets, $M$, of $\Lc$ in $\Lf$. 
For a given average power constraint, the system designer is usually faced with the task of maximizing the rate.
Equivalently, for a given rate, the average transmit power must be kept as small as possible. The transmit power is equal to the second moment of $\U$ (or $\V$). 
Therefore, while any characteristic function $\psi$ supported within
$\cV(\hLc)$ suffices for the construction of
Theorem~\ref{thm:L_thm}, we must use a $\psi$ for which 
$-\Delta \psi(\0)$ is the least among such $\psi$'s. This would yield random variables
$U$ and $V$ of least second moment (and hence least transmit power), and having the desired properties.

It is evident that by simply scaling the nested lattice pair, the average transmit power may be made as small as required. Suppose that the random vectors $\U$ and $\V$, distributed over a fine lattice $\Lf$, have second moment $P$. Then, for any $\alpha>0$, the random variables $\U'=\alpha \U$ and $\V'=\alpha \V$, distributed over $\alpha\Lf:=\{\alpha \z: \z\in\Lf \}$  have second moment $\alpha^{2}P$. Choosing a small enough $\alpha$ would suffice to satisfy the power constraint. However, as we will see in the following sections, when we have to deal with the additive noise in the MAC channel, it is not possible to scale down the lattice arbitrarily if the probability of error is to be made small. Also, for a given (fixed) coarse lattice, it turns out that the second moment (which depends solely on the choice of $\psi$) cannot be made arbitrarily
small. Indeed, the following result, adapted from \cite{EGR04},
gives a precise and complete answer to the question of how small
$-\Delta \psi(\0)$ can be for a characteristic function $\psi$
supported within a ball of radius $\rho$ in $\R^d$.

\begin{theorem}[\cite{EGR04}, Theorem 5.1]
Fix a $\rho>0$. If $\psi$ is a characteristic function of a random variable distributed over $\R^{d}$ such that $\psi(\t) = 0$ for $\Vert\t\Vert \ge \rho$, then
 \begin{equation}
 - \Delta \psi(\0) \ge \frac{4}{{\rho}^2} \, j^{2}_{\frac{d-2}{2}},
 \label{eq:wpsi}
 \end{equation}
 with equality iff $\psi(\t) = \wpsi(\t/\rho)$ for $\wpsi = \omega_{d} \tilde{*} \omega_{d} $. \markchange{Here, $ j_{k} $ denotes the first positive zero of the Bessel function  $ J_{k} $. Also,} $ \omega_{d}(\t) = \gamma_{d} \, \Omega_{d}(2\Vert \t \Vert j_{\frac{d-2}{2}})  $ for $ \Vert \t \Vert \leq 1/2 $ and $ \omega_{d}(\t)=0 $ for $\Vert\t\Vert>1/2 $, and 
 	$$
 	\omega_{d}\tilde{*}\omega_{d}(\t) = \int \omega_{d}(\tbf)\overline{\omega_{d}(\t+\tbf)} \, d\tbf
 	$$
 	 denotes the folded-over self convolution of $ \omega_{d} $,
          with $ \overline{\omega_{d}(\t)} $ denoting the complex conjugate of $ \omega_{d}(\t) $.  Furthermore, for $t \in \R$,
 $$
 	\Omega_{d}(t) = \Gamma(d/2) \Big( \frac{2}{t} \Big)^{\frac{d-2}{2}} J_{\frac{d-2}{2}}(t)
 $$
 and
$$
 	\gamma_{d}^{2} = \frac{4j_{\frac{d-2}{2}}^{d-2}}{\pi^{d/2} \Gamma(d/2) J^2_{\frac{d}{2}}(j_{\frac{d-2}{2}})},
 $$
 where $\Gamma(\cdot) $ denotes the Gamma function. 
        The density $f$ corresponding to the minimum-variance $\psi$ is given by $f(\x) = \rho^d \tilde{f}(\rho\x)$,
 where 
 	\begin{equation}
 	\tilde{f}(\x) = c_{d}\left( \frac{\Omega_{d}(\Vert \x \Vert / 2)}{j_{\frac{d-2}{2}}^{2} - (\Vert \x \Vert / 2)^{2} } \right)^2,
 	\label{eq:minvar_pdf}
 	\end{equation}
 where 
  \markchange{$$
 	c_{d} = \frac{4j_{\frac{d-2}{2}}^{2}}{4^{d} \pi^{d/2} \Gamma(d/2)}.
 $$}
  	\label{thm:minvar}
 \end{theorem}

\begin{remark}
     Observe that Theorem~\ref{thm:L_thm} is true for any nested lattice pair $(\L,\L_0)$. As long as $\psi(\t)$ is a characteristic function supported within \markchange{$\cV(\hat{\L}_0)$},
we have an encoding scheme that satisfies (S1)--(S3). If we restrict $\psi$ to be supported within a ball of radius $\rho$, which is contained within $\cV(\hat{\L}_0)$, then Theorem~\ref{thm:minvar}
gives us a suitable candidate for $\psi$ that can be used to obtain perfect secrecy. Since we are interested in minimizing the transmission power, we can choose $\rho$ to be as large as $\rpack(\hat{\L}_0)$,
where $\rpack(\hat{\L}_0)$ denotes the packing radius of $\hat{\L}_0$. Hence, we now have a coding scheme that achieves perfect secrecy for any arbitrary nested lattice pair.
This is rather interesting, since earlier work on weak and strong secrecy using lattices~\cite{HeYener,HeYenerstrong,Ling13} invariably required that the nested lattices satisfy certain goodness properties.
\markchange{Therefore, ours is an explicit scheme which specifies, for any nested lattice pair, a distribution to be used for randomization at the encoder in order to obtain perfect secrecy. 
In particular, our randomization scheme can also be used in conjunction with ``practical'' lattice coding schemes (e.g., \cite{diPietro13,Sommer,Yan13}) that have low decoding complexity.} 
\end{remark}

\section{The Gaussian Noise Setting}\label{sec:gnoise}

Given any nested lattice pair, we now have a scheme whereby the relay can compute $X\oplus Y$ from $\U+\V$, but cannot determine $X$ or $Y$ separately.  We next consider the scenario where the symbols received by the relay are corrupted by noise, and prove the achievability of the power-rate pairs described in Theorem~\ref{thm:perfect_main}. Recall that in the MAC phase, the relay receives
\[
 \W=\U+\V+\bZ,
\]
 where $\bZ$ is zero-mean iid Gaussian noise with variance $\sigma^{2}$. The coding scheme that we use is largely based on the work in~\cite{Erez04,Nazer11}, and is described below.

\subsection{Coding Scheme for Perfect Secrecy} \label{sec:scheme_perfect}
We now describe the sequence of $(d,M^{(d)})$ (recall Definition~\ref{def:code}) codes that achieve perfect secrecy.

\noindent {\underline{Code}}: A $(\L^{(d)},\L^{(d)}_0)$ \emph{nested lattice code} consists of a pair of full-rank nested lattices $\L^{(d)}_0 \subseteq \L^{(d)}$ in $\R^d$. The messages are chosen from the group $\Gp^{(d)} = \L^{(d)}/\L^{(d)}_0$, whose $M^{(d)}:= |\L^{(d)}/\L^{(d)}_0|$ elements are listed as $\L_0,\L_1,\ldots,\L_{M^{(d)}-1}$.

\noindent {\underline{Encoding}}: We have messages $X,Y$ at nodes $\mathtt{A},\mathtt{B}$ that are independent rvs, uniformly distributed over $\Gp^{(d)}$.
 We first pick a characteristic function $\psi$ supported within $\cV(\hat{\L}^{(d)}_0)$, 
as needed in Theorem~\ref{thm:L_thm}. 
We impose the restriction that $\psi$ be supported within a ball centered 
at $\0$ with radius equal to the packing radius, $\rpack(\hat{\L}^{(d)}_0)$, 
of the dual lattice $\hat{\L}^{(d)}_0$. Recall that the packing radius is, 
by definition, the largest radius of a ball centered at $\0$ 
that is contained within $\cV(\hat{\L}^{(d)}_0)$. So, if
 $\psi(\t) = 0$ for $\|\t\| \ge r_{\text{pack}}(\hat{\L}_0)$, 
then $\psi(\t)$ is certainly supported within $\cV(\hat{\L}_0)$. 
If $X = \L_j$, node $\mathtt{A}$ transmits a random vector $\U \in \L_j$
 picked according to the distribution $p_j$ of Theorem~\ref{thm:L_thm}. 
Similarly, if $Y = \L_k$, node $\mathtt{B}$ transmits a random
 vector $\V \in \L_k$ picked according to the distribution $p_k$. The rate of 
transmission from $\mathtt{A}$ or 
$\mathtt{B}$ is $R^{(d)}=\frac{1}{d} \log_2 M^{(d)}$. 
The average transmit power per dimension at each node is $P^{(d)} = \frac{-\Delta \psi(\0)}{d}$, as
 in Theorem~\ref{thm:L_thm}. 

From Theorem~\ref{thm:minvar}, we see that an 
average transmit power per dimension as low as 
\begin{equation}
P^{(d)} = \frac{4j_{\frac{d-2}{2}}^2}{d\left(\rpack(\hLcd)\right)^2}, 
\label{eq:Pd}
\end{equation}
 is achievable by a suitable choice of $\psi$. It was shown in \cite{BESSEL2} (see also \cite{BESSEL1}) that the first positive zero of the Bessel function $J_k$ can be written as $j_k = k+ bk^{1/3} + \mathcal{O}(k^{-1/3})$, where $b$ is a constant independent of $k$. Therefore,  
\begin{equation}
  \label{eq:power}
P^{(d)} =  \frac{d}{{\rpack}^2(\hat{\L}^{(d)}_0)}(1+o_d(1)),
\end{equation}
where $o_d(1)\to0$ as $d\to\infty$, is achievable by a suitable choice of $\psi$ using Theorem \ref{thm:minvar}.


\begin{figure}[t]
 \centerline{\resizebox{10cm}{!}{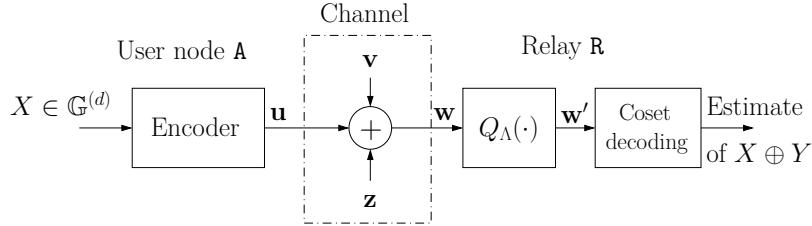}}
 \caption{The operations performed by the user nodes and the relay.}
 \label{fig:relayop}
\end{figure}

\noindent {\underline{Decoding}}: The relay $\mathtt{R}$ receives $\W = \U + \V + \bZ$, where $\bZ$ is a Gaussian noise vector with $d$ independent $\mathcal{N}(0,\sigma^2)$ components, which are all independent of $\U$ and $\V$. The relay estimates $\L_j \oplus \L_k$ to be the coset of $\Lcd$ represented by \markchange{$Q_{\Lfd}(\W)$,} the closest vector to $\W$ in the lattice $\Lfd$. \markchange{The decoder mapping is denoted by $\cD(\cdot)$.}
 


\noindent {\underline{Security}}: Since the noise $\bZ$ is independent of everything else, Theorem~\ref{thm:L_thm} shows that $\W$ is independent of the individual messages $X,Y$. Hence, even in the noisy setting, perfect security continues to be guaranteed at the relay for any choice of the nested lattice code. 
\markchange{It is worth reiterating that perfect secrecy can be guaranteed irrespective of the noise $\bZ$. The distribution of $\bZ$ only determines the reliability of decoding, which in turn influences the power-rate pairs achievable with perfect secrecy.}

\noindent {\underline{Reliability and achievable power-rate pairs}}: Let $\Pe$ denote the average probability that $Q_\L(\W)$ is different from the coset to which $\U+\V$ belongs. From Definition~\ref{def:pr_pairsperfect}, a pair $(\cP,\cR)$ is achievable if for every $\delta>0$, there exists a sequence of nested lattice codes $(\Lf^{(d)},\Lc^{(d)})$ for which the following hold for sufficiently large $d$: $\rte^{(d)}> \cR-\delta$, $\pow^{(d)}<\cP+\delta$ and $\Pe<\delta$. 

\markchange{For a given nested lattice pair, Theorem~\ref{thm:minvar} gives us the minimum average transmit power per dimension that guarantees perfect secrecy (subject to the condition that the characteristic function is supported within a ball of radius $\rpack(\hLcd)$), and the pmf $p_j$ that achieves the minimum. The choice of the nested lattices affects the reliability of decoding $X\oplus Y$ at the relay, and consequently determines achievable transmission rates. 
}
To guarantee secure and reliable computation at the relay, we restrict the class of  nested lattice pairs $( \Lfd,\Lcd)$  to those which satisfy the following \markchange{``goodness''} properties\footnote{For definitions of lattices good for covering, packing, and AWGN channel coding, the reader is directed to Appendix~C.}:
\begin{enumerate}
 \item[$(G_1)$] The sequence of coarse lattices, $\{\Lcd\}$, is good for covering and AWGN channel coding.
 \item[$(G_2)$] The sequence of dual lattices, $\{\hLcd\}$, is good for packing.
 \item[$(G_3)$] The sequence of fine lattices, $\{\Lfd\}$, is good for AWGN channel coding.
\end{enumerate}
Unlike prior work on nested lattices~\cite{Agrawal09,Erez04,Nazer11,Nitinawarat12} which only required $\{\Lcd\}$ and $\{\Lfd\}$ to satisfy properties $(G_1)$ and $(G_3)$ above, we have the additional requirement  that the \markchange{ sequence of Fourier duals}, $\{\hLcd\}$ must be good for packing. \markchange{While it is well established that there exist nested lattices satisfying $(G_1)$ and $(G_3)$~\cite{Erez04,Erez05,Nazer11}, it turns out that the duals of most of these lattices also satisfy the goodness properties. In the next section, we will formally describe an ensemble of lattices, also studied in~\cite{Erez04,Nazer11}, and show that most of the lattices in this ensemble satisfy all the above properties.} 
 \subsection{Good Ensembles of Nested Lattices with Good Duals}\label{sec:dkq_ensemble}
\markchange{ Our description of the construction of the $(\Lfd,\Lcd)$ nested lattice codes is based on~\cite{Erez04,Nazer11}}.
 Let $d$ and $k$ be positive integers with $k\leq d$, and let $q$ be a prime number.  Let $\Z_q$ denote the field of integers modulo $q$.
 \markchange{The $(d,k,q)$ ensemble of lattices (in the terminology of \cite{Erez05}) is used in the construction. A lattice from the $(d,k,q)$ ensemble is sampled as follows: }
\markchange{
\begin{enumerate} 
\item Choose a $k\times d$ matrix  $\mathsf{G}$ with entries from $\Z_q$ uniformly at random. 
 Note that $\mathsf{G}$ need not be full-rank. However, the probability that $\mathsf{G}$ is full-rank goes to 1 as $(d-k)$ tends to $\infty $~\cite{Erez05}. The linear code over $\Z_q$ generated by $\mathsf{G}$ is denoted by $\cC(\mathsf{G})=\{(\mathsf{G}^{T}\y)\bmod q: \y \in \Z_{q}^{k}\}$.
\item Apply Construction A on the code $\cC(\mathsf{G})$. This is done as follows:
\begin{itemize} 
\item[($c_1$)] The codebook is  scaled so that the scaled codewords lie within the $d$-dimensional unit cube: $\cC'=(1/q)\cC(\mathsf{G})=\{(1/q)\x: \x\in \cC(\mathsf{G})\}$.
\item[($c_2$)] The lattice is obtained by tessellating the entire space, $\R^{d} $, with copies of $\cC'$, i.e., 
${\L}(\cC)=\cC'+\Z^{d} := \{ \mathbf{c} + \x : \mathbf{c}\in \cC', \x \in \Z^{d} \} $.
\end{itemize}
\end{enumerate}
}

From the construction, it is clear that $\Z^{d}$ is a sublattice of ${\L}(\cC)$. More detail regarding  Construction-A lattices can be found in~\cite{Conway}. \markchange{We would like to make note of one important property of these lattices: if the generator matrix of a Construction-A lattice $\L$ has rank $d$, then the effective radius of $\L$ is given by~\cite{Erez05}
\begin{equation}
  \reff(\L)=\left(  \frac{\Gamma\left(\frac{d}{2}+1\right)}{\pi^{d/2}q^k}\right)^{1/d}.
\label{eq:reff_constA}
\end{equation}
}


Choose a sequence of coarse lattices $\{ \Lcd\}$, each 
$\Lcd$ selected uniformly at random from the $(d,k,q)$ ensemble, where $k$ and $q$ may be functions of $d$ chosen beforehand. 
\markchange{For $d\in\{ 1,2,3,\ldots\}$,} let $\mathsf{A}^{(d)}$ be the generator matrix of the coarse lattice 
$\Lcd$. For this choice of 
$\{ \Lcd\}$, we construct another ensemble of lattices from which we pick the sequence of fine lattices $\{\Lfd\}$. This consists of two steps: 
\begin{itemize}
 \item[($f_1$)] Choose a sequence of lattices, $\{\tilde{\L}_{f}^{(d)}\}$, with each $\tilde{\L}_{f}^{(d)}$ coming  from the $(d,k_{1},q_1)$ ensemble of Construction-A lattices. 
As mentioned earlier, $\tilde{\L}_{f}^{(d)}$ contains $\Z^{d}$ as a sublattice. If the generator matrix of  $\tilde{\L}_{f}^{(d)}$ has full rank, then the number of cosets of $\Z^{d}$ in $\tilde{\L}_{f}^{(d)}$ is $q_{1}^{k_1}$.
 \item[($f_2$)] The lattice $\tilde{\L}_{f}^{(d)}$ is subjected to a linear transformation by the matrix $(\mathsf{A}^{(d)})^{T}$, to get
$\Lfd=(\mathsf{A}^{(d)})^{T}\tilde{\L}_{f}^{(d)}:= \{ (\mathsf{A}^{(d)})^{T}\y : \y\in \tilde{\L}_{f}^{(d)} \}$ .
\end{itemize}
We will call this ensemble of 
$(\Lfd,\Lcd)$
pairs as the \emph{$(d,k,q,k_1,q_1)$ ensemble}. The lattice pair can be scaled appropriately so as to satisfy the average power constraint. We have $M^{(d)}=|\Lfd/\Lcd|=q_{1}^{k_1}$ with probability tending to $1$ as $d-k$ tends to $\infty$~\cite{Nazer11}. Hence, the rate of the $(\Lfd,\Lcd)$ code  will be
\begin{equation}
 \rte^{(d)}=\frac{k_1}{d}\log_{2}(q_1).
\label{eq:rte_constA}
\end{equation}

 We choose 
\begin{equation}
      k=\beta_0d,\;\;\text{ and }\;\; k_1=\beta_1 d,
\label{eq:k_k1}
\end{equation}
 for some $0< \beta_0,\beta_1 < 1/2$, and $q$ and $q_1$ are prime numbers chosen such that
\begin{equation}
    \lim_{d\to\infty}\frac{d}{q_1} = 0, \;\;\text{ and }\;\;   r_{\text{min}}^{(0)}<\reff(\Lcd)<2r_{\text{min}}^{(0)} ,
\label{eq:rmin_constraints}
\end{equation}
for some $0<r_{\text{min}}^{(0)}<1/4$.
It is possible to choose primes that satisfy the above conditions, and we direct the interested reader to~\cite{Erez05} for the details. We then have the following lemma, which is proved in Appendix~D.
\begin{lemma}
 Let 
$(\Lfd,\Lcd)$
 be a nested lattice pair chosen uniformly at random from the $(d,k,q,k_1,q_1)$ ensemble, with the parameters $k,q,k_1,q_1$ chosen so as to satisfy (\ref{eq:k_k1}) and (\ref{eq:rmin_constraints}). Then, the probability that 
$(\Lfd,\Lcd)$
 satisfies $(G_1)$--$(G_3)$ tends to one as $d$ approaches infinity.
\label{lemma:goodlattice1}
\end{lemma}

\subsection{Achievable Rates} \label{sec:ach_rate}

We now find achievable transmission rates for reliable and secure computation of $X\oplus Y$ at the relay.
The analysis closely follows that in~\cite{Erez04,Nazer11, NazerProc}. 
\markchange{As defined in Section~\ref{sec:scheme_perfect},
let $\cD(\W)$ be the estimate of $X\oplus Y$ made by the relay; to be precise, $\cD(\W)$ is the coset of $\Lcd$ to which $Q_{\Lfd}(\W)$ belongs.}
This is the same as the coset represented by $ Q_{\Lfd}([\W]\bmod \Lcd)$.

\begin{figure}
     \begin{center}
          \resizebox{9cm}{!}{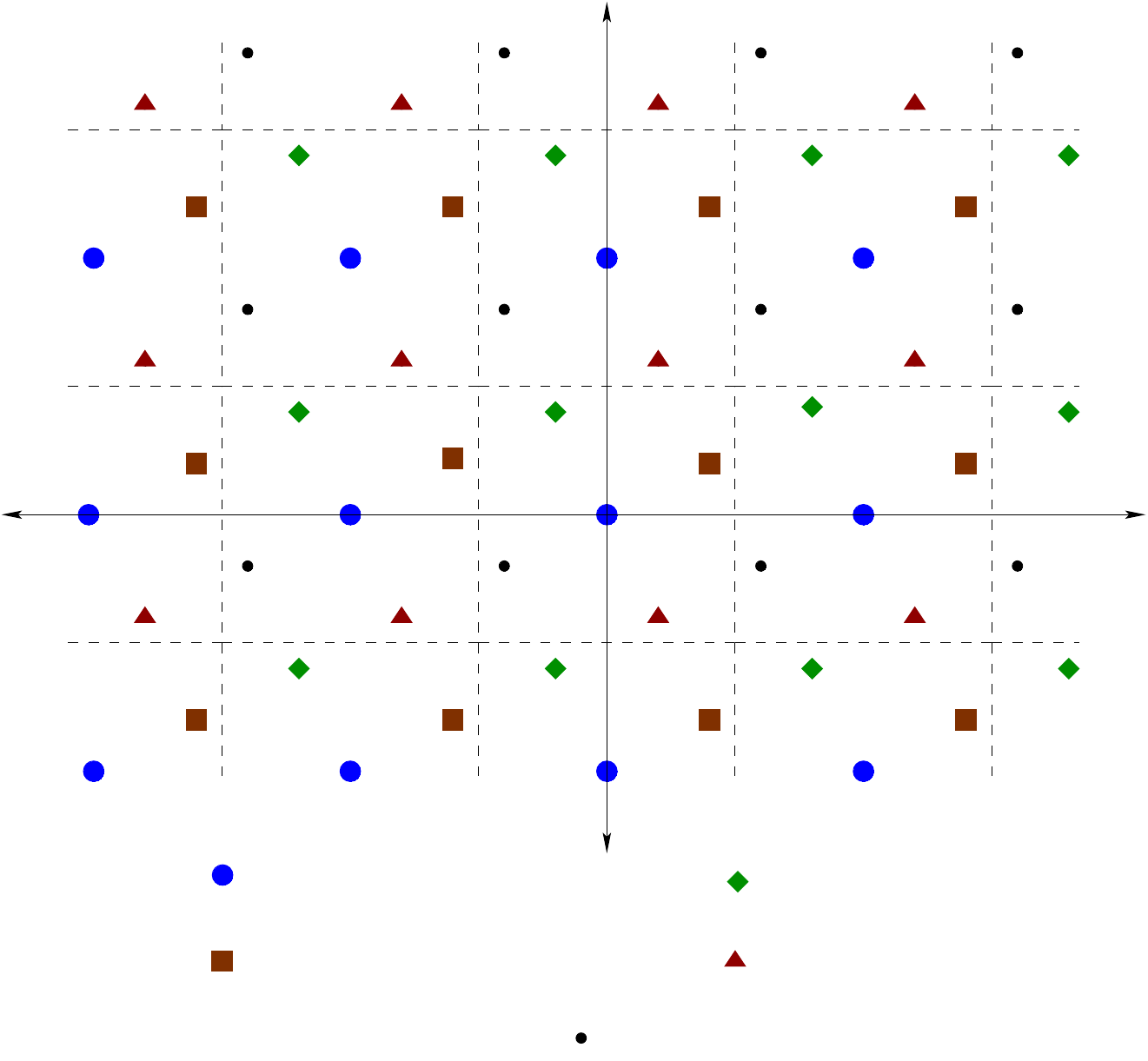}
	  \caption{Different cosets of $\L_0$ in $\L$. The coset representative of $\L_j$ within $\cV(\Lc)$ is $\lambda_j$.}
	  \label{fig:cosetrep}
     \end{center}
\end{figure}

Each lattice point in $\Lfd\cap \cV(\Lcd)$ is a coset representative for  a coset of $\Lcd$ in $\Lfd$.
This is illustrated in Fig.~\ref{fig:cosetrep}.
Suppose that $\L_j$ and $\L_k$ are the cosets which represent the messages $X$ and $Y$, respectively.
Let $ \mathbf{X} = [\U]\bmod\Lcd $ and 
$\mathbf{Y}=[\V]\bmod\Lcd$ be the coset representatives of $ \L_j $ and $ \L_k $, respectively. Then, $\L_j\oplus \L_k$ has  $[\mathbf{X}+\mathbf{Y}]\bmod \Lcd$ as its representative.
Therefore, the estimate $\cD(\W)$ has  $\widehat{\W}=[Q_{\Lfd}(\W)]\bmod \Lcd$ as its coset representative.
This is equal to $\widehat{\W}=[Q_{\Lfd}([\W]\bmod \Lcd)]\bmod \Lcd$. Let $\widetilde{\W}=[\W]\bmod \Lcd$. Then, $\widehat{\W}=[Q_{\Lfd}(\widetilde{\W})]\bmod \Lcd$.
As a consequence of the transmitter-receiver operations, the ``effective'' channel from $ \mathbf{X}, \mathbf{Y} $ to $ \widetilde{\W} $ can be written as follows~\cite{Nazer11}: 
	\begin{align} 
	\widetilde{\W} & = [\U+\V+ \bZ] \bmod \Lcd & \notag \\
	&		 = \left[ \left( [\U+\V] \bmod \Lcd \right) +\bZ\right]\bmod \Lcd & \notag \\
	&		 = \left[ \left( [\mathbf{X}+\mathbf{Y}] \bmod \Lcd \right) +\bZ\right]\bmod \Lcd . & \notag 
	\end{align} 

\begin{figure}[t]
 \begin{center}
  \resizebox{9cm}{!}{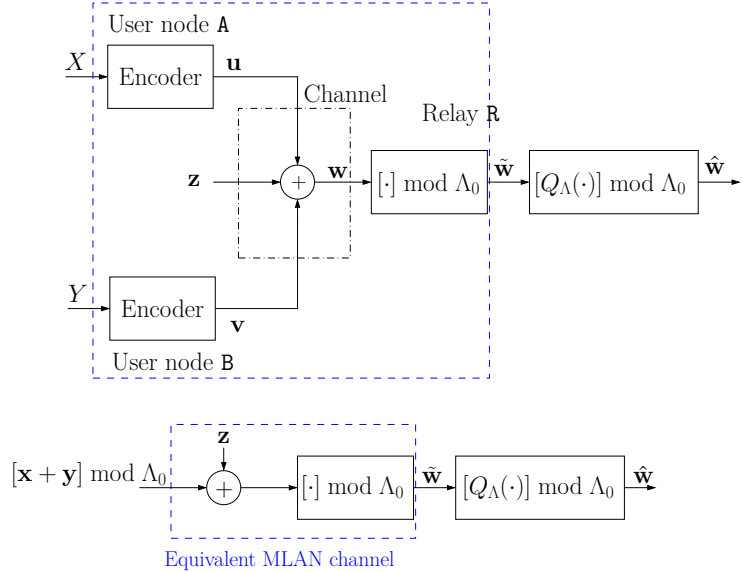}
 \end{center}
 \caption{MAC phase of the bidirectional relay and equivalent MLAN channel representation.}
 \label{fig:channelrep}
\end{figure}

	A channel of the form $ \W = [\mathbf{X} +\mathbf{N}]\bmod \Lcd $, where $\mathbf{N}$ denotes the noise vector, is called a $\Lcd$-modulo lattice additive noise ($\Lcd$-MLAN) channel~\cite{Erez04}. 
        The random variable $\widetilde{\W}$ behaves like the output of a point-to-point transmission over 
        a $\Lcd$-MLAN channel, with the transmitted vector being $ [\mathbf{X}+\mathbf{Y}]\bmod \Lcd $.
	Looking from $\widetilde{\W}$, the ``effective'' channel is a $\Lcd$-MLAN channel, and the relay has to decode $ [\mathbf{X}+\mathbf{Y}]\bmod \Lcd$ reliably from $\widetilde{\W}$.
        This is illustrated in Fig.~\ref{fig:channelrep}. 
       We will use the properties of the $\Lcd$-MLAN channel to determine achievable rate regions for our coding scheme. 

  We choose a sequence of nested lattice pairs that satisfy $(G_1)$--$(G_3)$, with each nested lattice pair coming from a $(d,k,q,k_1,q_1)$ ensemble, where $k,q,k_1$ and $q_1$ satisfy (\ref{eq:k_k1}) and (\ref{eq:rmin_constraints}). Using the coding scheme of Section~\ref{sec:scheme_perfect}, we can achieve perfect secrecy. \markchange{The proposition below provides us with the means of determining the rates achievable with this coding scheme.}

\markchange{
\begin{proposition}
Let $\textsf{M} > 0$ be a constant, and $\{ \Lfd,\Lcd \}$ be a sequence of nested lattice pairs that satisfy $(G_1)$--$(G_3)$, and scaled so as to satisfy $\reff(\Lcd)=\sqrt{d\textsf{M}}$. Then, using the coding scheme of Section~\ref{sec:scheme_perfect} with this sequence of nested lattice pairs, any rate less than $\frac{1}{2} \log_2\left(\frac{\textsf{M}}{\sigma^{2}}\right)$ is achievable with perfect secrecy.
\label{prop:rel}
\end{proposition}
}

\markchange{The proposition can be proved along the same lines as~\cite[Theorem 4]{Erez04}; we omit the details.}

\subsection{Relating Achievable Rates to Transmit Power} \label{sec:rate_power}

From (\ref{eq:power}), we know that as long as the average transmit power per dimension is less than $\left(d/\rpack^2(\hLcd)\right)(1+o_d(1))$, we 
can guarantee perfect secrecy at the relay. From Proposition~\ref{prop:rel}, we see that as long as the transmission rate is less than $\frac{1}{2}\log_2(\reff^2(\Lcd)/(d\nsvar))$,
the relay can reliably compute $X\oplus Y$ from $\W$. 
In order to achieve positive rates, we need $\reff(\Lcd)$ to grow at least as fast as $\sqrt{d}$, i.e., $\reff(\Lcd)=\Omega(\sqrt{d})$. 
Furthermore, to satisfy an average power constraint, we require $\rpack(\hLcd)=\Omega(\sqrt{d})$. The rate is an increasing function of $\reff(\Lcd)$, and the average transmit power  per dimension is a decreasing function of $\rpack(\hLcd)$. Since we want to maximize the rate for a given power constraint, we would like both $\reff(\Lcd)$ and $\rpack(\hLcd)$ to be as large as possible. However, for any lattice $\Lcd$, we have $\rcov(\Lcd)\rpack(\hLcd)\leq\pi d$~\cite[Theorem 18.3]{Barvinok}, and since \markchange{$\reff(\Lcd)\leq \rcov(\Lcd)$}, we get $\reff(\Lcd)\rpack(\hLcd)\leq\pi d$. Hence, to obtain positive rates and at the same time satisfy the power constraint, 
both $\reff(\Lcd)$ and $\rpack(\hLcd)$ must grow roughly as  $\sqrt{d}$. Therefore, we seek lattices satisfying properties $(G_1)$--$(G_3)$, for which the product $\reff(\Lcd)\rpack(\hLcd)$ is close to the upper bound of $\pi d$.


For a sequence of Construction-A coarse lattices satisfying $(G_1)$ and $(G_2)$, we can find an asymptotic lower bound for $(1/d)\reff(\Lcd)\rpack(\hLcd)$,\footnote{The product $\reff(\Lcd)\rpack(\hLcd)$ is invariant to scaling of $\Lcd$. This is because, for a constant $\alpha>0$, $\reff(\alpha\Lcd)=\alpha\reff(\Lcd)$, and if $\L'=\alpha\Lcd$, then the Fourier dual of $\L'$ is $(1/\alpha)\hLcd$.} as the following theorem shows.
\begin{lemma}
  Let $\{\Lcd\}$ be a sequence of coarse lattices, with each $\Lcd$ chosen from a $(d,k,q)$ ensemble and $k,q$ satisfying (\ref{eq:k_k1}) and (\ref{eq:rmin_constraints}). If $\{ \Lcd \}$ satisfies conditions $(G_1)$--$(G_2)$, then,
\begin{equation}
      \lim_{d\to\infty}\frac{\reff(\Lcd)\rpack(\hLcd)}{d}\ge \frac{1}{2e}.
\end{equation}
\label{lemma:goodlattice2}
\end{lemma}
\begin{proof}
     See Appendix~E.
\end{proof}

\subsection{Proof of Theorem~\ref{thm:perfect_main}}
Let us choose $\reff(\Lcd)= \frac{1}{2e}\sqrt{d\cP}$, for a constant \markchange{$\cP>4e^2\nsvar$}. Fix a $\delta>0$. Using Lemma~\ref{lemma:goodlattice2}, we see that 
\begin{equation}
     \rpack(\hLcd)\geq \frac{d}{2e\reff(\Lcd)}(1-o_d(1))\geq \frac{\sqrt{d}}{\sqrt{\cP}}(1-o_d(1)).  
     \label{eq:rpackbound1}
\end{equation} 
From (\ref{eq:power}), we see that perfect secrecy can be achieved with an average power constraint as low as $P^{(d)}=\left(d/\rpack^2(\hLcd)\right)(1+o_d(1))$. 
Combining this and (\ref{eq:rpackbound1}), perfect secrecy can be achieved with an average transmission power,
\begin{equation}
     P^{(d)}<\cP+\delta
     \label{eq:pdbound1}
\end{equation}
for all sufficiently large $d$.
From Proposition~\ref{prop:rel}, we have seen that the average probability of error can be made to go down to zero as long as 
\begin{equation}
     R^{(d)}<\cR:=\frac{1}{2}\log_2\frac{\cP}{(2e)^2\nsvar}.
     \label{eq:rdbound1}
\end{equation}
Therefore, for every $\delta>0$, we can choose a sequence of nested lattice codes such that for all sufficiently large $d$, we have $R^{(d)}>\cR-\delta$, $P^{(d)}<\cP+\delta$ and $\Pe<\delta$.
Hence, a power-rate pair of 
\[
      \left( \cP, \left[\frac{1}{2}\log_2\frac{\cP}{\nsvar}-\log_2 2e\right]^{+} \right)
\]
is achievable with perfect secrecy, concluding the proof of Theorem~\ref{thm:perfect_main}.\qed

\section{Strong Secrecy}\label{sec:strongsecr}

%


\markchange{A natural question that arises is what happens if we replace $f$ in Theorem~\ref{thm:L_thm} by a density function for which the support of the characteristic function goes beyond $\cV(\hLcd)$. Can we obtain different secrecy properties by simply changing the density $f$?} Specifically, let $\psi(\t)$ be a characteristic function which is supported within a ball of radius $\rho>\rpack(\hLcd)$, and choose the characteristic function $\phi_{U|X=\x}(\t)=\sum_{\n\in\hLcd}\psi(\t+\n)e^{-i\langle \n,\x\rangle}$. Clearly, we cannot expect perfect secrecy, but
can we at least obtain strong secrecy? Let us take $\psi$ to be the characteristic function of the minimum-variance distribution in (\ref{eq:minvar_pdf}), with the support of $\psi$ chosen to be a ball of radius $\rho=\text{min}\{\reff(\hLcd),2\rpack(\hLcd)\}$.\footnote{If we have $\rho>2\rpack(\hLcd)$, then $\sum_{\n\in\hLcd}\psi(\t+\n)e^{-i\langle \n,\x\rangle}$ would have to be normalized to make \markchange{it} a characteristic function, and this makes analysis more complicated.} Doing so would give us an improved rate of \markchange{$\left[\frac{1}{2}\log_2\frac{\cP}{\nsvar}-\log_2 e\right]^{+}$}.
However, for such a coding scheme, we are only able to show that the $\ell^2$ norm of the difference between $p_{U+V,X}$ and $p_{U+V}p_X$ goes to zero as $d\to\infty$.
Knowing only that the $\ell^2$ norm of the difference between $p_{U+V,X}$ and $p_{U+V}p_X$ goes to zero as $d\to\infty$, we cannot conclude whether strong secrecy is obtained. In fact, by itself, the $\ell^2$ norm is not a good measure of secrecy. In any case, we will use a different approach to obtaining strong secrecy, and show that an even higher transmission rate of 
\markchange{$\left[\frac{1}{2}\log_2\left(\frac{1}{2}+\frac{\cP}{\nsvar}\right)-\frac{1}{2}\log_2 2e\right]^{+}$} is achievable.

Instead of using distributions with \markchange{compactly} supported characteristic functions, we will use a sampled Gaussian density for randomization at the encoders. Such a scheme was used in context of the wiretap channel in~\cite{Ling13}. \markchange{We will show that if a Gaussian pdf is used instead of a density $f$ having a compactly supported characteristic function, then we can obtain strong secrecy. It is interesting to note that the same basic coding scheme, but with a different pdf used for randomization, can give different secrecy properties.}

\subsection{The Gaussian Density}
We now introduce some notation that will be used in the sequel. Let $\L$ be a lattice in $\R^d$. For any $\x\in \R^d$, and any $\kappa>0$, \markchange{we define $g_{\kappa,\x}(\cdot)$ to be the Gaussian density with mean $\x$ and covariance matrix $\kappa^2\mathsf{I}_d$, i.e., $\forall \z\in \R^d$,}
\begin{equation}
      g_{\kappa,\x}(\z):= \frac{1}{(2\pi \kappa^2)^{d/2}}e^{-\frac{\Vert \z-\x \Vert^2}{2\kappa^2}}.
\label{eq:discretegauss}
\end{equation}
We also define
\begin{equation}
      g_{\kappa,\x}(\L):=\sum_{\lambda\in\L}g_{\kappa,\x}(\lambda).
\label{eq:discretegauss_lambda}
\end{equation}
We will use $g_{\kappa}(\z)$ and $ g_{\kappa}(\L)$ to denote $g_{\kappa,\mathbf{0}}(\z)$ and $ g_{\kappa,\mathbf{0}}(\L)$, respectively.

\subsection{Coding Scheme for Strong Secrecy}\label{sec:scheme_strong}

\noindent{\underline{Code:}} Following Section~\ref{sec:scheme_perfect}, we use a $(\Lfd,\Lcd)$ nested lattice code, with $\Lcd\subseteq \Lfd$.
As before, the messages are chosen from $\Gpd:=\Lfd/\Lcd$, and $\oplus $ is the addition operation on $\Gpd$. The $M^{(d)}:=|\Gpd|$ cosets of
$\Lcd$ in $\Lfd$ are denoted by $\L_0,\ldots,\L_{M^{(d)}-1}$.

\noindent{\underline{Encoding:}} For a coset $\L_j$ of $\Lcd$ in $\Lfd$, let $\mathbf{\lambda}_j$ denote its representative within $\cV(\Lcd)$ (see Fig.~\ref{fig:cosetrep} for an illustration). Fix a $\kappa>0$.
Corresponding to the message $\L_j$, the user node transmits a random lattice point from $\L_j$, according to the distribution
\begin{equation}
      p_{j}(\u)=\begin{cases}
			      \frac{g_{\kappa}(\u)}{g_{\kappa,-\mathbf{\lambda}_j}(\Lcd)} & \text{if } \u\in\L_j, \\
			      \mathbf{0} & \text{otherwise.}
                \end{cases}
\label{eq:pj_strong}
\end{equation}

\noindent{\underline{Decoding:}} \markchange{The relay computes the closest point in $\Lfd$ to the linear minimum mean-squared error (MMSE) estimate of the received vector, as in~\cite{Ling13,Erez04,Nazer11}, and the output of the decoder is the coset to which this point belongs. Let $\cmmse=\frac{2\kappa^2}{2\kappa^2+\nsvar}$ be the linear MMSE coefficient, and $\widetilde{\mathbf{W}}=[\cmmse\W]\bmod\Lcd$.
The estimate of $X\oplus Y$, denoted by $\cD(\W)$, is then the coset to which $Q_{\Lfd}(\widetilde{\mathbf{W}})$ belongs. }

\noindent{\underline{Achievable power-rate pair:}} A power-rate pair of $(\cP,\cR)$ is achievable if for every $\delta>0$, there exists a sequence of $(\Lfd,\Lcd)$
nested lattice codes \markchange{such} that for all sufficiently large $d$,
\begin{itemize}
      \item the average transmit power per dimension is less than $\cP+\delta$:
	    \[
		  P^{(d)}:=\frac{1}{d}\mathbb{E}\Vert \U \Vert^2 = \frac{1}{d}\mathbb{E}\Vert \V \Vert^2 <\cP+\delta;
	    \]
      \item the transmission rate is greater than $\cR-\delta$:
	    \[
	          R^{(d)}:=\frac{1}{d}\log_2 M^{(d)}>\cR-\delta;
	    \]
      \item the average probability of decoding $X\oplus Y$ incorrectly from $\W$ is less than $\delta$; and
      \item the mutual information between each message and $\U+\V$ is less than $\delta$:
	    \[
	          \mathcal{I}(X;\U+\V)=\mathcal{I}(Y;\U+\V)<\delta.
	    \]
\end{itemize}

In the next two subsections, we will prove that

\begin{theorem}
      A power-rate pair of 
\[
      \left( \cP,\left[\frac{1}{2}\log_2\left(\frac{1}{2}+\frac{\cP}{\nsvar}\right)-\frac{1}{2}\log_2 2e\right]^{+}\right)
\]
can be achieved with strong secrecy using the coding scheme of Section~\ref{sec:scheme_strong}.
\label{thm:strong_main}
\end{theorem}

\subsection{Strong Secrecy in the Absence of Noise}\label{sec:strongsec_noiseless}

We will first prove that the scheme described in the previous section achieves strong secrecy. Let us establish some more notation.
Let $p_{U+V}(\cdot)$ denote the distribution of $U+V$, and for any $\mathbf{x}\in \Lfd\cap\cV(\Lcd)$, let $p_{U+V|\x}(\cdot)$ denote
the distribution of $U+V$ conditioned on the event that $X$ is the coset to which $\mathbf{x}$ belongs. We will show that for every $\x$ in $\Lfd\cap\cV(\Lcd)$ the \emph{variational distance} (also called the total variation distance) between $p_{U+V}$ and $p_{U+V|\x}(\cdot)$, defined as\footnote{For probability measures $P_1$ and $P_2$ defined on a discrete alphabet $\cX$, the total variation distance between them is usually defined as $\mathbb{V}(P_1,P_2) := \sup_{A \subseteq \cX} |P_1(A) - P_2(A)|$. This can be shown to be equal to $\frac{1}{2} \sum_{x \in \cX} |P_1(x) - P_2(x)|$ (see e.g., \cite[Section~11.6]{Cover}). We have dropped the $\frac{1}{2}$ factor for simplicity.}
\begin{equation}
      \mathbb{V}(p_{U+V},p_{U+V|\x}):=\sum_{\w\in\Lfd}|p_{U+V}(\w)-p_{U+V|\x}(\w)|,
\end{equation}
 goes to zero exponentially in the dimension $d$. Therefore, the \emph{average variational distance} between the joint pmf of $U+V$ and $X$, and the product of the marginals,
\[
      \overline{\mathbb{V}}:= \sum_{\x\in \Lfd\cap\cV(\Lcd)}\frac{1}{|\Gpd|}\mathbb{V}(p_{U+V},p_{U+V|\x}),
\]
also goes to zero exponentially in $d$.
 We can then use the following lemma, which relates the mutual information and the variational distance.
\begin{lemma}[\cite{CN04}, Lemma 1]
      For $|\Gpd| \ge 4$, we have 
\begin{equation}
      \mathcal{I}(X;\U+\V)\leq \overline{\mathbb{V}}\left(\log_2|\Gpd|-\log_2(\overline{\mathbb{V}}) \right).
\label{eq:mi_vardist}
\end{equation}
 \label{lemma:csiszar}
\end{lemma}
Since  $|\Gpd|$ grows exponentially in $d$, it is sufficient to have $\overline{\mathbb{V}}$ going to zero as $o(1/d)$ for $\mathcal{I}(X;\U+\V)$ to go to zero.
\markchange{We will in fact show that $\overline{\mathbb{V}}$ can be made to go to zero exponentially in $d$, which} will guarantee that the mutual information also decays exponentially in $d$.
In order to have $\overline{\mathbb{V}}$ going to zero exponentially in $d$, we will require the coarse and fine lattices to satisfy certain properties.

For any lattice $\L$ in $\R^d$, and any $\theta>0$, the flatness factor $\epsilon_{\L}(\theta)$ is defined as~\cite{Ling13,Belfiore11}
\begin{equation}
      \epsilon_{\L}(\theta):=\frac{\max_{\x\in\cV(\L)}\vert \left( \sum_{\lambda\in \L}g_{\theta,\lambda}(\x) \right)-(1/\text{det}\L )\vert}{1/\text{det}\L}.
      \label{eq:flatnessfact}
\end{equation}
\markchange{A useful property of the flatness factor is that it is a monotonic function of $\theta$: for $a>b>0$, and any lattice $\L$, we have $\epsilon_{\L}(a)\leq \epsilon_{\L}(b)$~\cite[Remark 2]{Ling13}.}
Following~\cite{Ling13}, we define a sequence of lattices $\{\Lfd\}$ to be \emph{secrecy-good} if 
\[
     \epsilon_{\Lfd}(\theta)\leq 2^{-\Omega(d)} \text{ for all } \theta \text{ such that }\frac{(\text{det}(\Lfd))^{2/d}}{2\pi\theta^2}<1.
\]
It was shown in~\cite{Ling13} that there exist lattices that are secrecy-good and also satisfy all the goodness properties described in Appendix~C.

Let us choose $\kappa$ in (\ref{eq:pj_strong}) to be equal to $\sqrt{\cP}$. 
We can bound the variational distance in terms of the flatness factor of the coarse lattice as follows:
\begin{theorem}
      If the sequence of  nested lattice pairs $\{\Lfd,\Lcd\}$ satisfies $\epsilon^{(d)}:=\epsilon_{\Lcd}(\sqrt{\cP/2})<1/2$, then for every $\x\in \Lfd\cap\cV(\Lcd)$, we have 
\begin{equation}
      \mathbb{V}(p_{U+V},p_{U+V|\x})\leq 216\epsilon^{(d)}.
      \label{eq:vardistbound}
\end{equation}
\label{thm:vardistbound}
\end{theorem}
A proof of the above theorem is given in Appendix~F. \markchange{The constant $216$ in the above theorem can be improved, but we do not attempt to do so, as the exact constant is not important for our purposes.}

\markchange{
The following result from~\cite[Section V-B]{Ling13} tells us that if the flatness factor of the coarse lattice goes to zero as $d\to\infty$, then the average transmit power converges to $\cP$.
\begin{lemma}
If the flatness factor $\epsilon_1:=\epsilon_{\Lcd}\left(\cP\sqrt{1-1/(e\pi)}\right)<1/2$, then,
\[
 \left\vert \mathbb{E}\Vert \mathbf{U}\Vert^2 - d\cP \right\vert=\left\vert \mathbb{E}\Vert \mathbf{V}\Vert^2 - d\cP \right\vert \leq \frac{2\pi \epsilon_1}{1-\epsilon_1}\cP.
\]
\label{lemma:power_strong}
\end{lemma}
Since $\sqrt{1-1/(e\pi)}>1/\sqrt{2}$, it is sufficient to have (by monotonicity of the flatness factor) $\epsilon_{\Lcd}(\sqrt{\cP/2})\to 0$ to satisfy the power constraint for all sufficiently large $d$.
}
From Theorem~\ref{thm:vardistbound} and Lemma~\ref{lemma:csiszar}, we see that strong secrecy can be obtained in the noiseless scenario.

\subsection{Strong Secrecy and Reliability of Decoding in the Presence of AWGN}
Since the noise $\bZ$ is independent of everything else, we have strong secrecy in a noisy channel as well.
To see why this is the case, observe that $X\to(\U+\V)\to(\U+\V+\bZ)$ forms a Markov chain. Using the data-processing inequality, we see that $\mathcal{I}(X;\U+\V+\bZ)\leq \mathcal{I}(X;\U+\V)$, verifying our claim. \markchange{Note that the claim holds regardless of the probability distribution of the noise $\bZ$. The fact that the noise is Gaussian will be used to determine achievable rates for reliable decoding of $X\oplus Y$ at the relay.}


We choose our sequence of nested lattices $\{\Lfd,\Lcd\}$ so as to satisfy the following properties:
\begin{itemize}
	\item[(L1)] The sequence of coarse lattices, $\{\Lcd\}$, is good for covering, MSE quantization, and AWGN channel coding\footnote{For the definitions of lattices good for covering, MSE quantization, and AWGN channel coding, see Appendix~C.}.
	\item[(L2)]  The sequence of coarse lattices, $\{\Lcd\}$, is secrecy-good.	
		\item[(L3)] The sequence of fine lattices, $\{\Lfd\}$, is good for AWGN channel coding.
\end{itemize}

\markchange{Using (44) in~\cite[Appendix II]{Ling13} and \cite[Proposition 2]{Ling13},  we can show that if $\L_0$ is a lattice sampled uniformly at random from a $(d,k,q)$ ensemble, where $d,k,q$ satisfy (\ref{eq:k_k1}) and (\ref{eq:rmin_constraints}), then for all sufficiently large $d$, we have $\mathbb{E}[\epsilon_{\L_0}(\theta)]\leq 2\left( \frac{(\text{det}(\L_0))^{2/d}}{2\pi \theta^2} \right)^{d/2},$ which goes to zero exponentially in $d$ as long as $ \frac{(\text{det}(\L_0))^{2/d}}{2\pi \theta^2}<1$. Using the Markov inequality, we can say that  the probability of choosing a lattice whose flatness factor is less than $4\left( \frac{(\text{det}(\L_0))^{2/d}}{2\pi \theta^2} \right)^{d/2}$ is at least $1/2$ for all sufficiently large $d$. From Lemma~\ref{lemma:goodlattice1}, we know that a randomly chosen nested lattice pair satisfies (L1) and (L3) with probability tending to $1$ as $d\to\infty$. We can then use the union bound to conclude  that a randomly chosen pair of nested lattices from the $(d,k,q,k_1,q_1)$ ensemble satisfies (L1)--(L3) with probability 
at least 
$1/2$ as $d\to \infty$.}

\markchange{
We now work towards an estimate of the probability of error of decoding $X\oplus Y$ from $\mathbf{W}$.
Recall that the relay computes $\widetilde{\mathbf{W}}=[\cmmse\W]\bmod\Lcd$, where $\cmmse=\frac{2\cP}{2\cP+\nsvar}$, and the estimate of $X\oplus Y$ is the coset to which $Q_{\Lfd}(\widetilde{\mathbf{W}})$ belongs.
The quantity $\widetilde{\mathbf{W}}$ can be written as
\begin{align}
    \widetilde{\mathbf{W}} &= [\cmmse(\U+\V+\bZ)]\bmod\Lcd &\notag \\
                      &=  [\U+\V-(1-\cmmse)(\U+\V)+\cmmse\bZ]\bmod\Lcd &\notag  \\
                      &= \left[ [\mathbf{X}+\mathbf{Y}]\bmod\Lcd + \bZ'   \right]\bmod\Lcd,&
\end{align}
where $\bZ'=(\cmmse-1)(\U+\V)+\cmmse\bZ$ is the effective noise of the MLAN channel. Unlike in Section~\ref{sec:ach_rate}, $\bZ'$ is not statistically independent of $[\mathbf{X}+\mathbf{Y}]\bmod\Lcd$. 
However, as shown by the following lemma, 
if the flatness factor of the coarse lattice is small, then the effective noise behaves like an almost independent Gaussian vector. Let $f_{\bZ'|\x,\y}$ denote the density function of $\bZ'$ conditioned on $\mathbf{X}=\x$ and $\mathbf{Y}=\y$, and $f_{\mathbf{N}}$ denote the density function of a Gaussian random vector, $\mathbf{N}$, with mean $\0$ and covariance matrix $\big(2(1-\cmmse)^2\cP+(\cmmse)^2\nsvar\big)I_{d}$. Given two density functions $f_1$ and $f_2$ over $\R^d$, the variational distance between $f_1$ and $f_2$, denoted by $\mathbb{V}(f_1,f_2)$, is defined as
\[
   \mathbb{V}(f_1,f_2):=\int_{\x\in\R^d}\vert f_1(\x)-f_2(\x) \vert \: d\x.
\]
Then, we have the following lemma proved in Appendix~G.
\begin{lemma}
    If $\epsilon_{\Lcd}(\sqrt{\cmmse\cP})<1/2$, then for every $\x$ and $\y$ in $\Gpd$, 
    \[
        \mathbb{V}(f_{\bZ'|\x,\y},f_{\mathbf{N}}) \leq 8 \epsilon_{\Lcd}(\sqrt{\cmmse\cP}).
    \]
\label{lemma:MMSEeff_noise}
\end{lemma}
}

\markchange{
\subsubsection{Proof of Theorem~\ref{thm:strong_main}}
If $\mathbb{P}_1$ and $\mathbb{P}_2$ are probability measures on $\R^d$ having densities $f_1$ and $f_2$ respectively, then $\sup_{A\subset \R^d}\vert \mathbb{P}_1(A)-\mathbb{P}_2(A) \vert=\frac{1}{2} \mathbb{V}(f_1,f_2)$, where the supremum is taken over all measurable subsets of $\R^d$ (assuming that both $P_1$ and $P_2$ are defined on a common event space) \cite[Section 7.7]{Dasgupta}.
Using this and Lemma~\ref{lemma:MMSEeff_noise}, the probability of error of the decoder can be bounded by
\begin{align}
    \Pe &\leq \Pr\left[ \bZ' \notin \cV(\Lfd) \right] &\notag \\
           &\leq \Pr\left[ \mathbf{N}\notin \cV(\Lfd) \right] + 4 \epsilon_{\Lcd}(\sqrt{\cmmse\cP}). &
\end{align}
The variance of $\mathbf{N}$ is equal to $\sigma_N^2=2(1-\cmmse)^2 \cP+(\cmmse)^2\nsvar=\frac{2\cP\nsvar}{2\cP+\nsvar}$.
If the flatness factor $ \epsilon_{\Lcd}(\sqrt{\cmmse\cP})\to 0$ as $d\to\infty$, and the fine lattices are good for AWGN channel coding, then
the probability of error at the relay goes to zero as long as $   \frac{(\text{det}(\Lfd))^{2/d}}{2\pi e \sigma_N^2}>1$, or equivalently,
\[
    \frac{1}{|\Gpd|^{2/d}}\frac{(\text{det}(\Lcd))^{2/d}}{2\pi e \sigma_N^2}>1.
\]
In other words,
\begin{equation}
    R^{(d)}=\frac{1}{d}\log_2|\Gpd|< \frac{1}{2}\log_2\left( \frac{(\text{det}(\Lcd))^{2/d}}{2\pi e \sigma_N^2} \right).
    \label{eq:strong_rdconst}
\end{equation}
}

\markchange{
If we have $\cmmse\geq 1/2$, then by monotonicity of the flatness factor, $ \epsilon_{\Lcd}(\sqrt{\cmmse\cP})\leq  \epsilon_{\Lcd}(\sqrt{\cP/2})$.
This requires $\frac{2\cP}{2\cP+\nsvar}\geq1/2$,  or $\cP\geq \nsvar/2$. 
Observe that having $\epsilon_{\Lcd}(\sqrt{\cP/2})\to 0$ has three important consequences: (a) strong secrecy, even in the absence of noise (Theorem~\ref{thm:vardistbound}); (b) the average transmit power converges to $\cP$ (Lemma~\ref{lemma:power_strong}); and (c) the effective noise vector is ``almost'' independent of the message (Lemma~\ref{lemma:MMSEeff_noise}).
}

\markchange{
Using (L2), in order to have the flatness factor $\epsilon_{\Lcd}(\sqrt{\cP/2})\to 0$, the coarse lattices must be scaled so that
\begin{equation}
			\frac{\left(\text{det}(\Lcd)\right)^{2/d}}{2\pi (\cP/2)} <1.	
		\label{eq:elcd_cond}
\end{equation}
Let us choose $\left(\text{det}(\Lcd)\right)^{2/d}=\pi \cP-\delta$, for some arbitrary $\delta>0$, so as to satisfy (\ref{eq:elcd_cond}).
Substituting this in (\ref{eq:strong_rdconst}), we get that for  $\cP\geq\nsvar/2$, as long as 
\[
    R^{(d)}<\frac{1}{2}\log_2\left( \frac{\cP-\delta/\pi}{2e\sigma_N^2} \right),
\]
 the probability of error of decoding $X\oplus Y$ at the relay, as well as the mutual information between the individual messages and $\W$, go to zero as $d\to\infty$.
Substituting for $\sigma_N^2$, we complete the proof of Theorem~\ref{thm:strong_main}. \qed
}

\markchange{
\begin{remark}
In the perfect secrecy setting, we were not able to show that the technique of  MMSE scaling can be used to obtain an additional $1/2$ in the rate expression.
As in the strong-secrecy case, suppose that the relay computes $\widetilde{\W}:=[\alpha^{*}\W]\bmod \Lcd$, where $\alpha^{*}:=(2\cP)/(2\cP+\nsvar)$. 
The effective noise vector, $\bZ_{\text{eff}}=-(1-\alpha^{*})(\U+\V)+\alpha^{*}\bZ$
is not Gaussian, since $\U$ and $\V$ are not Gaussian. In order to find the probability of decoding error, we require an upper bound on the probability that $\bZ_{\text{eff}}\notin\cV(\Lfd)$, which is not straightforward unlike in the Gaussian case. Consequently, we were not able to say whether lattice decoding achieves vanishingly small error probabilities in this situation.
\end{remark}
}
\subsection{Prior Work on Strong Secrecy}
	The strongly secure scheme proposed by He and Yener in~\cite{HeYenerstrong}  also used nested lattice codes as we have done here. They obtain strong secrecy using universal hash functions, and show the existence of a suitable linear hash function that ensures that the mutual information decays exponentially in $d$. Unlike~\cite{HeYenerstrong}, we have used a sampled Gaussian pmf for randomization at the encoder, and hence, for a given pair of nested lattices, we explicitly specify the distribution used for randomization. Even using our scheme, the mutual information goes down to zero exponentially in $d$. But unlike~\cite{HeYenerstrong}, which was valid under a maximum power constraint at each node, the codebook we use is unbounded, so our scheme can only satisfy an average power constraint. Also, the achievable rate in the scheme of He and Yener is slightly higher (by $\frac{1}{2}\log_2\frac{e}{2}$ bits per channel use). 
 \markchange{On the other hand, the He-Yener randomization scheme uses hash functions whose existence is only guaranteed by a probabilistic argument, while our randomization scheme has the advantage of being specified by sampled Gaussian pmfs that can be given in explicit form}. 
The scheme in~\cite{HeYenerstrong} was coupled with an Algebraic Manipulation Detection (AMD) code~\cite{Cramer} for Byzantine detection, and it was shown that the probability of a  Byzantine attack being undetected could be made to decay to zero exponentially in $d$. We remark that our coding scheme can also be extended to this scenario, \markchange{where it can} be used as a replacement for the nested lattice code in~\cite{HeYenerstrong}.

%


\section{Multi-hop Line Network}\label{sec:multihop}
\markchange{The bidirectional relay can be viewed as a building block in many wireless networks. In particular, the problem of secure compute-and-forward can be extended to scenarios where we want secure relaying of messages from one point to another on a network with multiple honest-but-curious relays. As an example, we will extend our results} to the multi-hop line network studied in~\cite{HeYener}.
The structure of a multi-hop line network with $K+1$ hops is shown in Fig.~\ref{fig:multihopnet}. 
It consists of $K+2$ nodes: a source node, $\tS$, a destination node, $\tD$, and $K$ relay nodes, $\tR_1,\tR_2,\ldots,\tR_K$.
It is assumed that all links are identical AWGN (mean zero, variance $\nsvar$) wireless links. All nodes are half-duplex and can communicate only with their neighbours.
Nodes broadcast their messages to their immediate neighbours.

The source wants to send $N$ messages, $X_1,X_2,\ldots,X_N$, to the destination across the network of honest-but-curious relays.
The messages are assumed to be independent and uniformly distributed over  the set of all messages.
It is assumed that the relays do not co-operate with each other, i.e., the information available at a relay is not shared with the other relays.
As remarked by He and Yener in~\cite{HeYener}, this also takes care of the situation wherein the eavesdropper has access to one of the
relays, but it is not known which relay has been compromised.
We study this problem mainly under the strong secrecy constraint, but the arguments can be extended to the perfect secrecy scenario.

\markchange{He and Yener showed that their scheme~\cite{HeYener} achieves weak secrecy over the multi-hop line network, but their arguments cannot be
directly extended for strong secrecy. We give a new proof that shows that our strongly secure scheme for the bidirectional 
relay can be used with the He and Yener co-operative jamming protocol to obtain strong secrecy in a multi-hop line network.\footnote{\markchange{In fact, our proof shows that any strongly secure coding scheme for the bidirectional relay can be used to obtain strong secrecy in the multihop network. However, the achievable rate would depend on the coding scheme.}}}

\begin{figure}
 \centering
 	\resizebox{7cm}{!}{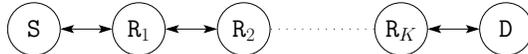}
 	\caption{Multi-hop line network with $K+1$ hops.}
 \label{fig:multihopnet}
\end{figure}

 
\subsubsection{The Communication Scheme}\label{sec:multischeme}

We use the co-operative jamming scheme proposed by He and Yener for relaying.
The communication takes place in $2N+K$ phases, where each phase consists of $d$ channel uses.
Let us choose a sequence of $(\Lfd,\Lcd)$ nested lattice pairs that satisfy properties (L1)--(L3). 
Each node in the network employs the encoding and decoding scheme described in Section~\ref{sec:scheme_strong}.
Let $\mathcal{D}:\R^d\to \Gpd$ denote the decoder map of Section~\ref{sec:scheme_strong}.
Also, for any $X\in \Gpd$, let $\mathcal{E}(X)$ denote the encoded form of $X$ as in Section~\ref{sec:scheme_strong}.

\begin{itemize}
 \item Each relay node $i$ ($i=1,\ldots ,K$) generates a \emph{jamming signal,} $J_{i}$, which is chosen uniformly at random from $\Gpd$, and independently of everything else. The destination generates $N$ independent jamming signals, $J_{K+l}$, for $l=1,2,\ldots,N$, where $N$ is the number of messages to be relayed.
 \item Let $\W_{i}[n]$ denote the $d$-dimensional vector received by the $i$th node in the $n$th phase, and let $\V_{i}[n]$ be the vector transmitted by the $i$th node in the $n$th phase. 
\end{itemize}

 An average power constraint is imposed at the nodes: $ \frac{1}{d}\E\Vert \V_{i}[n] \Vert^2 \leq P^{(d)}$
for $i=0,1,\ldots,K+1$ and $n=1,2,\ldots,K+2N$.

Since it takes $K+2N$ phases for sending $N$ messages, the rate of the scheme is defined as
\begin{equation}
 R^{(d)}_N := \frac{N}{d(K+2N)}\log_2|\Gpd| .
\label{eq:rd}
\end{equation}
We say that a \emph{power-rate pair} of $(\cP,\cR)$ is \emph{achievable for $N$-message transmission} with strong secrecy in a multi-hop line network with $K+1$ hops, if for every $\delta>0$, there exists a sequence of $(\Lfd,\Lcd)$ nested lattice codes such that for all sufficiently large $d$, we have 
\begin{itemize}
     \item $\Pd<\cP+\delta$;
     \item $R^{(d)}_N>\cR-\delta$;
     \item the probability of the destination decoding $X_1,X_2,\ldots,X_{N}$ incorrectly, $\Pe$, is less than $\delta$; and,
     \item for $k=1,2,\ldots,K$, the mutual information between the $N$ messages and all the variables available at the $k$th relay is less than $\delta$, i.e., 
     \[
	  \mathcal{I}(X_1,\ldots,X_N ; J_k,\W_{k}[1],\ldots,\W_{k}[2N+K])<\delta.  
     \]
\end{itemize}
We will describe the scheme for secure message relaying in the next subsection, and find achievable power-rate pairs. 
As the main result, letting the number of messages to go to infinity, we will show the following:
\begin{theorem}
A power-rate pair of 
 \[
 \left( \cP,\left[\frac{1}{4}\log_2\left(\frac{1}{2}+ \frac{\cP}{\nsvar}\right)-\frac{1}{4}\log_2 2e\right]^{+} \right)
 \]
  is achievable with strong secrecy\footnote{\markchange{If the scheme in~\cite{HeYenerstrong} is used at each node, then the achievable rate with strong secrecy can be improved to $\left[\frac{1}{4}\log_2\left(\frac{1}{2}+ \frac{\cP}{\nsvar}\right)-\frac{1}{2}\right]^{+}$.}}, and
a power-rate pair of 
 \[
 \left( \cP,\left[\frac{1}{4}\log_2\left( \frac{\cP}{\nsvar}\right)-\frac{1}{2}\log_2 2e\right]^{+} \right)
 \]
  is achievable with perfect secrecy at the relay nodes in a multi-hop line network with $K+1$ hops.
  \label{thm:multihop}
\end{theorem}

\begin{figure}
	\centering
		\resizebox{7cm}{!}{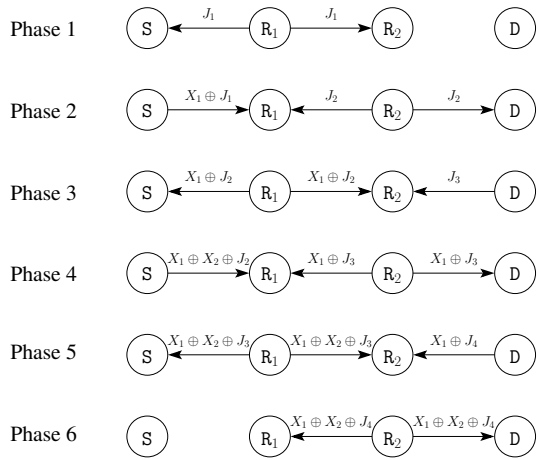}
	\caption{Secure relaying of two messages in a 3-hop relay network.}
	\label{fig:multihopeg3}
\end{figure}

\begin{table*}[t]
 \centering
  \begin{tabular}{|c|p{1.5cm}|p{1.5cm}|p{1.5cm}|p{1.5cm}|}
   \hline
   Phase & \multicolumn{4}{|c|}{Messages available at node at the end of phase} \\ \cline{2-5} 
     & \multicolumn{1}{|c|}{$\tS$} &  \multicolumn{1}{|c|}{$\tR_1$} & \multicolumn{1}{|c|}{$\tR_2$} & \multicolumn{1}{|c|}{$\tD$}  \\ \cline{1-5}
   0 & \multicolumn{1}{|c|}{ $X_1$, $X_2$} & \multicolumn{1}{|c|}{ $J_1$} & \multicolumn{1}{|c|}{$J_2$} & \multicolumn{1}{|c|}{$J_3$, $J_4$}  \\ \cline{1-5}
   1 & \multicolumn{1}{|c|}{ $X_1$, $X_2$, $J_1$} & \multicolumn{1}{|c|}{ $J_1$} & \multicolumn{1}{|c|}{$J_1$, $J_2$} & \multicolumn{1}{|c|}{$J_3$, $J_4$}  \\ \cline{1-5}
   2 & \multicolumn{1}{|c|}{ $X_1$, $X_2$, $J_1$} & \multicolumn{1}{|c|}{ $J_1$, $X_1\oplus J_2$} & \multicolumn{1}{|c|}{$J_1$, $J_2$} & \multicolumn{1}{|c|}{$J_2$, $J_3$, $J_4$}  \\ \cline{1-5}
   3 & \multicolumn{1}{|c|}{ $X_1$, $X_2$, $J_1$, $J_2$} & \multicolumn{1}{|c|}{ $J_1$, $X_1\oplus J_2$} & \multicolumn{1}{|c|}{$J_1$, $J_2$, $X_1\oplus J_3$} &  \multicolumn{1}{|c|}{$J_2$, $J_3$, $J_4$}  \\ \cline{1-5}
   4 & \multicolumn{1}{|c|}{ $X_1$, $X_2$, $J_1$, $J_2$} & \multicolumn{1}{|c|}{ $J_1$, $X_1\oplus J_2$, $X_1\oplus X_2\oplus J_3$} & \multicolumn{1}{|c|}{$J_1$, $J_2$, $X_1\oplus J_3$} &  \multicolumn{1}{|c|}{$X_1$, $J_2$, $J_3$, $J_4$}   \\ \cline{1-5}
   5 & \multicolumn{1}{|c|}{  $X_1$, $X_2$, $J_1$, } & \multicolumn{1}{|c|}{ $J_1$, $X_1\oplus J_2$,} & \multicolumn{1}{|c|}{$J_1$, $J_2$, $X_1\oplus J_3$, } &  \multicolumn{1}{|c|}{$X_1$, $J_2$, $J_3$, $J_4$}   \\ 
      & \multicolumn{1}{|c|}{ $J_2$, $J_3$} & \multicolumn{1}{|c|}{ $X_1\oplus X_2\oplus J_3$} & \multicolumn{1}{|c|}{ $X_1\oplus X_2\oplus J_4$} &  \multicolumn{1}{|c|}{}   \\ \cline{1-5}
   6 & \multicolumn{1}{|c|}{  $X_1$, $X_2$, $J_1$, } & \multicolumn{1}{|c|}{ $J_1$, $X_1\oplus J_2$, $X_1\oplus X_2\oplus J_3$, } & \multicolumn{1}{|c|}{$J_1$, $J_2$, $X_1\oplus J_3$, } &  \multicolumn{1}{|c|}{$X_1$, $X_2$, $J_2$, }  \\ 
     & \multicolumn{1}{|c|}{ $J_2$, $J_3$ } & \multicolumn{1}{|c|}{  $X_1\oplus X_2\oplus J_4$} & \multicolumn{1}{|c|}{ $X_1\oplus X_2\oplus J_4$} &  \multicolumn{1}{|c|}{ $J_3$, $J_4$}  \\ \cline{1-5}
 \end{tabular}
 \vspace{0.2cm}
 \caption{Messages available at various nodes at the end of each phase for the protocol in Fig.~\ref{fig:multihopeg3}.} 
 \label{tab:multihop} 
\end{table*}

\subsubsection{Scheme of He and Yener for Multi-Hop Relaying}
We now describe the scheme for secure relaying. A more detailed description can be found in~\cite{HeYener}.
The case where $\tS$ wants to send two messages, $X_1$ and $X_2$, to the destination
 is illustrated for a network with two relays in Fig.~\ref{fig:multihopeg3}. Only the messages (elements of $\Gpd$) transmitted by each node are indicated in the figure, and it is assumed that actual transmitted vectors are the encoded versions of the messages indicated. The messages available at various nodes at the end of each phase are tabulated in Table~\ref{tab:multihop}.
Let us use the notation $\oplus _{p=1}^{t}X_p$ to denote $X_1\oplus X_2\oplus\cdots\oplus X_t$.
\begin{itemize}
	\item The $i$th node ($i=0,1,2,\ldots,K+1$) transmits in the $(2t+i)$th phase, for $t=0,1,\ldots,N$. 
        \item In the $(2t+i)$th phase ($t=0,1,2,\ldots,N$), the $i$th node sends 
        \begin{equation}
             \V_i[2t+i]=\mathcal{E}\big( (\oplus_{p=1}^{t} X_p)\oplus J_{i+t} \big).
	  \label{eq:v_2ti}
        \end{equation}
        This holds for all nodes, $i=0,1,\ldots,K+1$. The $i$th node evaluates $(\oplus_{p=1}^{t} X_p)\oplus J_{i+t}$ by subtracting the message transmitted by it in the $(2t+i-2)$nd phase from the message decoded in the $(2t+i-1)$st phase.
\end{itemize}

%

Since the \markchange{destination} knows $J_{K+1},\ldots, J_{K+N}$, it can compute $ \oplus_{p=1}^{t} X_p $ from  $\mathcal{E}\big( (\oplus_{p=1}^{t} X_p)\oplus J_{K+t} \big)$, for $t=0,1,\ldots N$, and hence, each of the messages $X_l$.

\subsubsection{Secrecy}
Let us assume that all links are noiseless. As argued at the end of Section~\ref{sec:strongsec_noiseless}, it is enough to show that strong secrecy is obtained in this situation.
 Let $\{ X_p :p=1,\ldots,N\}$ denote the set of i.i.d. messages to be sent to the destination. Let us fix a $k$ from $\{1,2,\ldots,K\}$. 
In the $(2t+k-1)$st phase, the $k$th relay receives
\begin{align}
	\W_{k}[2t+k-1]&=\V_{k-1}[2t+k-1]+\V_{k+1}[2t+k-1]& \label{eq:wkvk}\\
				&=\mathcal{E}\Big((\oplus_{p=1}^{t}X_p)\oplus J_{k+t-1}\Big)+\mathcal{E}\Big((\oplus_{p=1}^{t-1}X_{p})\oplus J_{k+t}\Big),&
\end{align}
for $1\leq t\leq N$, and $\W_{k}[k-1]=\mathcal{E}(J_{k-1})$.
For $t=1,2,\ldots,N$, let us define 
\begin{equation}
     \Theta_{k,t}:= \{ J_k, J_{k-1}, \W_k[2m+k-1]: 1\leq m\leq t \}
\end{equation}
to be the set of all random variables available at the $k$th relay at the end of the $(2t+k-1)$st phase. We also define $\Theta_{k,0}:=\{ J_k,J_{k-1} \}$. Note that $\Theta_{k,t-1}\subset \Theta_{k,t}$ for $t=1,2,\ldots,N$, and $\Theta_{k,N}$ is the set of all random variables available at the $k$th relay at the end of all phases. We have to show that $\mathcal{I}(X_1,\ldots,X_N ; \Theta_{k,N})\to 0$ as $d\to\infty$. 

\begin{lemma}
     Let $\epsilon^{(d)}:=\epsilon_{\Lcd}(\sqrt{\cP/2})<1/2$. Then, the total information available at the $k$th relay node at the end of all relaying phases can be bounded from above as follows:
     \begin{equation}
          \mathcal{I}(X_1,\ldots,X_N ; \Theta_{k,N})\leq N\epsilon^{(d)}\left( \log_2|\Gpd|-\log_2\epsilon^{(d)} \right).
	  \label{eq:mi_multihopbound}
     \end{equation}
\label{lemma:mi_multihopbound}
\end{lemma}
\begin{proof}
     See Appendix~H.
\end{proof}

Since for our choice of nested lattices, $\epsilon^{(d)}\to 0$ exponentially in $d$, the mutual information $\mathcal{I}(X_1,\ldots,X_N ; \Theta_{k,N})$ also goes to zero exponentially in $d$, thereby guaranteeing strong secrecy.

\subsubsection{Achievable Rate and Proof of Theorem~\ref{thm:multihop}}

Using the union bound, one can show that for each $N$, the probability of the $k$th relay being in error in the $i$th phase goes to zero as $d\to\infty$ for all $k$ and $i$.
Using Theorem~\ref{thm:strong_main}, we can say that a power-rate pair of $\left( \cP,\frac{N}{2(K+2N+1)}\left[\log_2\left(\frac{1}{2}+ \frac{\cP}{\nsvar}\right)-\log_2 2e\right]^+  \right)$ is achievable for the transmission of $N$ messages using this scheme. Letting the number of messages, $N$, go to infinity, we have the first part of Theorem~\ref{thm:multihop}. The second part of the theorem can be proved in a similar manner. 


\section{Conclusion}\label{sec:conc}

We have described two coding schemes for secure bidirectional relaying in presence of an honest-but-curious relay.
We saw that using pmfs generated from density functions having compactly supported characteristic functions, one can obtain perfect secrecy.
We showed that reliable and perfectly secure computation at the relay is possible at transmission rates below $\left[\frac{1}{2}\log_2\frac{\cP}{\nsvar}-\log_2 2e\right]^{+}$. \markchange{This is the first such result for perfect secrecy in the context of the bidirectional relay.}
In order to achieve higher transmission rates, we relaxed the secrecy constraint, and only required that
the mutual information between $\U+\V$ and each individual message goes to zero for large block lengths.
Using pmfs obtained from sampled Gaussian functions, we could achieve a rate of $\left[\frac{1}{2}\log_2\left(\frac{1}{2}+\frac{\cP}{\nsvar}\right)-\frac{1}{2}\log_2 2e\right]^+$.
\markchange{Prior work by He and Yener showed that a rate of $\left[\frac{1}{2}\log_2\left(\frac{1}{2}+\frac{\cP}{\nsvar}\right)-1\right]^+$ is achievable with strong secrecy.}
These rates are within a constant gap of the best known achievable rate of $\left[\frac{1}{2}\log_2\left(\frac{1}{2}+\frac{\cP}{\nsvar}\right)\right]^+$ without secrecy constraints~\cite{Nazer11,Wilson}.

The main theme of this paper was the use of nested lattice codes, and explicit pmfs having infinite support to obtain security.
An inherent disadvantage of our scheme is that it is not possible to satisfy a maximum power constraint.
One could study the scenario where the support of the distributions we described are truncated, and find the 
performance of such a scheme; we are yet to carry out this study.

\markchange{All our results were derived under the assumptions that the messages are uniformly distributed, the channel gains from the user nodes to the relay are equal, and transmissions from both users are synchronized to arrive at the relay at the same time. Of course, in practice, these assumptions need not hold. Unfortunately, perfect secrecy does not appear to be robust to deviations from these assumptions. Indeed, if $X$ and $Y$ are not uniformly distributed, then we no longer have $(X\oplus Y)\independent X$ and $(X\oplus Y)\independent Y$. In general, if the channel gains are not equal and unknown at the user nodes, it is hard to get perfect secrecy. It can be shown that if $\u,\v\in \Lfd$, and $h_1,h_2$ are real numbers such that $h_1/h_2$ is irrational, then it is possible to exactly recover $(\u,\v)$ from $h_1\u+h_2\v$. However, it may be possible to obtain strong secrecy even when some of these assumptions do not hold, but this is left as future work. But it is worth noting that our scheme guarantees perfect (strong) security even in the absence of noise, and hence it achieves perfect (strong) secrecy even when the distribution of the additive noise is arbitrary and unknown, as long as it is independent of the transmitted codewords.}

  \markchange{The nested lattice coding schemes analyzed in this paper rely upon closest lattice point decoding, which is known to be computationally hard in general. However, recall that our randomization scheme for perfect secrecy works with \emph{any} pair of nested lattices. In particular, it would work with nested lattice pairs on which practical coding schemes can be based, where by ``practical coding schemes'' we mean explicitly constructed nested lattice codes that admit reliable decoding with low computational complexity. Lattice coding schemes with low-complexity decoders have been studied in the literature, e.g., \cite{diPietro13, tenBrink05, Sommer, Sommer09, Yan13}. Our scheme for strong secrecy, on the other hand, requires that the nested lattices satisfy various goodness properties. Further investigation is needed to determine whether all these goodness properties can be found in lattices that admit low-complexity decoding.}

Finally, in this paper, we only found achievable rates for secure and reliable computation at the relay.
As remarked in~\cite{HeYenerstrong}, finding a converse result is much harder. Even without any secrecy constraints, a nontrivial outer bound on the capacity of a bidirectional relay is not known.


\section*{Appendix~A: Technical Details of Example~\ref{ex:psi}}
We show here that the function $h(x) = (3\pi^2/4) \, [f(\pi x/ 4)]^2$, with $f$ as in
(\ref{eq:f2}), is a density function whose characteristic function 
is given by 
$$
\psi(t) = {\textstyle \frac32 \, g(\frac{4t}{\pi})},
$$
where $g$ is as in (\ref{eq:g}).

 Note first that $\hat{f}$ defined in (\ref{eq:fhat}) is also a
 probability density function --- it is non-negative and its integral over
 $(-\infty,\infty)$ is 1. By Fourier inversion, its characteristic
 function is $2\pi f$. Therefore, $g = \hat{f} \ast
 \hat{f}$ is a density with characteristic function $4\pi^2 f^2$.

 Now, $f^2$ is integrable since $(\hat{f})^2$ is integrable (see corollary
 to Theorem 3 of Section~XV.3 of \cite{Feller}). Hence, $\tilde{h}(x) =
 f^2(x)/(\int_{-\infty}^{\infty} f^2(y) \, dy$) is a probability density function.
 The integral in the denominator can be explicitly evaluated by 
 means of the Plancherel identity:
 $$
 \int_{-\infty}^\infty f^2(y) \, dy = \frac{1}{2\pi} \int_{-\infty}^\infty
 [\hat{f}(t)]^2 \, dt = \frac{1}{2\pi} \, g(0) = \frac{1}{3\pi},
 $$
 the last equality following from (\ref{eq:g}). Thus, $\tilde{h}(x) =
 3\pi f^2(x)$. 

 From the fact that $4\pi^2 f^2$ is the characteristic function of
 $g$, it follows by Fourier inversion that $\tilde{h}$ has
 characteristic function given by $\tilde{\psi}(t) = \frac32 \, {g(t)}$.
 Hence, $h(x) = (\pi/4) \tilde{h}(\pi x/4)$ is a density function with
 characteristic function $\tilde{\psi}(4t/\pi)$, which is precisely $\psi(t)$.

\section*{Appendix~B: Proof of Theorem~\ref{thm:L_thm}} 
We are given an index-$M$ sublattice $\L_0$ of the lattice $\L$. 
Recall from Section~\ref{sec:latticedefns} that $(\det \L_0)/(\det \L) = M$. Let $\L_0,\L_1,\ldots,\L_{M-1}$ denote the $M$ cosets of $\L_0$ in $\L$. These constitute the elements of the quotient group $\Gp = \L/\L_0$. 

Suppose that $X,Y$ are iid random variables, each uniformly distributed
over $\mathbb{G}$. For each $j \in \{0,1,\ldots M-1\}$, let $p_j$ be a pmf supported
within the coset $\L_j$, so that $p_j(\k) = 0$ for $\k \notin \L_j$. 
We define a random variable $U$ (resp.\ $V$) jointly distributed with
$X$ (resp.\ $Y$) as follows: if $X = \L_j$ (resp.\ $Y = \L_j$), $U$
(resp.\ $V$) is a random point from $\L_j$ picked according to the
distribution $p_j$. Then, $U$ and $V$ are identically distributed with
$p_U = p_V = \frac{1}{M}\sum_{i=0}^{M-1} p_i$. Let $\varphi_U$,
$\varphi_V$ and $\varphi_j$, $j = 0,1,\ldots,M-1$, be the
characteristic functions corresponding to $p_U$, $p_V$ and $p_j$, 
$j = 0,1,\ldots,M-1$, respectively. We have the following straightforward
generalization of Lemma~\ref{basic_lemma}.

\begin{lemma}
Suppose that $\varphi_U \varphi_V = \varphi_j\varphi_V = \varphi_U
\varphi_j$ for $j = 0,1,\ldots,M-1$. Then, the random variables
$(U,V,X,Y)$ with joint pmf given by 
\begin{align}
p_{UVXY}(\k,\mathbf{l},\L_i,\L_j) & = (1/M)(1/M)  p_i(\k) p_j(\mathbf{l}) & \notag \\
   & \qquad \text{ for } \k,\mathbf{l} \in \L \text{ and } \L_i,\L_j \in \Gp & 
\label{eq:Z1modN}
\end{align}
have properties (S1)--(S3).
\label{Zd_lemma}
\end{lemma}

We will now construct the characteristic functions $\varphi_j$ 
that satisfy the above lemma. Let $f$
be the (continuous) probability density function corresponding to the
compactly supported characteristic function $\psi$ in the hypothesis
of Theorem~\ref{thm:L_thm}. The function $f$ can be retrieved from $\psi$
by Fourier inversion:
\begin{align}
f(\x) &= \frac{1}{(2\pi)^d} \int_{\R^d} \psi(\t) e^{-i\langle \t,\x \rangle} \, d\t & \notag \\
& = \frac{1}{(2\pi)^d} \int_{\cV(\hat{\L}_0)} \psi(\t) e^{-i\langle \t,\x \rangle} \, d\t . &
\label{eq:Fou}
\end{align}
Note that each coset $\L_j$ can be expressed 
as $\u_j + \L_0$ for some $\u_j \in \L$. We set 
\begin{equation}
\varphi_j(\zetabf) = \sum_{\n \in \hat{\L}_0} \psi(\zetabf + \n) \, e^{-i \langle \n, \, \u_j \rangle}
\label{eq:phij}
\end{equation}
for all $\zetabf \in \R^d$. Then, by Proposition~\ref{prop:psf_Rd}, we have that $p_j$ is supported within $\L_j$, and 
\begin{equation}
p_j(\k) = (\det \L_0) \, f(\k)  \text{ for all } \k \in \L_j.
\label{eq:pj}
\end{equation}
Finally, define 
\begin{equation}
\varphi(\zetabf) = \sum_{\n \in \hat{\L}} \psi(\zetabf + \n) \,
\label{eq:phiUV}
\end{equation}
for all $\zetabf\in\R^d$. 

We make two claims: 
\begin{itemize}
\item[(i)] $\varphi^2 = \varphi \varphi_j$ for $j = 0,1,\ldots,M-1$;
\item[(ii)] $\varphi = \varphi_U = \varphi_V$.
\end{itemize}
Given these claims, by Lemma~\ref{Zd_lemma}, the random variables 
$U,V$ satisfy the properties \markchange{(S1)--(S3).} 

Both claims follow from the fact that $\hat{\L}$ is a sublattice of
$\hat{\L}_0$. (If a lattice $\G$ contains a sublattice $\G_0$, then the
dual $\G^*$ is a sublattice of $\G_0^*$.) To see (i), we re-write
(\ref{eq:phiUV}) as 
\begin{equation}
\varphi(\zetabf) = \sum_{\n \in \hat{\L}} \psi(\zetabf + \n) \, e^{-i \langle \n , \, \u_j \rangle}.
\label{eq:phiUV2}
\end{equation}
This is possible because, for $\n \in \hat{\L} = 2\pi\L^*$ and $\u_j \in \L$, we have $e^{-i \langle \n , \, \u_j \rangle} = 1$. Comparing (\ref{eq:phij}) and (\ref{eq:phiUV2}), and noting that $\psi$ is supported within $\cV(\hat{\L}_0)$, it is evident that $\supp(\varphi) := \{\zetabf: \varphi(\zetabf) \neq 0\}$ is contained in 
$\supp(\varphi_j) := \{\zetabf: \varphi_j(\zetabf) \neq 0\}$. 
Furthermore, for all $\zetabf \in \supp(\varphi)$, we have $\varphi(\zetabf) = \varphi_j(\zetabf)$. 
Claim~(i) directly follows from this.

For Claim~(ii), we note that $\cV(\hat{\L}_0) \subseteq \cV(\hat{\L})$, since $\hat{\L}$ is a sublattice of $\hat{\L}_0$. Hence, we can apply Proposition~\ref{prop:psf_Rd} to deduce that $\varphi$ is the characteristic function of a pmf $p$ supported within $\L$, with 
$$
p(\k) = (\det \L) \, f(\k) \text{ for all } \k \in \L.
$$
Thus, from (\ref{eq:pj}) and the fact that $(\det \L_0)/(\det \L) = M$, we see that $p = \frac{1}{M}\sum_{j=0}^{M-1}p_j$. In other words, $p = p_U = p_V$, which proves Claim~(ii).

\markchange{It remains to prove the statements concerning finiteness of $\E{\|U\|}^2$ and $\E{\|V\|}^2$. 
Theorem~1 in \cite{Wolfe73} shows that these moments are finite iff $\varphi$ is twice differentiable at $\0$
(i.e., all second-order partial derivatives exist at $\0$).
From (\ref{eq:phiUV}), we see that $\varphi$ agrees with $\psi$ in a small neighbourhood around $\0$; hence,
$\varphi$ is twice differentiable at $\0$ iff $\psi$ is twice differentiable at $\0$.}

\markchange{Assuming that $\psi$ has all second-order partial derivatives at $\0$, we must show that
$\E{\|U\|}^2 = \E{\|V\|}^2 = -\Delta \psi(\0)$.  Since $\U$ and $\V$ are identically distributed, it is enough to show that $\E{\|\U\|}^2 = -\Delta \psi(\0)$.}
Write $\U = (U_1,\ldots,U_d)$, so that ${\|\U\|}^2 = U_1^2 + \cdots + U_d^2$. 
We want to show that $\E[U_j^2] = -\frac{\partial^2}{\partial t_j^2}
\psi(\0)$, for $j = 1,\ldots,d$.  For notational simplicity, we show
this for $j=1$. Note that the characteristic function of \markchange{$U_1$} is given
by $\varphi_{U_1}(t_1) = \varphi_U(t_1,0,\ldots,0)$. As argued prior to the statement of
Theorem~\ref{mod2_thm} in Section~\ref{sec:constr}, \markchange{$\E[U_1^2] = -\varphi_{U_1}''(0)$}.
Now, $\varphi_{U_1}''(0) = \frac{\partial^2}{\partial
  t_1^2}\varphi_U(0,0,\ldots,0)$. From
(\ref{eq:phiUV}), we have that $\varphi_U = \psi$ in a small
neighbourhood around $\0 = (0,0,\ldots,0)$. Therefore, 
$\frac{\partial^2}{\partial t_1^2}\varphi_U(\0) 
= \frac{\partial^2}{\partial t_1^2}\psi(\0)$, and hence,
\markchange{$\E[U_1^2] = -\frac{\partial^2}{\partial t_1^2}\psi(\0)$}, as desired.

This concludes the proof of Theorem~\ref{thm:L_thm}.\qed

\section*{Appendix~C: ``Good'' Lattice Properties}

In this appendix, we briefly review certain ``good'' lattice properties, and some results in the literature. This is almost entirely based on~\cite{Erez05}.
Let $\{\L^{(d)}\}$ be a sequence of lattices, with each $\L^{(d)}$ chosen uniformly at random from a $(d,k,q)$ ensemble described in Section~\ref{sec:dkq_ensemble}. 

We say that that the sequence of lattices $\{\L^{(d)}\}$
is \emph{good for covering} if 
\[
      \lim_{d\to\infty}\frac{\rcov(\L^{(d)})}{\reff(\L^{(d)})}=1.
\]
We say that $\{\L^{(d)}\}$ is \emph{good for packing} if
\[
      \lim_{d\to\infty}\frac{\rpack(\L^{(d)})}{\reff(\L^{(d)})}\markchange{\geq}\frac{1}{2}.
\]
Let $\mathcal{G}_{\L^{(d)}}$ denote the normalized second moment per dimension of $\L^{(d)}$, as defined in Section~\ref{sec:latticedefns}. A sequence of lattices $\{\L^{(d)}\}$ is said to be \emph{good for MSE quantization} if $\mathcal{G}_{\L^{(d)}}\to\frac{1}{2\pi e}$ as $d\to\infty$.

Let $\bZ$ be a zero-mean $d$-dimensional white Gaussian vector having second moment per dimension equal to $\nsvar$. Let
\[
\markchange{      \mu:=\frac{\vol\big( \cV(\L^{(d)}) \big)^{2/d}}{\sigma^{2}}  }.
\]
 Then we say that $\{\L^{(d)}\}$ is \emph{good for AWGN channel coding} if the probability that $\bZ$ lies outside the fundamental Voronoi region of $\L^{(d)}$ is upper bounded by
\[
\Pr[ \bZ \notin \cV(\L^{(d)}) ] \leq e^{-d\big( E_{U}(\mu)-o_{d}(1) \big)}
\]
for all $\nsvar$ that satisfy $\mu\geq 2\pi e$.
Here, $E_{U}(\cdot)$, called the \emph{Poltyrev exponent} is defined as follows:
\begin{equation}
E_{U}(\mu) = \begin{cases}
\frac{\mu}{16\pi e} & \text{ if } 8\pi e \leq \mu \\
\frac{1}{2}\ln\frac{\mu}{8\pi} &\text{ if } 4\pi e \leq \mu \leq 8\pi e \\
\frac{\mu}{4\pi e} - \frac{1}{2}\ln\frac{\mu}{2\pi} &\text{ if } 2\pi e \leq \mu \leq 4\pi e.
\end{cases}
\label{eq:polty_exp}
\end{equation}
Suppose that we use a subcollection of points from $\L^{(d)}$ as the codebook for transmission over an AWGN channel. Then, as long as 
\[
\frac{\vol\big( \cV(\L^{(d)}) \big)^{2/d}}{\sigma^{2}} \geq 2\pi e,
\] 
the probability that a lattice decoder decodes to a lattice point other than the one that was transmitted, decays exponentially in the dimension $d$, with the exponent given by (\ref{eq:polty_exp}). 

It is worth noting that the above ``goodness'' properties are invariant to scaling. If $\{\Lfd\}$ is a sequence of lattices that is good for covering, packing, and AWGN channel coding, then for any $\alpha>0$,
$\{\alpha\Lfd\}$ is also good for covering, packing and AWGN channel coding. This is because of the fact that $\rpack(\alpha\Lfd)=\alpha\rpack(\Lfd)$, $\rcov(\alpha\Lfd)=\alpha\rcov(\Lfd)$, and  $\reff(\alpha\Lfd)=\alpha\reff(\Lfd)$.  

\section*{Appendix~D: Proof of Lemma~\ref{lemma:goodlattice1}}
In proving Lemma~\ref{lemma:goodlattice1}, we use the following theorem from~\cite{Erez05}, which says that if the parameters $k$ and $q$
are selected appropriately, then almost all lattices in a $(d,k,q)$ ensemble satisfy the ``goodness'' properties described in Appendix~C.
\begin{theorem}[\cite{Erez05}, Theorem 5]
 Let $0<r_\text{min}<\frac{1}{4}$ be chosen arbitrarily. Let $\L^{(d)}$ be a sequence of lattices selected uniformly at random from a $(d,k,q)$ ensemble, such that
\begin{itemize}
      \item $k\le\beta_1 d$ for some $0<\beta_1<1$, but $k$ grows faster than $\log^2 d$, and
      \item  $q$ is chosen so that $\reff(\L^{(d)})$, as given by (\ref{eq:reff_constA}), satisfies $r_\text{min}<\reff(\L^{(d)})<2r_\text{min}$. 
\end{itemize}
 Then, the sequence of lattices $\L^{(d)}$ is simultaneously good for covering, packing and MSE quantization,  with probability approaching $1$ as $d$ tends to infinity. If, in addition, we have $\beta_1<1/2$, then the sequence of lattices is also simultaneously good for AWGN channel coding with probability tending to $1$ as $d\to\infty$. 
\end{theorem}
 Therefore, if we choose $k$ and $q$ that satisfy the hypotheses of Lemma~\ref{lemma:goodlattice1}, then from the above theorem, the probability that a uniformly chosen $\Lcd$ satisfies condition $(G_1)$ tends to $1$ as $d\to\infty$. 

Recall from Section~\ref{sec:latticedefns} that if $\mathsf{A}$ is a generator matrix of a lattice $\L$, then the dual lattice of $\L$, denoted by $\L^{*}$, is the set of all integer linear combinations of the rows of $\mathsf{A}^{-1}$. It turns out that the dual of a Construction-A lattice is also a Construction-A lattice, as seen from the following.

\begin{proposition}
Suppose that $\mathsf{G}$ is \markchange{the  $k\times d$ systematic} generator matrix of a $(d,k)$ linear code $\cC$ over $\Z_q$, $q$ being prime, i.e., $\mathsf{G}$ has the form
\[
\mathsf{G}=\left[ \begin{array}{cc}
    \mathsf{I}_{k} & \mathsf{B} \end{array} \right],
\]
 where $\mathsf{I}_k$ denotes the $k\times k$ identity matrix. Let $\L(\cC)$ be the lattice obtained by employing Construction~A on the code $\cC$. Then, the matrix
\begin{equation}
\mathsf{A}=\frac{1}{q}\left[ \begin{array}{cc}
                       \mathsf{I}_{k} & \mathsf{B} \\
                       \mathsf{0} & q\mathsf{I}_{(d-k)} \end{array}\right]
\label{eq:Amatrix}
\end{equation}
is a generator matrix for the lattice $\L(\cC)$.
\label{prop:Amatrix}
\end{proposition}

\begin{proof}
      \markchange{We want to show that 
$\mathsf{A}^{T}\Z^{d}:= \{ \mathsf{A}^{T}\y : \y \in \Z^{d} \} = \L(\cC) $. 
By definition, $\L(\cC)= \{ \x\in \R^{d} : (q\x) \bmod q \in \cC \}$. 
Fix any $\z \in \Z^{d}$.
 Then, it can be verified that $ (q\mathsf{A}^{T}\z)\bmod q = (\mathsf{G}^{T}\hat{\z})\bmod q  $ (which is a codeword in $\cC$)
for some  $\hat{\z}\in\{ 0,1,\ldots,q-1 \}^k$.
Therefore, $(q\mathsf{A}^{T}\z)\bmod q \in \cC$, and hence, $\mathsf{A}^{T}\Z^{d}\subseteq \L(\cC)$.}
For the converse, define $\cC'=\{ \frac{1}{q}\c : \c \in \cC \}$. Then, $\L(\cC)= \cC'+\Z^{d}:= \{ \c+\z: \c\in \cC', \z\in \Z^{d} \}$. The set $\mathsf{A}^{T}\Z^{d}$ forms a group under (componentwise) addition. Hence, it is sufficient to show that $\cC'\subseteq \mathsf{A}^{T}\Z^{d}$, and $\Z^{d}\subseteq \mathsf{A}^{T}\Z^{d}$. Fix an arbitrary $\c\in \cC$. Let $\c' = \frac{1}{q}\c$. By definition, there exists an $\x\in \Z_{q}^{k}$ such that 
\begin{align}
 \c & = \left( \left[ \begin{array}{cc}
                       \mathsf{I}_{k} & \mathsf{B} \end{array}\right]^{T} \x \right)\bmod q & \notag \\
   &  = \left[ \begin{array}{c}
                        \x \\
                       \mathsf{B}^{T}\x \end{array}\right] - q \left[ \begin{array}{c}
								  \mathsf{0}\\
								    \z' \end{array}\right]
\end{align}
for some $\z'\in \Z^{d-k}$. Therefore,
\[
 \c = \left[ \begin{array}{cc}
                       \mathsf{I}_{k} & \mathsf{B} \\
                       \mathsf{0} & q\mathsf{I}_{(d-k)} \end{array}\right]^{T} \left[ \begin{array}{c}
											  \x \\
											  -\z' \end{array}\right].
\]

 Hence, there exists 
\[
\z = \left[ \begin{array}{c}
             \x \\
             \z '
            \end{array} \right] \in \Z^{d}
\]
 so that $\c'= \mathsf{A}^{T}\z$. Therefore, we can say that $\cC'\subseteq \mathsf{A}^{T}\Z^{d}$. Next, consider $\z\in \Z^{d}$. Let $\mathsf{A}^{*}$ be 
\markchange{defined as
\begin{equation}
\mathsf{A}^{*}=\left[\begin{array}{cc}
            q\mathsf{I}_{k} & \mathsf{0} \\
            -\mathsf{B}^{T} & \mathsf{I}_{(d-k)} \end{array} \right],
 \label{eq:Astar1}
\end{equation}}
 and note that $\mathsf{A}^{T}\mathsf{A}^{*}= \mathsf{I}_{d}$, the $d\times d$ identity matrix. Let $\z ' =\mathsf{A}^{*}\z\in \Z^{d}$. Then, $\mathsf{A}^{T}\z'= \mathsf{A}^{T}(\mathsf{A}^{*}\z) = (\mathsf{A}^{T}\mathsf{A}^{*})\z= \z$. Hence, we can say that for every $\z\in \Z^{d}$, there exists a $\z ' \in \Z^{d} $ so that \markchange{$\z =\mathsf{A}^{T}\z'$,} and hence $\Z^{d}\subseteq \mathsf{A}^{T}\Z^{d}$,  thus concluding the proof. 
\end{proof}

 It can be shown in a similar manner that if $\mathsf{G}$ has the form 
\[
 \mathsf{G}=\left[ \begin{array}{cc}
    \mathsf{B} & \mathsf{I}_{k} \end{array} \right],
\]
then, 
\[
 \mathsf{A}=\frac{1}{q}\left[ \begin{array}{cc}
                       \mathsf{B} & \mathsf{I}_{k} \\
                       q\mathsf{I}_{(d-k)} & \mathsf{0}  \end{array}\right]
\]
 is a generator matrix for $\L(\cC)$.

It is easy to verify that if $\mathsf{A}$ is full rank, then  $\mathsf{A}^{*}$ defined in (\ref{eq:Astar1}) is the inverse of $\mathsf{A}$, and $\mathsf{A}^{*}$ is a generator matrix of $\L^{*}(\cC)$.
Since a permutation of the rows of a generator matrix of a lattice also yields a valid generator matrix for the same lattice,
\[
 \mathsf{A}^{*}_{1}=\left[\begin{array}{cc}
            -\mathsf{B}^{T} & \mathsf{I}_{(d-k)} \\
            q\mathsf{I}_{k} & \mathsf{0} \end{array} \right]
\]
is also a generator matrix for $\L^{*}(\cC)$. If $\cC^{\perp}$ denotes the dual code of $\cC$, then $\cC^{\perp}$ has a generator matrix~\cite{Roth}
\[
 \mathsf{G}=\left[ \begin{array}{cc}
    -\mathsf{B}^{T} & \mathsf{I}_{(d-k)} \end{array} \right].
\]
We thus have the following result.
\begin{lemma}
Let $\cC$, $\mathsf{G}$, $\L(\cC)$ be as in Proposition~\ref{prop:Amatrix}. Then, the dual of $\L(\cC)$, denoted by $\L^{*}(\cC)$, has generator matrix
\begin{equation}
\mathsf{A}^{*}=\left[\begin{array}{cc}
            q\mathsf{I}_{k} & \mathsf{0} \\
            -\mathsf{B}^{T} & \mathsf{I}_{(d-k)} \end{array} \right].
\label{eq:Astar}
\end{equation}
Therefore, $\L^{*}(\cC)=q\L(\cC^{\perp})$, where $\cC^{\perp}$ denotes the dual code of $\cC$.
\label{lemma:Astar}
\end{lemma}
%
%
 Since $\L^{*}(\cC)=q\L(\cC^{\perp})$, if the generator matrix is full-rank, then $\L(\cC^{\perp})$ belongs to a $(d,d-k,q)$ ensemble. Therefore, from~\cite{Erez05}, we can say that a randomly picked $\L(\cC^{\perp})$ is good for packing and covering with probability tending to 1 as $d\to \infty$, as long as $d-k\leq \beta_1 d$ for some \markchange{$0<\beta_1<1$}, and $d-k$ grows faster than $\log^{2}d$. From the definitions, we see that the properties of covering and packing goodness are invariant to any scaling of the lattices. Therefore, if $\L(\cC^{\perp})$ is good for packing and covering, then $q\L(\cC^{\perp})$, and hence $\L^{*}(\cC)$ is also good for packing and covering. \markchange{We have seen that the probability of $\{ \Lcd \}$ being simultaneously good for covering, packing and AWGN channel coding tends to $1$ as  $d$ tends to $\infty$. If we choose $k=\beta_1 d$ for some $\beta_1<1/2$, then the sequence of dual lattices is good for packing with probability tending to $1$ as $d\to\infty$. Using the union bound, we can argue that a randomly picked sequence of coarse lattices satisfies $(G_1)$ and $(G_2)$ with probability going to $1$ as $d\to \infty$.} .

It was also shown in~\cite{Erez04} that if the coarse lattices are good for covering and AWGN channel coding, then as long as $d/q_1\to 0$ as $d\to\infty$, the probability that a uniformly chosen sequence of fine lattices is good for AWGN channel coding tends to $1$ as $d\to\infty$. 
This completes the proof of Lemma~\ref{lemma:goodlattice1}.

\section*{Appendix~E: Proof of Lemma~\ref{lemma:goodlattice2}}
For ease of notation, denote by $\reff$, the effective radius of $\Lcd$. The index, $d$, in $\reff$ has been dropped but it must be understood that this is a function of $d$.  Let $\cC^{(d)}$ denote the $(d,k)$ code over $\Z_{q}$ that is used to generate the coarse lattice. 
Using (\ref{eq:reff_constA}),
\begin{align}
q^{k}& = \frac{\Gamma(d/2+1)}{\pi^{d/2}\reff^{d}} & \notag \\
  &   = \sqrt{d\pi}\left(\frac{d}{2\pi e \reff^{2}} \right)^{d/2}(1+o_d(1)), & \label{eq:qk_reff_approx}
\end{align}
where the second step uses Stirling's approximation, and $o_d(1)$ is a term that approaches 0 as $d\to \infty$. From (\ref{eq:k_k1}), $k=\beta_0 d$ for some $0<\beta_0<1/2$. Substituting this in the above, and raising both sides to the power $1/d$, we get
\begin{equation}
q^{\beta_0} = (d\pi)^{\frac{1}{2d}}\left(\frac{d}{2\pi e \reff^{2}} \right)^{1/2} (1+o_d(1))^{1/d}=(d\pi)^{\frac{1}{2d}}\frac{\sqrt{d}}{\sqrt{2\pi e}\reff}(1+o_d(1)).
     \label{eq:AE_1}
\end{equation}
Let $\L_0^{(d)*}$ denote the dual of $\Lcd$, and $\reff^{*}$ denote the effective radius of $\L_0^{(d)*}$. Let $\L_0(\cC^{(d)\perp})$ be the lattice obtained by applying Construction-A on the dual of $\cC^{(d)}$, i.e., on $\cC^{(d)\perp}$. As remarked in Appendix~D, $\L_0(\cC^{(d)\perp})$ comes from a $(d,d-k,q)$ ensemble. From Lemma~\ref{lemma:Astar}, $\L_0^{(d)*}=q\L_0(\cC^{(d)\perp})$. Therefore, $(1/q)\L_0^{(d)*}= \L_0(\cC^{(d)\perp})$ will satisfy
\[
q^{d-k}=\sqrt{d \pi}  \left( \frac{d}{2\pi e \left(\reff(\frac{1}{q}\L_0^{(d)*})\right)^{2}} \right)^{d/2} (1+o_d(1)),
\]
where $o_d(1) \to 0$ as $d\to \infty$.
But $\reff(\frac{1}{q}\L_0^{(d)*}) = \frac{1}{q} \reff^{*}$, and hence, analogous to (\ref{eq:AE_1}), we have
\begin{equation}
q^{d(1-\beta_0)}  =  \sqrt{d \pi} \left( \frac{d}{2 \pi e (1/q)^{2} (\reff^{*})^{2}} \right)^{d/2} (1+o_d(1)) .
\end{equation}
Rearranging,
\begin{equation}
\reff^{*}   =  (d\pi)^{\frac{1}{2d}} \frac{\sqrt{d}q^{\beta_0}}{\sqrt{2 \pi e}} (1+o_d(1))^{1/d}  .
\label{eq:reffstar}
\end{equation}

Let the packing radius of $\L_0^{(d)*}$ be $\rpack(\L_0^{(d)*})=\gamma(d)\reff^{*}$. From the definition of the packing radius,  $\gamma(d)\leq 1$ for all $d$. Again, since the dual lattice is good for packing, $\lim_{d\to \infty}\gamma(d)\geq 1/2$. Also, since $o_d(1)\to 0$ as $d\to\infty$, we have $(1+o_d(1))^{1/d}=(1+o_d(1))$. Therefore, we have,
\begin{align}
\reff(\L_0^{(d)})\rpack(\L_0^{(d)*}) & = \gamma(d) \reff(\Lcd) (d\pi)^{(1/2d)} \frac{\sqrt{d}q^{\beta_0}}{\sqrt{2 \pi e}} (1+o_d(1)) .& \notag 
\end{align}
Substituting for $q^{\beta_0}$ from (\ref{eq:AE_1}) in the above equation, we get
\begin{align}
\frac{\reff(\L_0^{(d)})\rpack(\L_0^{(d)*})}{d} & = \gamma(d) (d \pi)^{(1/d)} \frac{1}{2\pi e} (1+o_d(1)) .& 
\end{align}
Therefore, as $d\to \infty$, the above expression converges to a value greater than or equal to $1/4\pi e$. Using $\rpack(\hLcd)=2\pi\rpack(\L_0^{(d)*})$, we get Lemma~\ref{lemma:goodlattice2}.
\qed
\section*{Appendix~F: Proof of Theorem~\ref{thm:vardistbound}}
The following lemma from~\cite{Ling13} will be used in the proof.
\begin{lemma}[\cite{Ling13}, Lemma~4]
	Let $\L$ be a lattice in $\R^d$. Then, for all $\z \in \R^d$, and $\kappa>0$,
	\[
		\frac{1-\epsilon_{\L}(\kappa)}{1+\epsilon_{\L}(\kappa)}\leq \frac{g_{\kappa,\z}(\L)}{g_{\kappa}(\L)} \leq 1.
	\]
	\label{lemma:g_sigma}
\end{lemma}

	For ease of notation, we will suppress the index $d$ in $\epsilon^{(d)}$, $\Lcd$ and $\Lfd$. 
	We will find upper and lower bounds for $p_{U+V}(\u)$ and $p_{U+V|\x}(\u)$, and then use these to get an upper bound on the absolute value of the difference between the two.

	For a message $X$ chosen at node $\tA$, let $\x$ be the coset representative of $X$ from $\Lf \cap \cV(\Lc)$.
	For any subset $S \subseteq \R^d$, let $\mathbf{1}_S(\cdot)$ 
	denote the indicator function of $S$, i.e., $\mathbf{1}_S(\u)$  is $1$ if $\u\in S$, and $0$ otherwise.
	 From (\ref{eq:pj_strong}), with $\kappa = \sqrt{\cP}$, we have
	\begin{equation}
		p_{U|\x}(\u)= \frac{g_{\sqrt{\cP}}(\u)}{g_{\sqrt{\cP},-\x}(\Lc)} \mathbf{1}_{\Lc+\x}(\u). 
	\label{eq:pux}	
	\end{equation}

	Let $\Gx:=\Lf\cap \cV(\Lc)$, and $M:=|\Gpd|=|\Gx|$. 
	Since the messages are uniformly distributed,
	\begin{equation}
		p_U(\u) =\sum_{\x\in \Gx} \frac{g_{\sqrt{\cP}}(\u)}{g_{\sqrt{\cP},-\x}(\Lc)}\frac{\mathbf{1}_{(\Lc+\x)}(\u)}{M}.  
                \label{eq:pu}
	\end{equation}
	\markchange{By monotonicity of the flatness factor,} $\epsilon_{\Lc}(\sqrt{\cP})<\epsilon_{\Lc}(\sqrt{\cP/2})=\epsilon$, and using Lemma~\ref{lemma:g_sigma},
	  \[
	       \frac{g_{\sqrt{\cP}}(\u)}{ g_{\sqrt{\cP}}(\Lc)} \leq \frac{g_{\sqrt{\cP}}(\u)}{g_{\sqrt{\cP},-\x}(\Lc)}\leq \frac{g_{\sqrt{\cP}}(\u)}{g_{\sqrt{\cP}}(\Lc)} \frac{1+\epsilon}{1-\epsilon}.
	  \]
     Using this in (\ref{eq:pu}), we get for $\u\in \L$,
	\begin{equation}
		\frac{g_{\sqrt{\cP}}(\u)}{M g_{\sqrt{\cP}}(\Lc)} \leq p_U(\u) \leq \frac{g_{\sqrt{\cP}}(\u)}{M g_{\sqrt{\cP}}(\Lc)} \frac{1+\epsilon}{1-\epsilon}.
		\label{eq:pu_bounds}
	\end{equation}
We will require bounds on $g_{\sqrt{\cP}}(\Lf)$ in the proof. 
Rearranging the terms above,
\[
	\left(\frac{1-\epsilon}{1+\epsilon}\right) p_U(\u)M g_{\sqrt{\cP}}(\Lc)  \leq g_{\sqrt{\cP}}(\u) \leq p_U(\u)M g_{\sqrt{\cP}}(\Lc).
\]
Since $p_{U}$ is a pmf supported over $\Lf$, and $\sum_{\u\in\Lf}p_{U}(\u)=1$, we can get
\begin{equation}
	\left(\frac{1-\epsilon}{1+\epsilon}\right) M g_{\sqrt{\cP}}(\Lc)  \leq g_{\sqrt{\cP}}(\Lf) \leq M g_{\sqrt{\cP}}(\Lc).
	\label{eq:8}
\end{equation}
It can be similarly verified that
for any $\a\in \R^n$,
\begin{equation}
	\left(\frac{1-\epsilon}{1+\epsilon}\right) M g_{\sqrt{\frac{\cP}{2}},\a}(\Lc)  \leq g_{\sqrt{\frac{\cP}{2}},\a}(\Lf) \leq M g_{\sqrt{\frac{\cP}{2}},\a}(\Lc).
	\label{eq:p_Lbound}
\end{equation}

We establish some more notation for convenience. Let
\begin{equation}
	\alpha(\w):=\frac{g_{\sqrt{2\cP}}(\w)}{Mg_{\sqrt{\cP}}(\Lc)}\frac{g_{\sqrt{\frac{\cP}{2}}}(\Lc)}{g_{\sqrt{\cP}}(\Lc)},
	\label{eq:defalpha}
\end{equation}
\begin{equation}
	\beta(\x,\w):=\left(\frac{g_{\sqrt{\frac{\cP}{2}},\frac{\w}{2}-\x}(\Lc)}{g_{\sqrt{\frac{\cP}{2}}}(\Lc)}\right) \left(\frac{g_{\sqrt{\cP},-\x}(\Lc)}{g_{\sqrt{\cP}}(\Lc)}\right)^{-1}.
	\label{eq:defnbeta}
\end{equation}
We can bound $p_{U+V|\x}$ and $p_{U+V}$ as follows.

\begin{lemma}
	For any lattice point $\w\in\Lf$, and any $\x\in\Gx$, we have 
	\begin{equation}
	\left( \frac{1-\epsilon}{1+\epsilon} \right) \alpha(\w) \leq p_{U+V}(\w) \leq \left( \frac{1+\epsilon}{1-\epsilon} \right)^2 \alpha(\w)
	\label{eq:puv_bound2}
\end{equation}
\begin{equation}
	\beta(\x,\w) \alpha(\w) \leq p_{U+V|\x}(\w) \leq \left( \frac{1+\epsilon}{1-\epsilon} \right) \beta(\x,\w)\alpha(\w).
	\label{eq:puvx_bound2}
\end{equation}
	\label{lemma:bounds}
\end{lemma}
\begin{proof}
	Let $\x$ be any fine lattice point from $\Gx$. Then,
\begin{align}
	p_{U+V|\x}(\w) &= \sum_{\t\in \Lc+\x}p_{U|\x}(\t)p_{V}(\w-\t) .& \notag 
\end{align}	
Using (\ref{eq:pux}) and (\ref{eq:pu_bounds}) in the above equation, we obtain
\begin{equation}
	\sum_{\t\in \Lc+\x}\frac{g_{\sqrt{\cP}}(\t)}{g_{\sqrt{\cP},-\x}(\Lc)}\frac{g_{\sqrt{\cP}}(\w-\t)}{M g_{\sqrt{\cP}}(\Lc)} \leq p_{U+V|\x}(\w) \leq \sum_{\t\in \Lc+\x}\frac{g_{\sqrt{\cP}}(\t)}{g_{\sqrt{\cP},-\x}(\Lc)}\frac{g_{\sqrt{\cP}}(\w-\t)}{M g_{\sqrt{\cP}}(\Lc)} \left(\frac{1+\epsilon}{1-\epsilon}\right).
	\label{eq:1}
\end{equation}
Consider the term
\begin{align}
 \sum_{\t\in \Lc+\x}\frac{g_{\sqrt{\cP}}(\t)}{g_{\sqrt{\cP},-\x}(\Lc)}\frac{g_{\sqrt{\cP}}(\w-\t)}{M g_{\sqrt{\cP}}(\Lc)} 
	&=\frac{1}{g_{\sqrt{\cP},-\x}(\Lc)}\frac{1}{M g_{\sqrt{\cP}}(\Lc)}\sum_{\t\in \Lc+\x}\frac{e^{\left( -\frac{\Vert \t \Vert^2}{2\cP} -\frac{\Vert \t-\w \Vert^2}{2\cP}\right)}}{(2\pi\markchange{\cP})^d} & \notag \\
	&=\frac{1}{M g_{\sqrt{\cP}}(\Lc)g_{\sqrt{\cP},-\x}(\Lc)}\sum_{\t\in \Lc+\x}\frac{e^{\left( -\frac{\Vert \w \Vert^2}{4\cP} -\frac{\Vert \t-\frac{\w}{2} \Vert^2}{\cP}\right)}}{(2\pi\markchange{\cP})^d} & \notag \\
	&=\frac{g_{\sqrt{2\cP}}(\w)}{M g_{\sqrt{\cP}}(\Lc)g_{\sqrt{\cP},-\x}(\Lc)}\sum_{\t\in \Lc+\x} g_{\sqrt{\frac{\cP}{2}},\frac{\w}{2}}(\t) & \notag \\
	&=\frac{g_{\sqrt{2\cP}}(\w)}{M g_{\sqrt{\cP}}(\Lc)} \frac{g_{\sqrt{\frac{\cP}{2}},\frac{\w}{2}-\x}(\Lc)}{g_{\sqrt{\cP},-\x}(\Lc)} . &
\end{align}

Substituting this in (\ref{eq:1}), and writing this in terms of $\alpha$ and $\beta$, we obtain (\ref{eq:puvx_bound2}).
Similarly, bounding both $p_{U}$ and $p_V$ from above and below using (\ref{eq:pu_bounds}), proceeding as above,
and finally using (\ref{eq:p_Lbound}) to bound $g_{\sqrt{\frac{\cP}{2}},\frac{\w}{2}}(\Lf)$, we get (\ref{eq:puv_bound2}).
\end{proof}

Observe that $\beta(\x,\w)$ in (\ref{eq:defnbeta}) is a ratio of two terms, both of which can be bounded using Lemma~\ref{lemma:g_sigma} to get
\begin{equation}
	\left( \frac{1-\epsilon}{1+\epsilon} \right)\leq \beta(\x,\w)\leq  \left( \frac{1+\epsilon}{1-\epsilon} \right).
	\label{eq:boundbeta}
\end{equation}

Let $\overline{p}_{U+V}$ and $\underline{p}_{U+V}$ respectively denote the upper and lower bounds for $p_{U+V}$ in (\ref{eq:puv_bound2}), and let $\overline{p}_{U+V|\x}$ and $\underline{p}_{U+V|\x}$ respectively denote the upper and lower bounds for $p_{U+V|\x}$ in (\ref{eq:puvx_bound2}).
Then, we can say that
$\vert p_{U+V|\x}(\w)-p_{U+V}(\w) \vert$ is less than or equal to the maximum of $\vert\overline{p}_{U+V|\x}(\w)-\underline{p}_{U+V}(\w)\vert$ and $\vert\underline{p}_{U+V|\x}(\w)-\overline{p}_{U+V}(\w)\vert $. 
	
Substituting for $\vert \overline{p}_{U+V|\x}(\w)-\underline{p}_{U+V}(\w) \vert$, we get
\begin{equation}
     \vert \overline{p}_{U+V|\x}(\w)-\underline{p}_{U+V}(\w) \vert  = \alpha(\w)\left(\frac{1-\epsilon}{1+\epsilon} \right)\left\vert \left(\frac{1+\epsilon}{1-\epsilon} \right)^2 \beta(\x,\w) - 1  \right\vert . \label{eq:6}
\end{equation}
However, from (\ref{eq:boundbeta}), we see that
\[
	1< \left(\frac{1+\epsilon}{1-\epsilon} \right) \le \left(\frac{1+\epsilon}{1-\epsilon} \right)^2\beta(\x,\w) \le \left(\frac{1+\epsilon}{1-\epsilon} \right)^3,
\]
and for $\epsilon\le 1/2$, we have $\left(\frac{1+\epsilon}{1-\epsilon} \right)^3\le 1+64\epsilon$. Therefore,
\begin{equation}
	\vert \overline{p}_{U+V|\x}(\w)-\underline{p}_{U+V}(\w) \vert \le \alpha(\w)\left(\frac{1-\epsilon}{1+\epsilon} \right) 64\epsilon.
	\label{eq:diff_b1}
\end{equation}
Similarly, expressing $\vert \underline{p}_{U+V|\x}(\w)-\overline{p}_{U+V}(\w) \vert$ in terms of $\alpha$ and $\beta$, and using the fact that $\left((1-\epsilon)/(1+\epsilon) \right)^3\geq 1-8\epsilon$ for $\epsilon<1/2$, we get
\begin{equation}
	\vert \underline{p}_{U+V|\x}(\w)-\overline{p}_{U+V}(\w) \vert \le \alpha(\w)\left(\frac{1+\epsilon}{1-\epsilon} \right)^2 8\epsilon.
	\label{eq:diff_b2}
\end{equation}
Rearranging (\ref{eq:puv_bound2}), and observing that $\sum_{\w\in\Lf}p_{U+V}(\w)=1$, we have
\begin{equation}
	\left( \frac{1-\epsilon}{1+\epsilon} \right)^2\leq \sum_{\w\in\Lf}\alpha(\w)\leq \left( \frac{1+\epsilon}{1-\epsilon} \right).
	\label{eq:sum_alpha}
\end{equation}
Combining (\ref{eq:diff_b1}) and (\ref{eq:diff_b2}), and summing over $\w$, we get 
\begin{align}
	\mathbb{V}(p_{U+V},p_{U+V|\x}) & \leq \sum_{\w\in\Lf}\alpha(\w)\max \left\{ \left(\frac{1-\epsilon}{1+\epsilon} \right) 64\epsilon , \left(\frac{1+\epsilon}{1-\epsilon} \right)^2 8\epsilon \right\},  &\notag 
\end{align}
and using (\ref{eq:sum_alpha}) to bound $\sum_{\w\in\Lf}\alpha(\w)$ from above, we get
\begin{align}
\mathbb{V}(p_{U+V},p_{U+V|\x}) & \leq \max \left\{  64\epsilon , \left(\frac{1+\epsilon}{1-\epsilon} \right)^3 8\epsilon \right\} \leq \max \left\{ 64\epsilon , 27\times 8\epsilon \right\}, &\notag
\end{align}
since $\epsilon\leq 1/2$. Therefore,
\[
	\mathbb{V}(p_{U+V},p_{U+V|\x})\leq 216\epsilon,
\]
thereby completing the proof.\qed

\section*{Appendix~G: Proof of Lemma~\ref{lemma:MMSEeff_noise}}

\markchange{Recall that $\bZ$ is the additive Gaussian noise vector in the MAC phase having mean zero and variance $\nsvar$, and $\bZ'$ denotes the additive noise in the effective MLAN channel, and is equal to $(\cmmse-1)(\U+\V)+\cmmse \bZ$. Let $\mathbf{N}$ denote a zero-mean Gaussian vector with covariance matrix $((1-\cmmse)^22\cP+(\cmmse)^2 \nsvar)\mathsf{I}_d$, and $\mathbf{N}'$ denote a zero-mean Gaussian vector with covariance matrix $((1-\cmmse)^2\cP+(\cmmse)^2\nsvar)\mathsf{I}_d$.
Let $f_{\mathbf{N}}$ and $f_{\mathbf{N}'}$ denote the densities of $\mathbf{N}$ and $\mathbf{N}'$ respectively, and $f_{\U'|\x}$ denote the density of  $\U':=(\cmmse-1)\U+\mathbf{N}'$ conditioned on $\mathbf{X}=\x$. Let $f_{\V'|\y}$ denote the density function of $\V':=(\cmmse-1)\V+\cmmse\bZ$ conditioned on $\mathbf{Y}=\y$. Then, we can write
\[
    \mathbb{V}(f_{\bZ'|\x,\y},f_{\mathbf{N}})\leq  \mathbb{V}(f_{\bZ'|\x,\y},f_{\U'|\x})+ \mathbb{V}(f_{\U'|\x},f_{\mathbf{N}}).
\]
But
\begin{align*}
     \mathbb{V}(f_{\bZ'|\x,\y},f_{\U'|\x}) &= \int_{\w\in\R^d} \vert f_{Z'|\x,\y}(\w)-f_{\U'|\x}(\w) \vert \: d\w & \\
                                                                         &= \int_{\w\in\R^d}\biggl\vert \sum_{\u\in\Lcd+\x} p_{\U|\x}(\u)\Big(f_{\V'|\y}(\w-(\cmmse-1)\u)-f_{\mathbf{N}'}(\w-(\cmmse-1)\u)\Big) \biggr\vert \: d\w & \\
                                                                         &\leq \sum_{\u\in\Lcd+\x}p_{\U|\x}(\u) \left( \int_{\w\in\R^d} \left\vert f_{\V'|\y}(\w-(\cmmse-1)\u)-f_{\mathbf{N}'}(\w-(\cmmse-1)\u) \right\vert \:  d\w \right) & \\
                                                                         &= \sum_{\u\in\Lcd+\x}p_{\U|\x}(\u) \mathbb{V}(f_{\V'|\y},f_{\mathbf{N}'}) &\\
                                                                         & =  \mathbb{V}(f_{\V'|\y},f_{\mathbf{N}'}).
\end{align*}
Therefore, 
\begin{equation}
  \mathbb{V}(f_{\bZ'|\x,\y},f_{\mathbf{N}})\leq   \mathbb{V}(f_{\V'|\y},f_{\mathbf{N}'})+\mathbb{V}(f_{\U'|\x},f_{\mathbf{N}}).
  \label{eq:vdist_mmsenoise}
\end{equation}
\begin{lemma}[\cite{Ling13}, Lemma 8]
  Let $\L$ be a lattice in $\R^d$, $\x\in\R^d$, and $\sigma_1,\sigma_2>0$. Let $\U$ be a random vector supported on $\L+\x$, having pmf $g_{\sigma_1}(\u)/g_{\sigma_1,\x}(\L)$. If $\bZ$ is an iid Gaussian random vector with mean zero and variance $\sigma_2^2$, and $\epsilon_{\L}\left(\frac{\sigma_1\sigma_2}{\sqrt{\sigma_1^2+\sigma_2^2}}\right)<1/2$, then the density of  $\U+\bZ$, $f_{\U+\bZ}$, satisfies
\[
  \mathbb{V}(f_{\U+\bZ},g_{\sqrt{\sigma_1^2+\sigma_2^2}})\leq 4 \, \epsilon_{\L}\left(\frac{\sigma_1\sigma_2}{\sqrt{\sigma_1^2+\sigma_2^2}}\right).
\]
\label{lemm:gauss_approx}
\end{lemma}
}

\markchange{
 Using Lemma~\ref{lemm:gauss_approx} and the fact that for any constant $a>0$, $\epsilon_{a\Lcd}(a\theta)=\epsilon_{\Lcd}(\theta)$~\cite[Remark 4]{Ling13}, we get
\begin{align}
    \mathbb{V}(f_{\V'|\y},f_{\mathbf{N}'}) &\leq  4\epsilon_{\Lcd}\left( \frac{\cmmse\sqrt{\cP\nsvar} }{\sqrt{(1-\cmmse)^2\cP+(\cmmse)^2\nsvar}} \right) & \\
                                                                               &\leq 4\epsilon_{\Lcd}\left( \frac{\cmmse\sqrt{\cP\nsvar} }{\sqrt{(1-\cmmse)^2 2\cP+(\cmmse)^2\nsvar}} \right) &\label{eq:var_1} \\
                                                                               &= 4\epsilon_{\Lcd}\left( \sqrt{\cmmse\cP} \right), &\label{eq:var_3}
\end{align}
where (\ref{eq:var_1}) is by the monotonicity of the flatness factor. Equation (\ref{eq:var_3}) is then obtained by substituting $\cmmse=2\cP/(2\cP+\nsvar)$ and simplifying. By similar arguments, we can show that 
\[
     \mathbb{V}(f_{\U'|\x},f_{\mathbf{N}})\leq 4\epsilon_{\Lcd}\left( \sqrt{\cmmse \cP} \right).
\]
Substituting in (\ref{eq:vdist_mmsenoise}) completes the proof. \qed
}

\section*{Appendix~H: Proof of Lemma~\ref{lemma:mi_multihopbound}}

\markchange{We want to show that $ \mathcal{I}(X_1,\ldots,X_N ; \Theta_{k,N})$ is arbitrarily small for all sufficiently large $d$. Using the chain rule of mutual information, and making some observations about the conditional independence of these random variables, we will show that this quantity can be written as a sum of mutual information terms between the $i$th message, $X_i$, and the vector $\mathbf{W}_k[2i+k-1]$, conditioned on everything observed by the $k$th relay in the first $2i+k-2$ phases. We will then bound each of these mutual information terms from above by a quantity of the form $\mathcal{I}(X;\mathcal{E}(X)+\mathcal{E}(Y))$, so that we can invoke the results of Section~\ref{sec:strongsecr} to conclude that each of these terms go to zero as $d\to\infty$. We would like to remark that the techniques used in this proof hold good for any coding scheme that achieves strong secrecy over the bidirectional relay, and in particular, the one in~\cite{HeYenerstrong}.}
 
Making repeated use of the chain rule of mutual information, we see that
\begin{align}
     \mathcal{I}(X_1,\ldots,X_N ; \Theta_{k,N})&  =\sum_{t=1}^{N}\mathcal{I}(X_t;\Theta_{k,N}\big\vert X_1,\ldots,X_{t-1}) &\notag\\
									  & =\sum_{t=1}^{N}\left[ \mathcal{I}(X_t;J_{k},J_{k-1}\big\vert X_1,\ldots,X_{t-1})+ \sum_{n=1}^{N} \mathcal{I}(X_t;\W_{k}[2n+k-1]\big\vert X_1,\ldots,X_{t-1},\Theta_{k,n-1}) \right] & \notag \\
									  & =\sum_{t=1}^{N} \sum_{n=1}^{N} \mathcal{I}(X_t;\W_{k}[2n+k-1]\big\vert X_1,\ldots,X_{t-1},\Theta_{k,n-1}), &\label {eq:mi}
\end{align}
where the last step follows from the fact that $ \mathcal{I}(X_t;J_{k},J_{k-1}\big\vert X_1,\ldots,X_{t-1})=0$ for $1\leq t\leq N$, since the messages and the jamming signals are independent.
%

We will first show that many terms in the above summation are zero. We will make use of the fact that if $X,Y,$ and $Z$ are random variables distributed over a finite group $\Gp$, with $X$ being uniformly distributed over $\Gp$ and independent of $(Y,Z)$, then $X\oplus Y$ is uniformly distributed over $\Gp$ and independent of $Z$. Observe that for $n\in \{ 1,2,\ldots,N \}$, $\Theta_{k,n-1}$ consists of random variables which are all functions of $X_1,\ldots,X_{n-1}$ and $J_{k-1},\ldots, J_{k+n-1}$, which are all independent of $X_{t}$ for $n\leq t$ (even when conditioned on the first $l-1< t$ messages). Therefore,

\begin{proposition}
     Let $1\leq t\leq N$, and $n,l\in \{ 1,2,\ldots, t\}$. Then, the message $X_t$ is conditionally independent of $\Theta_{k,n-1}$ given $X_1,X_2,\ldots,X_{l-1}$.
     \label{prop:multihop1}
\end{proposition}


Using a similar argument, we obtain
\begin{proposition}
     Let $1\leq t< n\leq N$.The vector $\W_{k}[2n+k-1]$ received by the $k$th relay in the $(2n+k-1)$st phase is independent of $X_1,\ldots,X_t$ and $\Theta_{k,n-1}$. 
     \label{prop:multihop2}
\end{proposition}
%

We now evaluate the terms in (\ref{eq:mi}). 
Using Proposition~\ref{prop:multihop1}, we get
\begin{equation}
	  \mathcal{I}(X_t;\W_{k}[2n+k-1]\big\vert X_1,\ldots,X_{t-1},\Theta_{k,n-1})=0
\end{equation}
for all $1\leq n<t\leq N$.
Similarly, using Proposition~\ref{prop:multihop2}, 
\begin{equation}
	  \mathcal{I}(X_t;\W_{k}[2n+k-1]\big\vert X_1,\ldots,X_{t-1},\Theta_{k,n-1})=0
\end{equation}
for all $1\leq t<n\leq N$.  Therefore, (\ref{eq:mi}) reduces to
\begin{equation}
     \mathcal{I}(X_1,\ldots,X_N ; \Theta_{k,N}) = \sum_{t=1}^{N} \mathcal{I}(X_t;\W_{k}[2t+k-1]\big\vert X_1,\ldots,X_{t-1},\Theta_{k,t-1}). 
     \label{eq:mi_multihop1}
\end{equation}

The mutual information $\mathcal{I}(X_t;\W_{k}[2t+k-1]\big\vert X_1,\ldots,X_{t-1},\Theta_{k,t-1})$ can be written in terms of conditional entropies as
\begin{align}
     \mathcal{I}(X_t;\W_{k}[2t+k-1]\big\vert X_1,\ldots,X_{t-1},\Theta_{k,t-1})&=\mathcal{H}(\W_k[2t+k+1]\big\vert X_1,\ldots,X_{t-1},\Theta_{k,t-1}) & \notag \\
                                                                                                                     &  \qquad -\mathcal{H}(\W_k[2t+k+1]\big\vert X_1,\ldots,X_{t},\Theta_{k,t-1}). & \label{eq:AH_2}
\end{align}
Let us evaluate each of the terms on the right hand side.  Consider the second term, 
\begin{align}
     \mathcal{H}(\W_k[2t+k+1]\big\vert X_1,\ldots,X_{t},\Theta_{k,t-1}) &\geq \mathcal{H}(\W_k[2t+k+1]\big\vert X_1,\ldots,X_{t},\oplus_{p=1}^{t}X_p\oplus J_{k+t-1},\Theta_{k,t-1}) &\notag \\
										       &= \mathcal{H}(\W_k[2t+k+1]\big\vert \oplus_{p=1}^{t}X_p\oplus J_{k+t-1}). &\label{eq:mi_multihop2}
\end{align}
The first step is true because conditioning reduces entropy.
The second step requires more justification. 
Given $\oplus_{p=1}^{t}X_p\oplus J_{k+t-1}$, the term $\V_{k-1}[2t+k-1]$ is independent of $X_{1},\ldots,X_t,\Theta_{k,t-1}$. The jamming signal, $J_{k+t}$ is independent of $\Theta_{k,t-1}$,
all the first $t$ messages, and $\oplus_{p=1}^{t}X_p\oplus J_{k+t-1}$. Therefore,  $\V_{k+1}[2t+k-1]$, and hence, $\W_k[2t+k-1]$ is also independent of $\Theta_{k,t-1}$, the first $t$ messages and $\oplus_{p=1}^{t}X_p\oplus J_{k+t-1}$, thus justifying (\ref{eq:mi_multihop2}). 
Now, define $X:=\oplus_{p=1}^{t}X_p\oplus J_{k+t-1}$, and $Y:=\oplus_{p=1}^{t-1}X_p\oplus J_{k+t}$. 
Then, we have,
\[
     \mathcal{H}(\W_k[2t+k+1]\big\vert X_1,\ldots,X_{t},\Theta_{k,t-1})\geq \mathcal{H}(\mathcal{E}(X)+\mathcal{E}(Y)\big\vert X).
\]
From Proposition~\ref{prop:multihop2}, the first term of (\ref{eq:AH_2}), $\mathcal{H}(\W_k[2t+k+1]\big\vert X_1,\ldots,X_{t-1},\Theta_{k,t-1}) =\mathcal{H}\left( \mathcal{E}(X)+\mathcal{E}(Y) \right)$. 
Therefore, $ \mathcal{I}(X_t;\W_{k}[2t+k-1]\big\vert X_1,\ldots,X_{t-1},\Theta_{k,t-1})$ is bounded above by $\mathcal{I}(X;\mathcal{E}(X)+\mathcal{E}(Y))$, and the random variables $X$ and $Y$ are independent and uniformly distributed over $\Gpd$.
The lemma now follows by using Theorem~\ref{thm:vardistbound} and Lemma~\ref{lemma:csiszar} to bound this quantity. 
\qed

\section*{Acknowledgement} We would like to thank the anonymous reviewers for carefully reading our manuscript and suggesting several improvements to the presentation. In particular, we thank the reviewer who suggested a means of avoiding the use of a random dither to achieve the rate in Section~\ref{sec:strongsecr}. We are also grateful to Manjunath Krishnapur for providing Proposition~\ref{prop:exp_decay} and its proof.

\end{document}

%% file: bidirection.pdf_t
\begin{picture}(0,0)%
\includegraphics{bidirection.pdf}%
\end{picture}%
\setlength{\unitlength}{4144sp}%
\begingroup\makeatletter\ifx\SetFigFont\undefined%
\gdef\SetFigFont#1#2#3#4#5{%
  \reset@font\fontsize{#1}{#2pt}%
  \fontfamily{#3}\fontseries{#4}\fontshape{#5}%
  \selectfont}%
\fi\endgroup%
\begin{picture}(9930,1828)(-4964,-75)
\put(-4049,614){\makebox(0,0)[b]{\smash{{\SetFigFont{45}{54.0}{\rmdefault}{\mddefault}{\updefault}{\color[rgb]{0,0,0}$\mathtt{A}$}%
}}}}
\put(  1,614){\makebox(0,0)[b]{\smash{{\SetFigFont{45}{54.0}{\rmdefault}{\mddefault}{\updefault}{\color[rgb]{0,0,0}$\mathtt{R}$}%
}}}}
\put(4051,614){\makebox(0,0)[b]{\smash{{\SetFigFont{45}{54.0}{\rmdefault}{\mddefault}{\updefault}{\color[rgb]{0,0,0}$\mathtt{B}$}%
}}}}
\end{picture}%

%% file: bidirect1.pdf_t
\begin{picture}(0,0)%
\includegraphics{bidirect1.pdf}%
\end{picture}%
\setlength{\unitlength}{4144sp}%
\begingroup\makeatletter\ifx\SetFigFont\undefined%
\gdef\SetFigFont#1#2#3#4#5{%
  \reset@font\fontsize{#1}{#2pt}%
  \fontfamily{#3}\fontseries{#4}\fontshape{#5}%
  \selectfont}%
\fi\endgroup%
\begin{picture}(11339,12613)(1104,-12233)
\put(12016,-1546){\makebox(0,0)[lb]{\smash{{\SetFigFont{45}{54.0}{\familydefault}{\mddefault}{\updefault}{\color[rgb]{0,0,0}$Y$}%
}}}}
\put(6616,-196){\makebox(0,0)[lb]{\smash{{\SetFigFont{45}{54.0}{\familydefault}{\mddefault}{\updefault}{\color[rgb]{0,0,0}$\mathbf{z}$}%
}}}}
\put(6886,-3526){\makebox(0,0)[lb]{\smash{{\SetFigFont{45}{54.0}{\familydefault}{\mddefault}{\updefault}{\color[rgb]{0,0,0}$\mathbf{w}$}%
}}}}
\put(4771,-1546){\makebox(0,0)[lb]{\smash{{\SetFigFont{45}{54.0}{\familydefault}{\mddefault}{\updefault}{\color[rgb]{0,0,0}$\mathbf{u}$}%
}}}}
\put(6526,-5191){\makebox(0,0)[lb]{\smash{{\SetFigFont{45}{54.0}{\familydefault}{\mddefault}{\updefault}{\color[rgb]{0,0,0}$\mathtt{R}$}%
}}}}
\put(6616,-6721){\makebox(0,0)[lb]{\smash{{\SetFigFont{45}{54.0}{\familydefault}{\mddefault}{\updefault}{\color[rgb]{0,0,0}$\mathbf{z}_{b}$}%
}}}}
\put(6526,-11851){\makebox(0,0)[lb]{\smash{{\SetFigFont{45}{54.0}{\familydefault}{\mddefault}{\updefault}{\color[rgb]{0,0,0}$\mathtt{R}$}%
}}}}
\put(9991,-8701){\makebox(0,0)[lb]{\smash{{\SetFigFont{45}{54.0}{\familydefault}{\mddefault}{\updefault}{\color[rgb]{0,0,0}$\mathtt{B}$}%
}}}}
\put(8326,-1546){\makebox(0,0)[lb]{\smash{{\SetFigFont{45}{54.0}{\familydefault}{\mddefault}{\updefault}{\color[rgb]{0,0,0}$\mathbf{v}$}%
}}}}
\put(11836,-7981){\makebox(0,0)[lb]{\smash{{\SetFigFont{45}{54.0}{\familydefault}{\mddefault}{\updefault}{\color[rgb]{0,0,0}$\hat{X}$}%
}}}}
\put(2116,-7891){\makebox(0,0)[lb]{\smash{{\SetFigFont{34}{40.8}{\rmdefault}{\mddefault}{\updefault}{\color[rgb]{0,0,0}User node}%
}}}}
\put(2116,-1411){\makebox(0,0)[lb]{\smash{{\SetFigFont{34}{40.8}{\rmdefault}{\mddefault}{\updefault}{\color[rgb]{0,0,0}User node}%
}}}}
\put(9046,-1456){\makebox(0,0)[lb]{\smash{{\SetFigFont{34}{40.8}{\rmdefault}{\mddefault}{\updefault}{\color[rgb]{0,0,0}User node}%
}}}}
\put(6076,-4516){\makebox(0,0)[lb]{\smash{{\SetFigFont{34}{40.8}{\rmdefault}{\mddefault}{\updefault}{\color[rgb]{0,0,0}Relay}%
}}}}
\put(6076,-10996){\makebox(0,0)[lb]{\smash{{\SetFigFont{34}{40.8}{\rmdefault}{\mddefault}{\updefault}{\color[rgb]{0,0,0}Relay}%
}}}}
\put(1171,-7936){\makebox(0,0)[lb]{\smash{{\SetFigFont{45}{54.0}{\familydefault}{\mddefault}{\updefault}{\color[rgb]{0,0,0}$\hat{Y}$}%
}}}}
\put(6886,-10006){\makebox(0,0)[lb]{\smash{{\SetFigFont{45}{54.0}{\familydefault}{\mddefault}{\updefault}{\color[rgb]{0,0,0}$\mathbf{w}_{b}$}%
}}}}
\put(3241,-8656){\makebox(0,0)[b]{\smash{{\SetFigFont{45}{54.0}{\familydefault}{\mddefault}{\updefault}{\color[rgb]{0,0,0}$\mathtt{A}$}%
}}}}
\put(1306,-1456){\makebox(0,0)[b]{\smash{{\SetFigFont{45}{54.0}{\familydefault}{\mddefault}{\updefault}{\color[rgb]{0,0,0}$X$}%
}}}}
\put(6796,-8386){\makebox(0,0)[b]{\smash{{\SetFigFont{50}{60.0}{\familydefault}{\mddefault}{\updefault}{\color[rgb]{0,0,0}+}%
}}}}
\put(6796,-1861){\makebox(0,0)[b]{\smash{{\SetFigFont{50}{60.0}{\familydefault}{\mddefault}{\updefault}{\color[rgb]{0,0,0}+}%
}}}}
\put(2926,-2131){\makebox(0,0)[lb]{\smash{{\SetFigFont{45}{54.0}{\familydefault}{\mddefault}{\updefault}{\color[rgb]{0,0,0}$\mathtt{A}$}%
}}}}
\put(9811,-2176){\makebox(0,0)[lb]{\smash{{\SetFigFont{45}{54.0}{\familydefault}{\mddefault}{\updefault}{\color[rgb]{0,0,0}$\mathtt{B}$}%
}}}}
\put(9046,-7936){\makebox(0,0)[lb]{\smash{{\SetFigFont{34}{40.8}{\rmdefault}{\mddefault}{\updefault}{\color[rgb]{0,0,0}User node}%
}}}}
\end{picture}%

%% file: modLfd.pdf_t
\begin{picture}(0,0)%
\includegraphics{modLfd.pdf}%
\end{picture}%
\setlength{\unitlength}{4144sp}%
\begingroup\makeatletter\ifx\SetFigFont\undefined%
\gdef\SetFigFont#1#2#3#4#5{%
  \reset@font\fontsize{#1}{#2pt}%
  \fontfamily{#3}\fontseries{#4}\fontshape{#5}%
  \selectfont}%
\fi\endgroup%
\begin{picture}(9249,8394)(-4511,-3043)
\put(136,479){\makebox(0,0)[lb]{\smash{{\SetFigFont{20}{24.0}{\rmdefault}{\mddefault}{\updefault}{\color[rgb]{0,0,0}$\mathbf{0}$}%
}}}}
\put(1486,3629){\makebox(0,0)[lb]{\smash{{\SetFigFont{20}{24.0}{\rmdefault}{\mddefault}{\updefault}{\color[rgb]{0,0,0}$\mathbf{x}$}%
}}}}
\put(-2069,1649){\makebox(0,0)[lb]{\smash{{\SetFigFont{20}{24.0}{\rmdefault}{\mddefault}{\updefault}{\color[rgb]{0,0,0}$[\mathbf{x}]\bmod \Lambda$}%
}}}}
\put(2161,2594){\makebox(0,0)[lb]{\smash{{\SetFigFont{20}{24.0}{\rmdefault}{\mddefault}{\updefault}{\color[rgb]{0,0,0}$Q_{\Lambda}(\mathbf{x})$}%
}}}}
\end{picture}%

%% file: rcov_rpack.pdf_t
\begin{picture}(0,0)%
\includegraphics{rcov_rpack.pdf}%
\end{picture}%
\setlength{\unitlength}{4144sp}%
\begingroup\makeatletter\ifx\SetFigFont\undefined%
\gdef\SetFigFont#1#2#3#4#5{%
  \reset@font\fontsize{#1}{#2pt}%
  \fontfamily{#3}\fontseries{#4}\fontshape{#5}%
  \selectfont}%
\fi\endgroup%
\begin{picture}(5010,5742)(-2504,-2032)
\put(1891,2549){\makebox(0,0)[lb]{\smash{{\SetFigFont{20}{24.0}{\rmdefault}{\mddefault}{\updefault}{\color[rgb]{0,0,0}$\mathcal{V}(\Lambda)$}%
}}}}
\put(1666,164){\makebox(0,0)[lb]{\smash{{\SetFigFont{20}{24.0}{\rmdefault}{\mddefault}{\updefault}{\color[rgb]{0,0,0}$r_{\text{pack}}(\Lambda)$}%
}}}}
\put(-134,1469){\makebox(0,0)[lb]{\smash{{\SetFigFont{20}{24.0}{\rmdefault}{\mddefault}{\updefault}{\color[rgb]{0,0,0}$r_{\text{cov}}(\Lambda)$}%
}}}}
\put(-1304,569){\makebox(0,0)[lb]{\smash{{\SetFigFont{20}{24.0}{\rmdefault}{\mddefault}{\updefault}{\color[rgb]{0,0,0}$r_{\text{eff}}(\Lambda)$}%
}}}}
\end{picture}%

%% file: isit_generic_psi.pdf_t
\begin{picture}(0,0)%
\epsfig{file=isit_generic_psi.pdf}%
\end{picture}%
\setlength{\unitlength}{1776sp}%
\begingroup\makeatletter\ifx\SetFigFont\undefined%
\gdef\SetFigFont#1#2#3#4#5{%
  \reset@font\fontsize{#1}{#2pt}%
  \fontfamily{#3}\fontseries{#4}\fontshape{#5}%
  \selectfont}%
\fi\endgroup%
\begin{picture}(5512,2346)(3279,-1939)
\put(8776,-1561){\makebox(0,0)[lb]{\smash{{\SetFigFont{8}{9.6}{\familydefault}{\mddefault}{\updefault}{\color[rgb]{0,0,0}$t$}%
}}}}
\put(6076,164){\makebox(0,0)[lb]{\smash{{\SetFigFont{8}{9.6}{\familydefault}{\mddefault}{\updefault}{\color[rgb]{0,0,0}$\psi(t)$}%
}}}}
\end{picture}%

%% file: phi.pstex_t
\begin{picture}(0,0)%
\epsfig{file=phi.pdf}%
\end{picture}%
\setlength{\unitlength}{2289sp}%
\begingroup\makeatletter\ifx\SetFigFont\undefined%
\gdef\SetFigFont#1#2#3#4#5{%
  \reset@font\fontsize{#1}{#2pt}%
  \fontfamily{#3}\fontseries{#4}\fontshape{#5}%
  \selectfont}%
\fi\endgroup%
\begin{picture}(12102,2346)(-11,-1939)
\put(6076,164){\makebox(0,0)[lb]{\smash{{\SetFigFont{10}{12.0}{\familydefault}{\mddefault}{\updefault}{\color[rgb]{0,0,0}$\varphi(t)$}%
}}}}
\put(12076,-1636){\makebox(0,0)[lb]{\smash{{\SetFigFont{10}{12.0}{\familydefault}{\mddefault}{\updefault}{\color[rgb]{0,0,0}$t$}%
}}}}
\end{picture}%

%% file: phi0_phi1.pdf_t
\begin{picture}(0,0)%
\includegraphics{phi0_phi1.pdf}%
\end{picture}%
\setlength{\unitlength}{2368sp}%
\begingroup\makeatletter\ifx\SetFigFont\undefined%
\gdef\SetFigFont#1#2#3#4#5{%
  \reset@font\fontsize{#1}{#2pt}%
  \fontfamily{#3}\fontseries{#4}\fontshape{#5}%
  \selectfont}%
\fi\endgroup%
\begin{picture}(12102,6330)(-11,-6073)
\put(12076,-1786){\makebox(0,0)[lb]{\smash{{\SetFigFont{11}{13.2}{\familydefault}{\mddefault}{\updefault}{\color[rgb]{0,0,0}$t$}%
}}}}
\put(6076, 14){\makebox(0,0)[lb]{\smash{{\SetFigFont{11}{13.2}{\familydefault}{\mddefault}{\updefault}{\color[rgb]{0,0,0}$\varphi_0(t)$}%
}}}}
\put(12076,-4486){\makebox(0,0)[lb]{\smash{{\SetFigFont{11}{13.2}{\familydefault}{\mddefault}{\updefault}{\color[rgb]{0,0,0}$t$}%
}}}}
\put(6076,-2836){\makebox(0,0)[lb]{\smash{{\SetFigFont{11}{13.2}{\familydefault}{\mddefault}{\updefault}{\color[rgb]{0,0,0}$\varphi_1(t)$}%
}}}}
\end{picture}%

%% file: triangle1.pdf_t
\begin{picture}(0,0)%
\epsfig{file=triangle1.pdf}%
\end{picture}%
\setlength{\unitlength}{2960sp}%
\begingroup\makeatletter\ifx\SetFigFont\undefined%
\gdef\SetFigFont#1#2#3#4#5{%
  \reset@font\fontsize{#1}{#2pt}%
  \fontfamily{#3}\fontseries{#4}\fontshape{#5}%
  \selectfont}%
\fi\endgroup%
\begin{picture}(3712,2359)(4179,-1952)
\put(6826,-1861){\makebox(0,0)[lb]{\smash{{\SetFigFont{14}{16.8}{\familydefault}{\mddefault}{\updefault}{\color[rgb]{0,0,0}$1$}%
}}}}
\put(6076,164){\makebox(0,0)[lb]{\smash{{\SetFigFont{14}{16.8}{\familydefault}{\mddefault}{\updefault}{\color[rgb]{0,0,0}$\hat{f}(t)$}%
}}}}
\put(7876,-1561){\makebox(0,0)[lb]{\smash{{\SetFigFont{14}{16.8}{\familydefault}{\mddefault}{\updefault}{\color[rgb]{0,0,0}$t$}%
}}}}
\put(4876,-1861){\makebox(0,0)[lb]{\smash{{\SetFigFont{14}{16.8}{\familydefault}{\mddefault}{\updefault}{\color[rgb]{0,0,0}$-1$}%
}}}}
\end{picture}%

%% file: relayop.pdf_t
\begin{picture}(0,0)%
\includegraphics{relayop.pdf}%
\end{picture}%
\setlength{\unitlength}{4144sp}%
\begingroup\makeatletter\ifx\SetFigFont\undefined%
\gdef\SetFigFont#1#2#3#4#5{%
  \reset@font\fontsize{#1}{#2pt}%
  \fontfamily{#3}\fontseries{#4}\fontshape{#5}%
  \selectfont}%
\fi\endgroup%
\begin{picture}(12794,3777)(-5279,1682)
\put(2566,4574){\makebox(0,0)[lb]{\smash{{\SetFigFont{29}{34.8}{\rmdefault}{\mddefault}{\updefault}{\color[rgb]{0,0,0}Relay $\mathtt{R}$}%
}}}}
\put(-179,3134){\makebox(0,0)[lb]{\smash{{\SetFigFont{34}{40.8}{\rmdefault}{\mddefault}{\updefault}{\color[rgb]{0,0,0}+}%
}}}}
\put(  1,1964){\makebox(0,0)[b]{\smash{{\SetFigFont{29}{34.8}{\rmdefault}{\mddefault}{\updefault}{\color[rgb]{0,0,0}$\mathbf{z}$}%
}}}}
\put(  1,4349){\makebox(0,0)[b]{\smash{{\SetFigFont{29}{34.8}{\rmdefault}{\mddefault}{\updefault}{\color[rgb]{0,0,0}$\mathbf{v}$}%
}}}}
\put(-854,5114){\makebox(0,0)[lb]{\smash{{\SetFigFont{29}{34.8}{\rmdefault}{\mddefault}{\updefault}{\color[rgb]{0,0,0}Channel}%
}}}}
\put(-3734,3179){\makebox(0,0)[lb]{\smash{{\SetFigFont{29}{34.8}{\rmdefault}{\mddefault}{\updefault}{\color[rgb]{0,0,0}Encoder}%
}}}}
\put(-4319,4484){\makebox(0,0)[lb]{\smash{{\SetFigFont{29}{34.8}{\rmdefault}{\mddefault}{\updefault}{\color[rgb]{0,0,0}User node $\mathtt{A}$}%
}}}}
\put(-5264,3539){\makebox(0,0)[b]{\smash{{\SetFigFont{29}{34.8}{\rmdefault}{\mddefault}{\updefault}{\color[rgb]{0,0,0}$X\in \mathbb{G}^{(d)}$}%
}}}}
\put(3241,3404){\makebox(0,0)[lb]{\smash{{\SetFigFont{29}{34.8}{\rmdefault}{\mddefault}{\updefault}{\color[rgb]{0,0,0}$\mathbf{w'}$}%
}}}}
\put(4231,3449){\makebox(0,0)[lb]{\smash{{\SetFigFont{25}{30.0}{\rmdefault}{\mddefault}{\updefault}{\color[rgb]{0,0,0}Coset}%
}}}}
\put(4006,2954){\makebox(0,0)[lb]{\smash{{\SetFigFont{25}{30.0}{\rmdefault}{\mddefault}{\updefault}{\color[rgb]{0,0,0}decoding}%
}}}}
\put(1261,3449){\makebox(0,0)[b]{\smash{{\SetFigFont{29}{34.8}{\rmdefault}{\mddefault}{\updefault}{\color[rgb]{0,0,0}$\mathbf{w}$}%
}}}}
\put(-1709,3449){\makebox(0,0)[lb]{\smash{{\SetFigFont{29}{34.8}{\rmdefault}{\mddefault}{\updefault}{\color[rgb]{0,0,0}$\mathbf{u}$}%
}}}}
\put(5716,3539){\makebox(0,0)[lb]{\smash{{\SetFigFont{29}{34.8}{\rmdefault}{\mddefault}{\updefault}{\color[rgb]{0,0,0}Estimate}%
}}}}
\put(5716,2729){\makebox(0,0)[lb]{\smash{{\SetFigFont{29}{34.8}{\rmdefault}{\mddefault}{\updefault}{\color[rgb]{0,0,0}of $X\oplus Y$}%
}}}}
\put(1846,3179){\makebox(0,0)[lb]{\smash{{\SetFigFont{29}{34.8}{\rmdefault}{\mddefault}{\updefault}{\color[rgb]{0,0,0}$Q_{\Lambda}(\cdot)$}%
}}}}
\end{picture}%

%% file: cosetrep_new.pdf_t
\begin{picture}(0,0)%
\includegraphics{cosetrep_new.pdf}%
\end{picture}%
\setlength{\unitlength}{4144sp}%
\begingroup\makeatletter\ifx\SetFigFont\undefined%
\gdef\SetFigFont#1#2#3#4#5{%
  \reset@font\fontsize{#1}{#2pt}%
  \fontfamily{#3}\fontseries{#4}\fontshape{#5}%
  \selectfont}%
\fi\endgroup%
\begin{picture}(10059,9301)(-5321,-3950)
\put(136,524){\makebox(0,0)[lb]{\smash{{\SetFigFont{20}{24.0}{\rmdefault}{\mddefault}{\updefault}{\color[rgb]{0,0,0}$\mathbf{0}$}%
}}}}
\put(-2924,-2491){\makebox(0,0)[lb]{\smash{{\SetFigFont{29}{34.8}{\rmdefault}{\mddefault}{\updefault}{\color[rgb]{0,0,0}$: \Lambda_{0}$}%
}}}}
\put(1441,-2536){\makebox(0,0)[lb]{\smash{{\SetFigFont{29}{34.8}{\rmdefault}{\mddefault}{\updefault}{\color[rgb]{0,0,0}$: \Lambda_{1}$}%
}}}}
\put(-2924,-3211){\makebox(0,0)[lb]{\smash{{\SetFigFont{29}{34.8}{\rmdefault}{\mddefault}{\updefault}{\color[rgb]{0,0,0}$: \Lambda_{2}$}%
}}}}
\put(1441,-3211){\makebox(0,0)[lb]{\smash{{\SetFigFont{29}{34.8}{\rmdefault}{\mddefault}{\updefault}{\color[rgb]{0,0,0}$: \Lambda_{3}$}%
}}}}
\put(181,-3886){\makebox(0,0)[lb]{\smash{{\SetFigFont{29}{34.8}{\rmdefault}{\mddefault}{\updefault}{\color[rgb]{0,0,0}$: \Lambda_{4}$}%
}}}}
\put(-629,1289){\makebox(0,0)[lb]{\smash{{\SetFigFont{20}{24.0}{\rmdefault}{\mddefault}{\updefault}{\color[rgb]{0,0,0}$\lambda_{1}$}%
}}}}
\put(676,1469){\makebox(0,0)[lb]{\smash{{\SetFigFont{20}{24.0}{\rmdefault}{\mddefault}{\updefault}{\color[rgb]{0,0,0}$\lambda_{2}$}%
}}}}
\put( 46,929){\makebox(0,0)[lb]{\smash{{\SetFigFont{20}{24.0}{\rmdefault}{\mddefault}{\updefault}{\color[rgb]{0,0,0}$\lambda_{0}$}%
}}}}
\put(-809,254){\makebox(0,0)[lb]{\smash{{\SetFigFont{20}{24.0}{\rmdefault}{\mddefault}{\updefault}{\color[rgb]{0,0,0}$\lambda_{4}$}%
}}}}
\put(541,-16){\makebox(0,0)[lb]{\smash{{\SetFigFont{20}{24.0}{\rmdefault}{\mddefault}{\updefault}{\color[rgb]{0,0,0}$\lambda_{3}$}%
}}}}
\end{picture}%

%% file: channelrep1.pdf_t
\begin{picture}(0,0)%
\includegraphics{channelrep1.pdf}%
\end{picture}%
\setlength{\unitlength}{4144sp}%
\begingroup\makeatletter\ifx\SetFigFont\undefined%
\gdef\SetFigFont#1#2#3#4#5{%
  \reset@font\fontsize{#1}{#2pt}%
  \fontfamily{#3}\fontseries{#4}\fontshape{#5}%
  \selectfont}%
\fi\endgroup%
\begin{picture}(14444,12182)(-4971,-5021)
\put(-134,3179){\makebox(0,0)[lb]{\smash{{\SetFigFont{29}{34.8}{\rmdefault}{\mddefault}{\updefault}{\color[rgb]{0,0,0}+}%
}}}}
\put(-4724,5699){\makebox(0,0)[b]{\smash{{\SetFigFont{34}{40.8}{\rmdefault}{\mddefault}{\updefault}{\color[rgb]{0,0,0}$X$}%
}}}}
\put(856,3404){\makebox(0,0)[b]{\smash{{\SetFigFont{34}{40.8}{\rmdefault}{\mddefault}{\updefault}{\color[rgb]{0,0,0}$\mathbf{w}$}%
}}}}
\put(2791,3134){\makebox(0,0)[b]{\smash{{\SetFigFont{34}{40.8}{\rmdefault}{\mddefault}{\updefault}{\color[rgb]{0,0,0}$[\cdot]\bmod{\Lambda_{0}}$}%
}}}}
\put(-4679,749){\makebox(0,0)[b]{\smash{{\SetFigFont{34}{40.8}{\rmdefault}{\mddefault}{\updefault}{\color[rgb]{0,0,0}$Y$}%
}}}}
\put(-1439,119){\makebox(0,0)[lb]{\smash{{\SetFigFont{34}{40.8}{\rmdefault}{\mddefault}{\updefault}{\color[rgb]{0,0,0}$\mathbf{v}$}%
}}}}
\put(-2924,5384){\makebox(0,0)[b]{\smash{{\SetFigFont{34}{40.8}{\rmdefault}{\mddefault}{\updefault}{\color[rgb]{0,0,0}Encoder}%
}}}}
\put(-2879,434){\makebox(0,0)[b]{\smash{{\SetFigFont{34}{40.8}{\rmdefault}{\mddefault}{\updefault}{\color[rgb]{0,0,0}Encoder}%
}}}}
\put(-4049,6509){\makebox(0,0)[lb]{\smash{{\SetFigFont{34}{40.8}{\rmdefault}{\mddefault}{\updefault}{\color[rgb]{0,0,0}User node $\mathtt{A}$}%
}}}}
\put(-3959,-691){\makebox(0,0)[lb]{\smash{{\SetFigFont{34}{40.8}{\rmdefault}{\mddefault}{\updefault}{\color[rgb]{0,0,0}User node $\mathtt{B}$}%
}}}}
\put(136,5024){\makebox(0,0)[lb]{\smash{{\SetFigFont{34}{40.8}{\rmdefault}{\mddefault}{\updefault}{\color[rgb]{0,0,0}Channel}%
}}}}
\put(2656,4619){\makebox(0,0)[lb]{\smash{{\SetFigFont{34}{40.8}{\rmdefault}{\mddefault}{\updefault}{\color[rgb]{0,0,0}Relay $\mathtt{R}$}%
}}}}
\put(5041,3134){\makebox(0,0)[lb]{\smash{{\SetFigFont{34}{40.8}{\rmdefault}{\mddefault}{\updefault}{\color[rgb]{0,0,0}$[Q_{\Lambda}(\cdot)]\bmod \Lambda_{0}$}%
}}}}
\put(8686,3494){\makebox(0,0)[lb]{\smash{{\SetFigFont{34}{40.8}{\rmdefault}{\mddefault}{\updefault}{\color[rgb]{0,0,0}$\hat{\mathbf{w}}$}%
}}}}
\put(4186,3449){\makebox(0,0)[lb]{\smash{{\SetFigFont{34}{40.8}{\rmdefault}{\mddefault}{\updefault}{\color[rgb]{0,0,0}$\tilde{\mathbf{w}}$}%
}}}}
\put(-2204,3224){\makebox(0,0)[b]{\smash{{\SetFigFont{34}{40.8}{\rmdefault}{\mddefault}{\updefault}{\color[rgb]{0,0,0}$\mathbf{z}$}%
}}}}
\put(-1529,5699){\makebox(0,0)[lb]{\smash{{\SetFigFont{34}{40.8}{\rmdefault}{\mddefault}{\updefault}{\color[rgb]{0,0,0}$\mathbf{u}$}%
}}}}
\put(-1574,-2266){\makebox(0,0)[b]{\smash{{\SetFigFont{34}{40.8}{\rmdefault}{\mddefault}{\updefault}{\color[rgb]{0,0,0}$\mathbf{z}$}%
}}}}
\put(-2834,-4876){\makebox(0,0)[lb]{\smash{{\SetFigFont{29}{34.8}{\rmdefault}{\mddefault}{\updefault}{\color[rgb]{0,0,1}Equivalent MLAN  channel}%
}}}}
\put(7381,-3121){\makebox(0,0)[b]{\smash{{\SetFigFont{34}{40.8}{\rmdefault}{\mddefault}{\updefault}{\color[rgb]{0,0,0}$\hat{\mathbf{w}}$}%
}}}}
\put(1171,-3346){\makebox(0,0)[b]{\smash{{\SetFigFont{34}{40.8}{\rmdefault}{\mddefault}{\updefault}{\color[rgb]{0,0,0}$[\cdot]\bmod{\Lambda_{0}}$}%
}}}}
\put(3511,-3346){\makebox(0,0)[lb]{\smash{{\SetFigFont{34}{40.8}{\rmdefault}{\mddefault}{\updefault}{\color[rgb]{0,0,0}$[Q_{\Lambda}(\cdot)]\bmod \Lambda_{0}$}%
}}}}
\put(-4454,-3076){\makebox(0,0)[b]{\smash{{\SetFigFont{34}{40.8}{\rmdefault}{\mddefault}{\updefault}{\color[rgb]{0,0,0}$[\mathbf{x}+\mathbf{y}]\bmod \Lambda_{0}$}%
}}}}
\put(2656,-3121){\makebox(0,0)[lb]{\smash{{\SetFigFont{34}{40.8}{\rmdefault}{\mddefault}{\updefault}{\color[rgb]{0,0,0}$\tilde{\mathbf{w}}$}%
}}}}
\put(-1754,-3389){\makebox(0,0)[lb]{\smash{{\SetFigFont{34}{40.8}{\rmdefault}{\mddefault}{\updefault}{\color[rgb]{0,0,0}+}%
}}}}
\end{picture}%

%% file: multihop1.pdf_t
\begin{picture}(0,0)%
\includegraphics{multihop1.pdf}%
\end{picture}%
\setlength{\unitlength}{4144sp}%
\begingroup\makeatletter\ifx\SetFigFont\undefined%
\gdef\SetFigFont#1#2#3#4#5{%
  \reset@font\fontsize{#1}{#2pt}%
  \fontfamily{#3}\fontseries{#4}\fontshape{#5}%
  \selectfont}%
\fi\endgroup%
\begin{picture}(9034,932)(2684,-4127)
\put(3151,-3796){\makebox(0,0)[b]{\smash{{\SetFigFont{29}{34.8}{\rmdefault}{\mddefault}{\updefault}{\color[rgb]{0,0,0}$\mathtt{S}$}%
}}}}
\put(6751,-3796){\makebox(0,0)[b]{\smash{{\SetFigFont{29}{34.8}{\rmdefault}{\mddefault}{\updefault}{\color[rgb]{0,0,0}$\mathtt{R}_2$}%
}}}}
\put(4951,-3796){\makebox(0,0)[b]{\smash{{\SetFigFont{29}{34.8}{\rmdefault}{\mddefault}{\updefault}{\color[rgb]{0,0,0}$\mathtt{R}_1$}%
}}}}
\put(9451,-3796){\makebox(0,0)[b]{\smash{{\SetFigFont{29}{34.8}{\rmdefault}{\mddefault}{\updefault}{\color[rgb]{0,0,0}$\mathtt{R}_K$}%
}}}}
\put(11251,-3796){\makebox(0,0)[b]{\smash{{\SetFigFont{29}{34.8}{\rmdefault}{\mddefault}{\updefault}{\color[rgb]{0,0,0}$\mathtt{D}$}%
}}}}
\end{picture}%

%% file: multihop_3hop_2mess.pdf_t
\begin{picture}(0,0)%
\includegraphics{multihop_3hop_2mess.pdf}%
\end{picture}%
\setlength{\unitlength}{4144sp}%
\begingroup\makeatletter\ifx\SetFigFont\undefined%
\gdef\SetFigFont#1#2#3#4#5{%
  \reset@font\fontsize{#1}{#2pt}%
  \fontfamily{#3}\fontseries{#4}\fontshape{#5}%
  \selectfont}%
\fi\endgroup%
\begin{picture}(11595,9932)(-329,-15375)
\put(2701,-6068){\makebox(0,0)[b]{\smash{{\SetFigFont{29}{34.8}{\rmdefault}{\mddefault}{\updefault}{\color[rgb]{0,0,0}$\mathtt{S}$}%
}}}}
\put(2701,-7869){\makebox(0,0)[b]{\smash{{\SetFigFont{29}{34.8}{\rmdefault}{\mddefault}{\updefault}{\color[rgb]{0,0,0}$\mathtt{S}$}%
}}}}
\put(2701,-9669){\makebox(0,0)[b]{\smash{{\SetFigFont{29}{34.8}{\rmdefault}{\mddefault}{\updefault}{\color[rgb]{0,0,0}$\mathtt{S}$}%
}}}}
\put(2701,-11469){\makebox(0,0)[b]{\smash{{\SetFigFont{29}{34.8}{\rmdefault}{\mddefault}{\updefault}{\color[rgb]{0,0,0}$\mathtt{S}$}%
}}}}
\put(2701,-13291){\makebox(0,0)[b]{\smash{{\SetFigFont{29}{34.8}{\rmdefault}{\mddefault}{\updefault}{\color[rgb]{0,0,0}$\mathtt{S}$}%
}}}}
\put(2701,-15091){\makebox(0,0)[b]{\smash{{\SetFigFont{29}{34.8}{\rmdefault}{\mddefault}{\updefault}{\color[rgb]{0,0,0}$\mathtt{S}$}%
}}}}
\put(10801,-6068){\makebox(0,0)[b]{\smash{{\SetFigFont{29}{34.8}{\rmdefault}{\mddefault}{\updefault}{\color[rgb]{0,0,0}$\mathtt{D}$}%
}}}}
\put(8101,-6068){\makebox(0,0)[b]{\smash{{\SetFigFont{29}{34.8}{\rmdefault}{\mddefault}{\updefault}{\color[rgb]{0,0,0}$\mathtt{R}_2$}%
}}}}
\put(5401,-6068){\makebox(0,0)[b]{\smash{{\SetFigFont{29}{34.8}{\rmdefault}{\mddefault}{\updefault}{\color[rgb]{0,0,0}$\mathtt{R}_1$}%
}}}}
\put(5401,-7869){\makebox(0,0)[b]{\smash{{\SetFigFont{29}{34.8}{\rmdefault}{\mddefault}{\updefault}{\color[rgb]{0,0,0}$\mathtt{R}_1$}%
}}}}
\put(8101,-7869){\makebox(0,0)[b]{\smash{{\SetFigFont{29}{34.8}{\rmdefault}{\mddefault}{\updefault}{\color[rgb]{0,0,0}$\mathtt{R}_2$}%
}}}}
\put(10801,-7869){\makebox(0,0)[b]{\smash{{\SetFigFont{29}{34.8}{\rmdefault}{\mddefault}{\updefault}{\color[rgb]{0,0,0}$\mathtt{D}$}%
}}}}
\put(10801,-9669){\makebox(0,0)[b]{\smash{{\SetFigFont{29}{34.8}{\rmdefault}{\mddefault}{\updefault}{\color[rgb]{0,0,0}$\mathtt{D}$}%
}}}}
\put(8101,-9669){\makebox(0,0)[b]{\smash{{\SetFigFont{29}{34.8}{\rmdefault}{\mddefault}{\updefault}{\color[rgb]{0,0,0}$\mathtt{R}_2$}%
}}}}
\put(5401,-9669){\makebox(0,0)[b]{\smash{{\SetFigFont{29}{34.8}{\rmdefault}{\mddefault}{\updefault}{\color[rgb]{0,0,0}$\mathtt{R}_1$}%
}}}}
\put(10801,-11469){\makebox(0,0)[b]{\smash{{\SetFigFont{29}{34.8}{\rmdefault}{\mddefault}{\updefault}{\color[rgb]{0,0,0}$\mathtt{D}$}%
}}}}
\put(8101,-11469){\makebox(0,0)[b]{\smash{{\SetFigFont{29}{34.8}{\rmdefault}{\mddefault}{\updefault}{\color[rgb]{0,0,0}$\mathtt{R}_2$}%
}}}}
\put(5401,-11469){\makebox(0,0)[b]{\smash{{\SetFigFont{29}{34.8}{\rmdefault}{\mddefault}{\updefault}{\color[rgb]{0,0,0}$\mathtt{R}_1$}%
}}}}
\put(5401,-13291){\makebox(0,0)[b]{\smash{{\SetFigFont{29}{34.8}{\rmdefault}{\mddefault}{\updefault}{\color[rgb]{0,0,0}$\mathtt{R}_1$}%
}}}}
\put(5401,-15091){\makebox(0,0)[b]{\smash{{\SetFigFont{29}{34.8}{\rmdefault}{\mddefault}{\updefault}{\color[rgb]{0,0,0}$\mathtt{R}_1$}%
}}}}
\put(8101,-13291){\makebox(0,0)[b]{\smash{{\SetFigFont{29}{34.8}{\rmdefault}{\mddefault}{\updefault}{\color[rgb]{0,0,0}$\mathtt{R}_2$}%
}}}}
\put(8101,-15091){\makebox(0,0)[b]{\smash{{\SetFigFont{29}{34.8}{\rmdefault}{\mddefault}{\updefault}{\color[rgb]{0,0,0}$\mathtt{R}_2$}%
}}}}
\put(10801,-13291){\makebox(0,0)[b]{\smash{{\SetFigFont{29}{34.8}{\rmdefault}{\mddefault}{\updefault}{\color[rgb]{0,0,0}$\mathtt{D}$}%
}}}}
\put(10801,-15091){\makebox(0,0)[b]{\smash{{\SetFigFont{29}{34.8}{\rmdefault}{\mddefault}{\updefault}{\color[rgb]{0,0,0}$\mathtt{D}$}%
}}}}
\put(4051,-5686){\makebox(0,0)[b]{\smash{{\SetFigFont{20}{24.0}{\rmdefault}{\mddefault}{\updefault}{\color[rgb]{0,0,0}$J_1$}%
}}}}
\put(4051,-7486){\makebox(0,0)[b]{\smash{{\SetFigFont{20}{24.0}{\rmdefault}{\mddefault}{\updefault}{\color[rgb]{0,0,0}$X_1\oplus J_1$}%
}}}}
\put(4051,-9286){\makebox(0,0)[b]{\smash{{\SetFigFont{20}{24.0}{\rmdefault}{\mddefault}{\updefault}{\color[rgb]{0,0,0}$X_1\oplus J_2$}%
}}}}
\put(4051,-11086){\makebox(0,0)[b]{\smash{{\SetFigFont{20}{24.0}{\rmdefault}{\mddefault}{\updefault}{\color[rgb]{0,0,0}$X_1\oplus X_2\oplus J_2$}%
}}}}
\put(4051,-12886){\makebox(0,0)[b]{\smash{{\SetFigFont{20}{24.0}{\rmdefault}{\mddefault}{\updefault}{\color[rgb]{0,0,0}$X_1\oplus X_2\oplus J_3$}%
}}}}
\put(6751,-5686){\makebox(0,0)[b]{\smash{{\SetFigFont{20}{24.0}{\rmdefault}{\mddefault}{\updefault}{\color[rgb]{0,0,0}$J_1$}%
}}}}
\put(6751,-7486){\makebox(0,0)[b]{\smash{{\SetFigFont{20}{24.0}{\rmdefault}{\mddefault}{\updefault}{\color[rgb]{0,0,0}$J_2$}%
}}}}
\put(6751,-9286){\makebox(0,0)[b]{\smash{{\SetFigFont{20}{24.0}{\rmdefault}{\mddefault}{\updefault}{\color[rgb]{0,0,0}$X_1\oplus J_2$}%
}}}}
\put(6751,-11086){\makebox(0,0)[b]{\smash{{\SetFigFont{20}{24.0}{\rmdefault}{\mddefault}{\updefault}{\color[rgb]{0,0,0}$X_1\oplus J_3$}%
}}}}
\put(6751,-12886){\makebox(0,0)[b]{\smash{{\SetFigFont{20}{24.0}{\rmdefault}{\mddefault}{\updefault}{\color[rgb]{0,0,0}$X_1\oplus X_2\oplus J_3$}%
}}}}
\put(6751,-14686){\makebox(0,0)[b]{\smash{{\SetFigFont{20}{24.0}{\rmdefault}{\mddefault}{\updefault}{\color[rgb]{0,0,0}$X_1\oplus X_2\oplus J_4$}%
}}}}
\put(9451,-14686){\makebox(0,0)[b]{\smash{{\SetFigFont{20}{24.0}{\rmdefault}{\mddefault}{\updefault}{\color[rgb]{0,0,0}$X_1\oplus X_2\oplus J_4$}%
}}}}
\put(9451,-12886){\makebox(0,0)[b]{\smash{{\SetFigFont{20}{24.0}{\rmdefault}{\mddefault}{\updefault}{\color[rgb]{0,0,0}$X_1\oplus J_4$}%
}}}}
\put(9451,-11086){\makebox(0,0)[b]{\smash{{\SetFigFont{20}{24.0}{\rmdefault}{\mddefault}{\updefault}{\color[rgb]{0,0,0}$X_1\oplus J_3$}%
}}}}
\put(9451,-9286){\makebox(0,0)[b]{\smash{{\SetFigFont{20}{24.0}{\rmdefault}{\mddefault}{\updefault}{\color[rgb]{0,0,0}$J_3$}%
}}}}
\put(9451,-7486){\makebox(0,0)[b]{\smash{{\SetFigFont{20}{24.0}{\rmdefault}{\mddefault}{\updefault}{\color[rgb]{0,0,0}$J_2$}%
}}}}
\end{picture}%

%% file: secure_2way_relay_IT_R1final.bbl
\begin{thebibliography}{99}

\bibitem{Agrawal09}
S.~Agrawal and S.~Vishwanath, ``On the secrecy rate of interference networks using structured codes,'' \emph{Proc.\ 2009 IEEE Int. Symp. Information Theory}, Seoul, Korea, pp.\ 2091--2095.



\bibitem{Baik}
I.-J. Baik and S.-Y. Chung, ``Network coding for two-way relay channels using
  lattices,'' in \emph{Proc.\ IEEE Int.\ Conf.\ Communications}, Beijing, China, 2008, pp.\ 3898--3902.

\bibitem{Barvinok}
A.~Barvinok, \emph{Math 669: Combinatorics, Geometry and Complexity of Integer Points}.
 [Online]. Available: http://www.math.lsa.umich.edu/$\sim$barvinok/latticenotes669.pdf .


\bibitem{Belfiore10}
J.-C.\ Belfiore and F.~Oggier, ``Secrecy gain: A wiretap lattice code design,'' in \emph{Proc. 2010 Int.\ Symp.\ Information Theory and  Its Applications}, Taichung, Taiwan, pp.\ 174--178.

\bibitem{Belfiore11}
J.-C.\ Belfiore, ``Lattice codes for the compute-and-forward protocol: The flatness factor,'' 
in \emph{Proc.\ 2011 Information Theory Workshop}, Paraty, Brazil, pp.\ 1--4.


\bibitem{Conway}
J.H.\ Conway and N.J.\ Sloane, \emph{Sphere Packings, Lattices and Groups}, New York: Springer-Verlag, 1988.

\bibitem{Cover}
T.~Cover and J.~Thomas, \emph{Elements of Information Theory}, 2nd ed. New York: Wiley-Interscience, 1996.

\bibitem{Cramer}
R.\ Cramer, Y.\ Dodis, S.\ Fehr, C.\ Padro, and D.\ Wichs, ``Detection of algebraic manipulation with applications to robust secret sharing and fuzzy extractors,'' \emph{Adv. Cryptology}, vol.\ 4965, pp.\ 471--488, 2008.

\bibitem{CN04}
I.~Csisz{\'a}r and P.~Narayan, ``Secrecy capacities for multiple terminals,''
\emph{IEEE Trans.\ Inf.\ Theory}, vol.\ 50, no.\ 12, pp.\ 3047--3061, Dec.\ 2004.

\bibitem{diPietro13}
N.\ di Pietro, G.~Z\'{e}mor, and J.\ J.\ Boutros, 
``New results on Construction A lattices based on very sparse parity-check matrices,'' in \emph{Proc.\ 2013 IEEE Int. Symp. Information Theory}, Istanbul, Turkey, pp.\ 1675--1679.

\bibitem{Dasgupta}
A.\ Dasgupta, 
\emph{Probability for Statistics and Machine Learning}, New York: Springer Texts in Statistics, 2011.

\bibitem{EGR04} W.\ Ehm, T.\ Gneiting, and D.\ Richards, ``Convolution
   roots of radial positive definite functions with compact support,''
   \emph{Trans.\ AMS}, vol.\ 356, no.\ 11, pp.\ 4655--4685, May 2004.

 
 \bibitem{Erez04}
U.~Erez and R.~Zamir, ``Achieving 1/2log(1+SNR) on the AWGN channel with lattice encoding and decoding,''
\emph{IEEE Trans.\ Inf.\ Theory}, vol.\ 50, no.\ 10, pp.\ 2293--2314, Oct.\ 2004.

\bibitem{Erez05}
U.~Erez, S.\ Litsyn, and R.\ Zamir, ``Lattices which are good for (almost) everything,'' 
\emph{IEEE Trans.\ Inf.\ Theory},  vol.\ 51, no.\ 10, pp.\ 3401--3416, Oct.\ 2005.

\bibitem{tenBrink05}
U.~Erez, S.\ ten Brink, ``A close-to-capacity dirty paper coding scheme,''
\emph{IEEE Trans.\ Inf.\ Theory},  vol.\ 51, no.\ 10, pp.\ 3417--3432, Oct.\ 2005.


 \bibitem{BESSEL1} A.\ Elbert and A.\ Laforgia, ``An asymptotic
   relation for the zeros of Bessel functions,''  \emph{J. 
     Math. Analysis and Applications}, vol.\ 98, no.\ 2, pp.\ 502--510, 1984.

\bibitem{Feller} W.\ Feller, \emph{An Introduction to Probability Theory
and Its Applications, Vol.\ 2}, 2nd ed. New York: Wiley, 1971.

\bibitem{HeYener}
X.~He and A.~Yener, ``Providing secrecy with lattice codes,'' \emph{Proc. 46th Annual Allerton Conf. on Communication, Control, and Computing}, 
Monticello, IL, 2008, pp.\ 1199--1206.

\bibitem{HeYenerstrong}
X.~He and A.~Yener, ``Strong secrecy and reliable Byzantine detection in the presence of an untrusted relay,'' 
\emph{IEEE Trans.\ Inf.\ Theory}, vol.\ 59, no.\ 1, pp.\ 177--192, Jan.\ 2013.

\bibitem{Herstein} I.N.\ Herstein, \emph{Topics in Algebra}, 2nd ed. New York: Wiley, 1975.

\bibitem{NSAconf}
  N.~Kashyap, V.~Shashank, and A.~Thangaraj, ``Secure computation in a bidirectional relay,'' in \emph{Proc.\ 2012 IEEE Int. Symp. Information Theory}, Cambridge, MA, pp.\ 1162--1166. 


\bibitem{Ling13} C.\ Ling, L.\ Luzzi, J.-C.\ Belfiore, and D.\ Stehl{\'e}, 
``Semantically secure lattice codes for the Gaussian wiretap channel,''
submitted for publication.
[Online]. Available: http://arxiv.org/abs/1210.6673.

\bibitem{Lukacs} E.~Lukacs, \emph{Characteristic Functions}, 2nd ed. London, U.K.: Griffin, 1970.

\bibitem{Maurer00}
U.~Maurer and S.~Wolf. ``Information-theoretic key agreement: From weak to strong
secrecy for free,''  \emph{Proc. EUROCRYPT--2000 on Advances in Cryptology}, vol.\ 1807, pp.\ 351--368,
Springer, 2000.

\bibitem{Nazer11}
B.~Nazer and M.~Gastpar, 
``Compute-and-forward: harnessing interference through structured codes,'' 
\emph{IEEE Trans.\ Inf.\ Theory}, vol.\ 57, no.\ 10, pp.\ 6463--6486, Oct.\ 2011.

\bibitem{NazerProc}
B.~Nazer and M.~Gastpar, ``Reliable physical layer network coding,''
  \emph{Proc. of the IEEE}, vol.~99, no.~3, pp.~438--460, Mar.\ 2011.

\bibitem{Nitinawarat12}
S.~Nitinawarat and P.~Narayan, ``Secret key generation for correlated Gaussian sources,'' 
\emph{IEEE Trans.\ Inf.\ Theory}, vol.\ 58, no.\ 6, pp.\ 3373--3391, Jun.\ 2012.


\bibitem{Belfiore}
F.~Oggier, P.~Sol\'{e}, and J.-C. Belfiore, ``Lattice codes for the wiretap Gaussian channel: construction and analysis,'' submitted for publication.
[Online]. Available: http://arxiv.org/abs/1103.4086.



\bibitem{Popovski}
P.~Popovski and H.~Yomo, ``Physical network coding in two-way wireless relay channels,'' in
  \emph{Proc.\ IEEE Int.\ Conf.\ Communications}, Glasgow, Scotland, 2007, pp.\ 707--712.



\bibitem{Roth}
R.M.~Roth, \emph{Introduction to Coding Theory}, Cambridge, U.K.: Cambridge University Press, 2006.

\bibitem{RS04} H.\ Rubin and T.M.\ Sellke, ``Zeroes of infinitely
  differentiable characteristic functions,'' in \emph{A Festschrift
    for Herman Rubin}, Anirban DasGupta, ed., 
  Institute of Mathematical Statistics Lecture Notes -- Monograph
  Series, vol.\ 45, pp.\ 164--170, 2004.

\bibitem{Sommer}
N.~Sommer, M.~Feder, and O.~Shalvi, ``Low density lattice codes,'' \emph{IEEE Trans.\ Inf.\ Theory}, vol.\ 54, no.\ 4, pp.\ 1561--1585, Apr.\ 2008. 

\bibitem{Sommer09}
N.~Sommer, M.~Feder, and O.~Shalvi, ``Shaping methods for low-density lattice codes,''
in \emph{Proc.\ 2009 Information Theory Workshop}, Taormina, Italy, pp.\ 238--242.

\bibitem{SW71} E.M.\ Stein and G.L.\ Weiss, 
\emph{Introduction to Fourier Analysis on Euclidean Spaces}, Princeton, NJ:
Princeton Univ.\ Press, 1971.

\bibitem{BESSEL2} F.G.\ Tricomi, ``Sulle funzioni di Bessel di ordine
   e argomento pressoch\'{e} uguali,'' \emph{Atti Accad. Sci. Torino
     Cl. Sci. Fis. Mat. Natur.}, vol.\ 83, pp.\ 3--20, 1949.

\bibitem{Wilson}
M.~Wilson, K.~Narayanan, H.~Pfister, and A.~Sprintson, ``Joint physical layer
  coding and network coding for bidirectional relaying,'' 
\emph{IEEE Trans.\ Inf.\ Theory}, vol.~56, no.~11, pp.\ 5641--5654, Nov.\ 2010.

\bibitem{Wolfe73} 
S.J.~Wolfe, ``On the finite series expansion of multivariate characteristic functions,'' 
\emph{J.~Multivariate Anal.}, vol.\ 3, pp.\ 328--335, 1973.


\bibitem{Yan13}
Y.~Yan, C.~Ling, and X.~Wu, ``Polar lattices: Where Arikan meets Forney,''  in \emph{Proc.\ 2013 IEEE Int. Symp. Information Theory}, Istanbul, Turkey, pp.\ 1292--1296.

\bibitem{Zhang}
S.~Zhang and S.-C. Liew, ``Channel coding and decoding in a relay system
  operated with physical-layer network coding,'' 
\emph{IEEE J.\ Sel.\ Areas Commun.}, vol.~27, no.~5, pp.\ 788--796, Jun.\ 2009.


\end{thebibliography}
